\documentclass[11pt]{article}
\usepackage[margin=1in]{geometry}
\usepackage{amsmath, amssymb, amsfonts, braket, complexity, enumerate, mathrsfs, tikz, caption, subcaption, xcolor, braket, dirtytalk, bm, bbm, float}
\usepackage{authblk}
\usepackage{mathtools}
\usepackage{amsthm}
\usepackage{thm-restate}
\usepackage{tabularx}

\usepackage[linesnumbered,ruled,vlined]{algorithm2e}
\SetKwInput{KwInput}{Input}
\SetKwInput{KwOutput}{Output}
\SetKwFor{RepTimes}{repeat}{times}{end}

\usepackage[pagebackref,colorlinks=true,citecolor=blue,linkcolor=magenta,bookmarks=true]{hyperref}
\usepackage{orcidlink}
\usepackage[nameinlink,capitalize]{cleveref}
\usepackage{enumitem}
\usepackage{import}

\newcommand{\wt}[1]{\widetilde{#1}}

\newcommand{\F}{\mathbb{F}}
\renewcommand{\hat}{\widehat}

\theoremstyle{plain}
\newtheorem{theorem}{Theorem}[section]
\newtheorem{corollary}[theorem]{Corollary}
\newtheorem{definition}[theorem]{Definition}
\newtheorem{proposition}[theorem]{Proposition}
\newtheorem{lemma}[theorem]{Lemma}
\newtheorem{claim}[theorem]{Claim}
\newtheorem{fact}[theorem]{Fact}

\newtheorem{remark}[theorem]{Remark}

\newtheorem{question}[theorem]{Question}

\newcommand{\eps}{\varepsilon}
\newcommand{\ri}{\mathrm{i}}
\newcommand{\rd}{\mathrm{d}}

\newcommand{\sgn}{\func{sgn}}

\renewcommand{\R}{\mathbb{R}}
\renewcommand{\C}{\mathbb{C}}
\newcommand{\N}{\mathbb{N}}
\newcommand{\majorana}{\gamma}

\newcommand{\identity}{\mathbb{I}}

\newcommand{\testproj}{\Pi_{\mathrm{test}}}

\newcommand{\symproj}{\Pi_{\mathrm{sym}}}

\newcommand{\SWAP}{\mathrm{SWAP}}
\newcommand{\CNOT}{\mathrm{CNOT}}
\newcommand{\CZ}{\mathrm{CZ}}

\newcommand{\indic}[1]{\bm{1}\{#1\}}

\newcommand{\beamsplitter}{U_{\mathrm{BS}}}
\newcommand{\coherentstateset}{\mathrm{Coh}}

\newcommand{\fubiniStudyDistance}{\mathsf{dist}_{\mathrm{FS}}}
\newcommand{\fubiniStudyMetric}{g_{\mathrm{FS}}}

\newcommand{\fubiniStudyLineElement}{ds_{\mathrm{FS}}}

\DeclareMathOperator{\argmax}{arg\,max}

\newcommand{\fermionicgaussianset}{\mathcal{G}_{\mathrm{f}}}
\newcommand{\slaterdeterminantset}{\mathcal{S}}

\newcommand{\bosonicgaussianset}{\mathcal{G}_{\mathrm{b}}}
\newcommand{\zeromeanbosonicgaussianset}{\mathcal{G}_{\mathrm{b},0}}

\newcommand{\zeromeanbosonictestproj}{\Pi_{\mathrm{test},\zeromeanbosonicgaussianset}}
\newcommand{\bosonictestproj}{\Pi_{\mathrm{test},\bosonicgaussianset}}
\newcommand{\coherenttestproj}{\Pi_{\mathrm{test},\coherentstateset}}
\newcommand{\fermionicgaussiantestproj}{\Pi_{\mathrm{test}, \fermionicgaussianset}}
\newcommand{\slatertestproj}{\Pi_{\mathrm{test}, \slaterdeterminantset}}

\newcommand{\accept}{\textup{\texttt{ACCEPT}} }
\newcommand{\reject}{\textup{\texttt{REJECT}} }

\newcommand{\dist}{\func{dist}}

\renewcommand{\Pr}{\mathop{\bf Pr\/}}
\newcommand{\Ex}{\mathop{\bf E\/}}

\renewcommand\Re{\mathop{\mathrm{Re}}}

\newcommand{\linspan}{\mathop{\mathrm{span}}}
\renewcommand{\ker}{\mathop{\mathrm{ker}}}
\newcommand{\spec}{\mathrm{spec}}

\newcommand{\tr}{\mathrm{tr}}

\DeclareMathOperator{\sech}{sech}
\DeclareMathOperator{\rank}{rank}

\newcommand{\calA}{\mathcal{A}}

\newcommand{\calC}{\mathcal{C}}

\newcommand{\calH}{\mathcal{H}}

\newcommand{\calK}{\mathcal{K}}

\newcommand{\calP}{\mathcal{P}}

\newcommand{\calV}{\mathcal{V}}

\newcommand{\calY}{\mathcal{Y}}

\newcommand{\ketbra}[2]{\ket{#1}\!\!\bra{#2}}
\renewcommand{\ket}[1]{\left|#1\right\rangle}

\DeclareMathOperator{\genlingrp}{GL}
\DeclareMathOperator{\spllingrp}{SL}
\DeclareMathOperator{\orthgrp}{O}
\DeclareMathOperator{\splorthgrp}{SO}
\DeclareMathOperator{\unitgrp}{U}

\DeclareMathOperator{\symplgrp}{Sp}
\DeclareMathOperator{\spingrp}{Spin}

\DeclareMathOperator{\spllinalg}{\mathfrak{sl}}
\DeclareMathOperator{\genlinalg}{\mathfrak{gl}}
\DeclareMathOperator{\splunitalg}{\mathfrak{su}}

\newcommand{\diag}{\mathrm{diag}}

\DeclareMathOperator{\Jacobi}{J}

\DeclareMathOperator{\Span}{Span}
\DeclareMathOperator{\Sym}{Sym}
\DeclareMathOperator{\Mp}{Mp}

\DeclareMathOperator{\Heisenberg}{H}

\DeclareMathOperator{\Orth}{O}

\DeclareMathOperator{\Stab}{Stab}
\newcommand{\transpose}{\mathrm{T}}

\DeclarePairedDelimiter{\braces}{\{}{\}}
\DeclarePairedDelimiter{\bracket}{[}{]}
\DeclarePairedDelimiter{\abs}{\lvert}{\rvert}
\DeclarePairedDelimiter{\norm}{\lVert}{\rVert}
\DeclarePairedDelimiter{\paren}{(}{)}
\DeclarePairedDelimiterX{\bracebar}[2]{\{}{\}}{#1 \,\delimsize\vert\, #2}

\newcommand{\ReynoldsOp}{\mathcal R}

\newcommand{\lojasiewicz}{{\L}ojasiewicz}

\title{Optimal testing of fermionic and bosonic Gaussian states}

\date{\today}

\author[1,2,3]{Vishnu Iyer\,\orcidlink{0000-0001-8072-1390}}
\author[2]{Bryan O'Gorman\,\orcidlink{0000-0001-5164-8083}}
\author[4]{Ojas Parekh\,\orcidlink{0000-0003-2689-9264}}
\author[4]{Kevin Thompson\,\orcidlink{0000-0001-5669-2200}}
\author[5]{Andrew Zhao\,\orcidlink{0000-0002-0299-0277}}

\affil[1]{University of Texas at Austin, Austin, TX}
\affil[2]{IBM Research, Cambridge, MA}
\affil[3]{QuICS, University of Maryland, College Park, MD}
\affil[4]{Sandia National Laboratories, Albuquerque, NM}
\affil[5]{Sandia National Laboratories, Livermore, CA}

\usepackage{todonotes}
\usepackage{comment}

\numberwithin{equation}{section}

\begin{document}

\maketitle

\begin{abstract}
We study the problem of property testing quantum states: given a pure state $|\psi\rangle$ that either (A) belongs to some class $\mathcal{C}$ of pure states, or (B) is $\varepsilon$-far in trace distance from all states in $\mathcal{C}$, determine which is the case with high probability. It is known that there exists classes $\mathcal{C}$ for which the sample complexity of property testing necessarily scales with the size of the system. In this work, we identify classes for which testing only requires a size-independent number of samples. We prove a tight sample complexity of $\Theta(1/\varepsilon^2)$ for the property testing of five important classes of Gaussian quantum states:
\begin{enumerate}
    \item Fermionic Gaussian states
    \item Slater determinants
    \item Bosonic Gaussian states
    \item Zero-mean bosonic Gaussian states
    \item Bosonic coherent states
\end{enumerate}
While property testers for some of these classes have been studied previously, those analyses only gave upper bounds scaling polynomially in the number of modes. In each case, our tester is the projection onto the ``top'' irrep of two or three copies. We get to constant sample complexity by lower bounding how quickly the rejection probability increases with distance from the class; the main technical component involves bounding the spectral gap of symmetric extension of the test projector. We then establish optimality by proving matching lower bounds based on binary state discrimination for each class. Our testers and optimality claims also extend to the tolerant setting, wherein $|\psi\rangle$ may be $O(\varepsilon)$-close to $\mathcal{C}$ in case (A), rather than lying exactly inside.
\end{abstract}

\clearpage

\tableofcontents

\section{Introduction}

Quantum property testing is the task of deciding whether an input state possesses some structural property of interest~\cite{montanaro2013survey}. More precisely, given a class of states $\calC$ and copies of an input state $\rho$, the goal is to distinguish, with high probability, the case in which $\rho \in \calC$ from the case in which $\rho$ is $\eps$-far from every state in $\calC$ (say, in trace distance). Of course, one desires to accomplish this as efficiently as possible, especially with respect to the number of copies of $\rho$ consumed. This problem arises naturally in the context of learning theory~\cite{mele2025efficient}, verification~\cite{aolita2015reliable}, sensing~\cite{apellaniz2017optimal}, and more \cite{montanaro2016survey}.
Beyond such applications, property testing also probes the very nature of the structure of quantum states: how much quantum data is required to certify the presence of a global property or symmetry within a state of very large dimension?

In the absence of additional structure, there is no reason to expect the copy complexity of property testing to be independent of the dimension \cite{odonnell2015quantum,buadescu2019quantum}. This is not surprising, as arbitrary classes can encode essentially all the information of their constituent states without any unifying structure. As such, deciding membership in the class can be as expensive as full state tomography itself, which scales exponentially in system size \cite{odonnell2016efficient,haah2016sample}. Nonetheless, for sufficiently structured families, property testing can be dramatically easier than learning the input state \cite{harrow2013testing,montanaro2016survey,gross2021schur}.
This allows for the possibility that checking membership in the class can be performed with a number of copies that is not just reasonable, but completely independent of the size of the input state.

Only a handful of families are known to exhibit this phenomenon. Product states can be tested via a two-copy protocol due to Harrow and Montanaro \cite{harrow2013testing}, and stabilizer states can be tested via a six-copy protocol due to Gross, Nezami, and Walter \cite{gross2021schur}. So far, these are the only nontrivial classes of states that have been demonstrated to possess property testers with constant sample complexity. These state classes also exemplify a remarkable balance: although they possess enough structure to be tested very efficiently, product states and (especially) stabilizer states retain enough expressiveness to drive a myriad of algorithmic, cryptographic, and information-theoretic applications \cite{hein2006entanglement}. This motivates the central question behind this paper:

\begin{question}
    What other interesting classes of quantum states can be tested with $O(1)$ sample complexity?
\end{question}

To explore this question, we will focus on fermionic and bosonic systems. In particular, we study the property testing problem for four important classes of states:
 
\begin{itemize}
    \item \textbf{Slater determinants} are fermionic analogues of product states \cite{slater1929theory}.  In second quantization, they are formed by populating a set of orthonormal single-particle orbitals, with one fermion in each mode.  Slater determinants admit an efficient classical description: they are uniquely specified by their one-particle reduced density matrix.  These states form a key ensemble across physics and chemistry;  for example, the Hartree-Fock method makes a mean-field approximation of interacting fermionic systems in terms of Slater determinants \cite{helgaker2013molecular}, serving as the first step for virtually all electronic-structure methods.
    
    \item \textbf{Fermionic Gaussian states} are a generalization of Slater determinants beyond fixed particle number or purity \cite{bardeen1957theory,bloch1962canonical}.  They are similarly described by an efficient classical representation: a covariance matrix encoding all fermionic one-body observables. Originally introduced to explain the pairing nature of BCS theory \cite{bardeen1957theory}, these days many other physical systems of electrons and nuclei are also understood using fermionic Gaussian states, such as non-interacting metals and topological insulators \cite{ring2004nuclear,coleman2015introduction}. This is related to the generalization of Hartree-Fock to fermionic Gaussian states, which can enable more accurate mean-field approximations in certain systems \cite{bach1994generalized}.
    
    \item \textbf{Bosonic Gaussian states} are multimode states whose Wigner functions are classical Gaussian distributions \cite{robinson1965ground,manuceau1968quasi}.  Equivalently, they are the Gibbs states of quadratic bosonic Hamiltonians, and they are also efficiently described by a covariance matrix and mean vector.  Gaussian states are a central class within the modern theory quantum optics \cite{glauber1963coherent}. Within quantum information \cite{weedbrook2012gaussian}, they are turnkey resources for many well-regarded proposals for fault-tolerant quantum computation \cite{menicucci2006universal}, boson sampling \cite{hamilton2017gaussian}, and quantum sensing \cite{caves1981quantum}. They have also been used as mean-field ansatz for cold atoms and other physical systems of interest \cite{jaksch1998cold}.
    
    \item \textbf{Coherent states} are a subset of bosonic Gaussian states that can be described extremely succinctly: they are displacements of the vacuum in phase space \cite{schrodinger1926stetige,glauber1963coherent}. Equivalently, this class can be described as all pure bosonic Gaussian states with zero squeezing. Coherent states constitute the standard quantum description of laser light, and therefore underpin much of modern optical communication and photonic technology \cite{mandel1996optical,zhang1990coherent}. Besides optics, coherent states have also been used to describe superfluidity \cite{langer1968coherent} and lattice vibrations via phonons \cite{lasher1968coherent}.
\end{itemize}

For each class above, we give an algorithm that distinguishes pure states in the class from pure states that are $\eps$-far in trace distance from all states in the class with sample complexity $O(1/\eps^2)$, independent of the number of modes. In addition, we prove matching lower bounds, and so our algorithms are asymptotically optimal. All of our algorithms are computationally efficient and only act on a few entangled copies of the input state at a time.

Besides revealing useful structural properties, our results may also enable a number of useful applications for testing Gaussianity.  
For example, physical systems known to exhibit Gaussian behavior \cite{ring2004nuclear,coleman2015introduction}, or quantum circuits designed to produce Gaussian states \cite{dallaire2019low, menicucci2006universal}, could be tested to determine if there are undesired interactions in the physical hardware. 
Determining the Gaussianity of a state also has implications for its resourcefulness toward universal quantum computation \cite{hebenstreit2019all}.
Gaussian state testing can also be used to determine if a physically accessible system can be effectively described by Gaussian theory, or if it otherwise possesses some ``hidden'' Gaussian structure \cite{lieb1961two,schultz1964two,kitaev2006anyons, fendley2019free,elman2021free}.
This is especially useful for systems where the ground state is only approximately Gaussian, since in those cases there might not be any exact mapping to a Gaussian theory.
In this case we may be able to find interesting physical systems with succinct classical descriptions. 

\subsection{Main results}

We begin with a formal definition of quantum property testing. We restrict our attention throughout to testing properties of \emph{pure} quantum states, and use the trace distance $\dist(\cdot, \cdot)$ as the metric.

\begin{definition}[Quantum property testing]\label{def:quantum-state-testing}
Let $\calH$ be a Hilbert space and $\mathcal C$ a class of pure states on $\calH$.
For $0 \leq \eps_1 < \eps_2$, $0 \leq \delta < \frac{1}{2}$, and integer $t$, a $t$-copy $(\eps_1, \eps_2, \delta)$-tester for $\mathcal C$ is an algorithm $\mathcal A: \Sym^t(\calH) \to \braces*{\accept, \reject}$ such that for all $\ket{\psi} \in \calH$,
\begin{enumerate}
    \item if $\dist\paren*{\ket{\psi}, \mathcal C} \leq \eps_1$, then 
        $\Pr\bracket*{\mathcal A\paren*{\ket{\psi}^{\otimes t}} = \accept} \geq 1 - \delta$; and
    \item if $\dist\paren*{\ket{\psi}, \mathcal C} \geq \eps_2$, then 
        $\Pr\bracket*{\mathcal A\paren*{\ket{\psi}^{\otimes t}} = \accept} \leq \delta$.
\end{enumerate}
\end{definition}

\begin{remark}[Tester as a projector]
Without loss of generality, we may assume the algorithm $\calA$ has the form of an orthogonal projector $\testproj$ on $\Sym^t(\calH)$ (possibly appended with some ancillary space). That is, the algorithm performs the two-outcome measurement $\braces*{\testproj, \identity - \testproj}$, accepting outcome $\testproj$ and rejecting on outcome $\identity - \testproj$.
\end{remark}

Our main result is that for both fermionic and bosonic gaussian states, and important subclasses thereof, there are tolerant testers with size-independent sample complexity. Furthermore, all of these testers act on at most two or three copies at a time, and they are computationally efficient.

\begin{theorem}[Testing upper bounds]\label{thm:general-testing-upper-bound}
    Let $\mathcal C$ be one of the following classes of pure states:
\begin{itemize}
    \item Fermionic Gaussian states,
    \item Slater determinants,
    \item Bosonic Gaussian states,
    \item Zero-mean bosonic Gaussian states,
    \item Coherent states,
\end{itemize}
each as a subset of the appropriate (fermionic or bosonic) Fock space.
Then there are constants $\beta \geq 1$ and $k \in \{2,3\}$, dependent on the class $\mathcal C$, such that for any $\delta \in (0, \frac{1}{2})$ and $0 \leq \eps_1 < \eps_2 < 1$ with $\eps_1 < \eps_2/\beta$, there is a 
$t$-copy $(\eps_1, \eps_2, \delta)$-tester for $\mathcal C$, acting on $k$ copies of the state at a time, with
\begin{equation}
    t = O\paren*{\frac{\log(1/\delta)}{\paren*{\eps_2 - \beta \eps_1}^2}}.
\end{equation}
\end{theorem}

In fact, only the full class of bosonic Gaussian states require $k = 3$; the other four testers measure $k = 2$ copies at a time. These results are displayed in \cref{table:results}.

\begin{table}
    \centering
    \begin{tabular}{l c c c l}
        \noalign{\hrule height 0.7pt}
        \textbf{Class of states}
        & \textbf{Theorem}
        & \textbf{Copies}
        & \textbf{Rounds}
        & \textbf{Previous upper bound} \\
        \hline
        \noalign{\vskip 1pt}
        Fermionic Gaussian
        & \cref{thm:testing-fermionic-gaussians}
        & $2$
        & $\Theta(1/\eps^2)$
        & $\widetilde{O}(n/\eps^2)$~\cite{leone2026fermionicentropy,haug2026practical}
        \\
        Slater determinant
        & \cref{thm:testing-slater}
        & $2$
        & $\Theta(1/\eps^2)$
        & $\widetilde{O}(n/\eps^2)$~\cite{leone2026fermionicentropy,haug2026practical}
        \\
        Bosonic Gaussian
        & \cref{thm:bosonic-gaussian-tester}
        & $3$
        & $\Theta(1/\eps^2)$
        & $O(n^4/\eps^2)$~\cite{girardi2025gaussiantesting}
        \\
        Zero-mean bosonic Gaussian
        & \cref{thm:testing-zero-mean-bosonic-gaussians}
        & $2$
        & $\Theta(1/\eps^2)$
        & $O(n^4/\eps^2)$~\cite{girardi2025gaussiantesting}
        \\
        Coherent
        & \cref{thm:testing-coherent-states}
        & $2$
        & $\Theta(1/\eps^2)$
        & $O(n/\eps^2)$
        \\
        \noalign{\hrule height 0.7pt}
    \end{tabular}
    \caption{Classes of states for which we obtain property testers with sample complexity $\Theta(1/\eps^2)$ along with the associated theorems in the manuscript. We specify the number of copies jointly measured in each round of the tester, and each of our algorithms consists of $\Theta(1/\eps^2)$ rounds. In the last column, we contrast our bounds to previous state-of-the-art upper bounds ($n$ is the number of fermionic/bosonic modes). The upper bound for testing coherent states is implicit in the literature, as coherent states can be learned via $O(n/\eps^2)$ homodyne measurements. Matching lower bounds for all classes are proved in \cref{sec:lower_bounds}.}
    \label{table:results}
\end{table}

\begin{remark}
One could also consider the task of testing these classes with some additional promise on the input.
For example, consider the task of testing whether or not a given $n$-mode, $\eta$-fermion state is (or is close to) a Slater determinant or not, as compared to the same task for a general $n$-mode state.
Our results are stronger in that they automatically encompass such promises but do not require them.
At best, additional promises on the input could only improve the constant factors.
\end{remark}

We also give lower bounds demonstrating that our result is essentially optimal.

\begin{theorem}[Testing lower bounds]\label{thm:general-testing-lower-bound}
Let $\mathcal C$ be one of the classes listed in \cref{thm:general-testing-upper-bound}. If there is a $t$-copy $(\eps_1, \eps_2, \frac{1}{3})$-tester for $\mathcal C$, then $t = \Omega\paren*{\frac{1}{\paren*{\eps_2 - \eps_1}^2}}$.
\end{theorem}

For intolerant testing ($\eps_1 = 0, \eps_2 \equiv \eps$), this proves the optimal sample complexity of $\Theta(1/\eps^2)$ with any constant probability for all of the classes considered. For tolerant testing, it is tight up to the constant $\beta$ appearing in the admissible gap between $\eps_1$ and $\eps_2$.

\subsection{Technical overview}\label{sec:technical_overview}

\paragraph{Upper bounds.} The classes we consider have well-known testers with perfect completeness. The main technical hurdle is to establish the soundness of the protocols. For all of our testers, we prove soundness by developing a new \say{symmetrized spectral gap} technique.
We overview the main steps of the argument below:

\begin{enumerate}
    \item\label{itm:proj_comm} \textbf{Project onto the commutant:}
    Fix a property $\calC$ and let $\Pi_{\mathrm{test}}$ be the projector onto the set $\{\ket{\psi}^{\otimes k}: \ket{\psi} \in \calC\}$. This projector will be the acceptance projector of our property testing algorithm. It can be easily argued that such a projector is complete, optimal and essentially unique among $k$-copy POVMs with perfect completeness~\cite{lovitz2026nearly,tarabunga2026fermionic}. 
    
    \item \textbf{Symmetrize:} For the purpose of the analysis, introduce an additional $m-k$ registers and let $\symproj^{(1,\ldots,m)}$ be the projector onto the symmetric subspace on all $m$ registers.
    Define the symmetrized test operator
    \[
    \Xi = \symproj^{(1,\ldots,m)}  \testproj^{(1,\ldots,k)}\symproj^{(1,\ldots,m)},
    \]
    where $\Pi_{\mathrm{test}}^{(1,\ldots,k)}$ is the projector from the previous step acting on the ``physical'' registers $1,\ldots,k$. The motivation for introducing the extra registers is the fact that we will take the gradient of the rejection probability $q(\psi) = 1 - \tr(\testproj \psi^{\otimes k})$ as a function of $\psi$; introducing extra registers and symmetrizing over them allows us more easily analyze the product rule unpacked within the trace. Note that we can safely symmetrize over all $m$ registers because the tester itself is already symmetric.
    
    \item \textbf{Prove a spectral gap:} It is clear that the maximum eigenvalue of $\Xi$ is $1$. Let $\lambda$ be the largest eigenvalue of $\Xi$ that is strictly less than $1$. The goal is to lower bound the gap $\Delta = 1 - \lambda$ by a function of $k$. This is the step that varies the most across different classes. For those considered in this work, there are nice Lie-group symmetries and so the eigenvalues of $\Xi$ can be analyzed from the representation theory of the commutant.

    \item \textbf{Gradient flow:} Define a gradient flow which begins at the input state $\ket{\psi}$ and moves in the direction of the class $\mathcal{C}$.  The length of the curve along this path is lower-bounded by a function of the distance between $\ket{\psi}$ and the class $\mathcal{C}$.  Meanwhile, the spectral gap of $\Xi$ implies a lower bound on the gradient of the curve along this path.  This yields an upper bound on the path length in terms of the failure probability $q(\psi)$ of the test.  Combining these bounds we can bound the distance between the input state and $\mathcal{C}$ in terms of this failure probability, which is precisely the soundness guarantee we need. In particular, if $\Delta = \Theta(1)$ is independent of dimension and larger than some $k$-dependent constant, then we get constant soundness.
\end{enumerate}

This general framework is applicable quite broadly, and we record the main technical conclusion as a master theorem in \cref{thm:rejection-infidelity-bounds}. Here, we now provide an overview the specific details involved in each of the classes that we consider in this paper.

\paragraph{Fermionic Gaussian states.} We analyze the following two-copy operator introduced by Bravyi \cite{bravyi2004lagrangian}, which has been considered for this problem in many previous studies:
\[
    \Lambda = \sum_{i=1}^n \left( c_i \otimes c_i^\dagger + c_i^\dagger \otimes c_i \right),
\]
where the $c_i$ and $c_i^\dagger$ are the fermionic lowering and raising operators, respectively. It is well-known that the set of states $\ket{\psi}^{\otimes 2}$ within $\ker \Lambda$ is exactly the class of pure Gaussian states, hence $\Pi_{\mathrm{test}}$ is the projector onto the kernel of $\Lambda$. For this class, we bound the spectral gap $\Delta = 5/12$ of $\Xi$ for $m=3$. We give two proofs of this spectral gap: one by explicitly computing the eigenbasis of $\Lambda$, and one by recognizing that $\Xi$ block-diagonalizes into the irreducible representations of tensor powers of the spin group.

\paragraph{Slater determinants.} For this class, we project onto the kernel of 
\[
K \coloneqq G^\dagger G, \quad \text{where } G \coloneqq \sum_{j=1}^n c_j \otimes c_j^\dagger.
\]
Note that this is related to the Bravyi operator via $\Lambda = G + G^\dagger$. However, it is Hermitianized by multiplication, rather than addition, as this ensures the resulting operator $K$ is particle-preserving in each register. Like $\Lambda$, it can be shown that $K\ket{\psi}^{\otimes 2} = 0$ if and only if $\ket{\psi}$ is a Slater determinant \cite{semenyakin2025classifying}. We also take $m = 3$ for this class, obtaining a spectral gap of $\Delta = 4/9$. Our analysis proceeds by diagonalizing an auxiliary operator sharing the eigenspectrum of $\Xi$ according to the irreducible representations of the general linear group in dimension $3$.

\paragraph{Bosonic Gaussian States.}

The $n$-mode bosonic Gaussian states are the orbit of a representation $U$ of the Jacobi group $\Jacobi(n) = \Heisenberg(n) \rtimes \Mp(2n)$.
The commutant of $U^{\otimes k}$ is $\Orth(k-1)$; this is a subgroup of the commutant $\Orth(k)$ for just $\Mp(2n)$ (corresponding to zero-mean Gaussian states).
The Gaussian states span the ``top'' irrep of $\Jacobi(n)$, which is paired with the trivial irrep of $\Orth(k-1)$.
For $k=2$, this includes all symmetric states, which necessitates going to $k=3$ copies to test for (general) Gaussianity.

Operationally, the three-copy test involves projecting onto the kernel of a quadratic operator $R$. This is easily done by applying a Gaussian change of basis in which $R$ is diagonal with eigenvalues $\paren*{\lambda_i}_i$, measuring in that Fock basis, and accepting if and only if $\sum_{i} \lambda_i x_i = 0$; it exploits the fact that the state $\ket{\psi}^{\otimes 3}$ to be tested is in the symmetric subspace.

\paragraph{Zero-mean bosonic Gaussian states.} In this case, the zero-mean assumption allows us to reduce the number of joint copies from three to two \cite{girardi2025gaussiantesting}. We consider a very close analogue of the fermionic Gaussian operator:
\[
    \Gamma = \ri \sum_{i=1}^n \left( a_i \otimes a_i^\dagger - a_i^\dagger \otimes a_i \right),
\]
where $a_i, a_i^\dagger$ are the bosonic ladder operators. The two-copy pure states $\ket{\psi}^{\otimes 2}$ that lie in $\ker \Gamma$ are exactly the bosonic Gaussian states with mean zero. For this class, we bound the spectral gap $\Delta = 5/12$ of $\Xi$ for $m=3$. Our analysis of this case is via the representation theory of the orthogonal group in dimension $3$.

\paragraph{Coherent states.} Our tester relies on a well-known identity governing the action of beam splitters on coherent states. Concretely, let $\beamsplitter$ be the evenly mixing beam splitter applied to each pair of modes across a central cut (the exact definition is given in \cref{sec:coherent-state-testing}). For a pair of coherent states $\ket{\alpha}, \ket{\beta}$, the beam splitter acts on them as follows:
\[
    \beamsplitter \ket{\alpha}\!\ket{\beta} = \left|\frac{\alpha + \beta}{\sqrt{2}}\right\rangle\!\left|\frac{\alpha - \beta}{\sqrt{2}}\right\rangle.
\]
In particular, applying $\beamsplitter$ to $\ket{\alpha}^{\otimes 2}$ leaves the second register in the vacuum state. Therefore, our tester simply applies the beam splitter and checks whether the second register is in the vacuum state. Perfect completeness then follows.

To show the spectral gap for coherent states, we are able to directly diagonalize an operator whose eigenspectrum is related to $\Xi$. 
Since the class of coherent states is product, the spectral gap analysis readily separates mode-by-mode. For each mode, we consider the action of $\beamsplitter$ on the raising and lowering operator on each mode, show that the corresponding $\testproj$ projects out one of the generators, and concretely compute the eigenvectors of the operator restricted to this subspace.\footnote{An alternative approach to proving soundness shows that the acceptance probability of the tester is proportional to the Husimi $Q$-function for single-mode coherent states, uses stability results for the Husimi function due to \cite{frank2023wehrl}, and extends this analysis to multiple modes. We decided to omit this approach as the analysis of \cite{frank2023wehrl} does not give an explicit constant, and the spectral gap approach is in accord with the rest of the results.}

\paragraph{Lower bounds.} We prove tightness of our upper bounds by constructing hard instances for testing. Specifically, any algorithm that can tolerantly test for class $\calC$ must at least be able to distinguish between two states, $\ket{\psi_A}$ which is $\eps_1$-close to $\calC$ and $\ket{\psi_B}$ which is $\eps_2$-far. The information-theoretic bound is Helstrom's \cite{helstrom1969quantum}: any two-outcome POVM on $\calH^{\otimes t}$ that successfully distinguishes between the two cases, with high probability, must have $t = \Omega(1/\log(1/F))$ where $F = \abs{\braket{\psi_A | \psi_B}}^2$ is the fidelity between the two states. Then, we construct the two states in a fairly general manner: suppose there exists some $\ket{A} \in \calC$ and another state $\ket{B}$, not necessarily in $\calC$, and orthogonal to $\ket{A}$. Then we can define a family of states
\[
\ket{\Psi_\theta} = \cos\theta\ket{A} + \sin\theta\ket{B}
\]
such that, as long as it holds that $|{\braket{A | \phi}}| + |{\braket{B | \phi}}| \leq 1$ for every $\ket{\phi} \in \calC$, then $\dist(\Psi_\theta, \calC) = \sin\theta$. Thus we can construct the desired hard instances $\ket{\psi_A}$, $\ket{\psi_B}$ by choosing small $\theta_1 = \arcsin \eps_1$, $\theta_2 = \arcsin \eps_2$. For such states, some elementary inequalities give us $\log(1/F) = O((\eps_2 - \eps_1)^2)$, and hence the lower bound on $t$ follows.

Note that even if $\calC' \subset \calC$, a lower bound for testing $\calC'$ does not imply one for $\calC$ (or vice versa) because the rejection probability is not monotone under set inclusion. Therefore, we have to exhibit hard instances for each class we consider. For fermionic Gaussian states, it suffices to take $\ket{A} = \ket{0^n}$ and $\ket{B} = \ket{1^4,0^{n-4}}$. For Slater determinants, any $\eta$-particle Fock states $\ket{A}, \ket{B}$ with disjoint supports works. Note that some mild restrictions are required for $n$ and $\eta$; every pure state on $n \leq 3$ modes is Gaussian \cite{bravyi2005classical}, while every pure fixed-particle state with $\eta \in \{0, 1, n-1, n\}$ is a Slater determinant.

For the bosonic cases, we use a slightly relaxed condition to construct the hard instances. Take some $\ket{X} \notin \calC$ such that
\[
    q \coloneqq \sup_{\ket{\phi} \in \calC} \abs{\braket{\phi | X}}^2 \leq \frac{1}{2}.
\]
Then the pair $\ket{A}, \ket{B}$ can be found by taking $\ket{A} \in \calC$ a maximizer to $q$, and
\[
\ket{B} = \frac{\ket{X} - \sqrt{q} \ket{A}}{\sqrt{1 - q}}.
\]
The family $\ket{\Psi_\theta}$ then induces the same type of hard instances. For both arbitrary bosonic Gaussian states and coherent states, the one-boson Fock state $\ket{X} = \ket{1, 0^{n-1}}$ suffices. For zero-mean Gaussian states we require even parity, so in that case we can simply use $\ket{X} = \ket{2, 0^{n-1}}$.

\subsection{Related work}

\paragraph{Classical property testing.} 
The notion of testing distributions originates from the classical literature \cite{goldreich1998property}, where some central questions include goodness-of-fit and independence testing \cite{canonne2022topics}. The classical problem analogous to our quantum one is testing an unknown distribution for Gaussianity with respect to total variation distance. However, the sample complexity of this problem is literally infinite \cite{rubinfeld2023testing}, in stark contrast to the quantum setting. Thus the problem of quantum Gaussian testing should be not conflated with or considered a strict generalization of classical Gaussian testing. Instead, the classical literature focuses on different properties such as whether the distribution, promised to be Gaussian, has zero mean, or if some finite subset of moments is consistent with Gaussians  \cite{ebner2020tests,narayanan2022private,diakonikolas2023gaussian}. For many applications of classical Gaussian testing, these less stringent features often suffice \cite{rubinfeld2023testing}.

\paragraph{Fermionic Gaussian testing.} 
The tester for this class that we study was first introduced by Bravyi \cite{bravyi2004lagrangian}, who showed that it has perfect completeness.  The underlying algebraic structure of our test has been studied in the past by de Melo, Ćwikliński, and Terhal \cite{melo2013power}, with a closely related recent line of inquiry aimed at understanding the commutant of Gaussian unitaries \cite{sierant2026theory, braccia2026commutant,lastres2026geometry}.  Historically, the observable associated with the tester was used to understand nonlocal quantum correlations via its Clifford-algebraic structure \cite{tsirel1987quantum}.  There are many alternative tests proposed in the literature including fermionic convolution \cite{lyu2024convolution}, fermionic antiflatness \cite{haug2026practical}, and simply performing full Gaussian tomography \cite{mele2025efficient,bittel2025optimal,chen2026optimal}.  However, these alternative tests are known to be suboptimal for two-copy tests with perfect completeness \cite{tarabunga2026fermionic}, suggesting that Bravyi's original operator is indeed the right choice to study.  For constant soundness, prior state-of-the-art results required a number of copies scaling no better than linearly with the number of modes, both for the Bravyi tester \cite{haug2026practical, tarabunga2026fermionic} and for other candidate testers \cite{leone2026fermionicentropy, walter2025random}. The related problem of testing fermionic Gaussian \textit{unitaries} was studied by Iyer \cite{iyer2025mildly}, who gave an algorithm to perform this task scaling polynomially in the number of modes; it was left open whether constant sample complexity was attainable in this context. 

For Slater determinants, a tester with perfect completeness has appeared many times in the literature \cite{kus2009classical, kotowski2010universal, oszmaniec2013universal, majtey2016multipartite, semenyakin2025classifying}, but no soundness results had been known.  The fermionic Gaussian tomography algorithms mentioned above can be applied to certify Slater determinants, but again require polynomially many copies of the input state.  There are also versions specifically adapted to Slater determinants, where the sample complexity is refined to depend on both number of modes and particle numbers \cite{aaronson_et_al:LIPIcs.TQC.2023.12,o2022fermionic,christensen2026learning}.

\paragraph{Bosonic Gaussian testing.} 
The most well-studied property tester for bosonic Gaussian states is due to Girardi \textit{et al}.~\cite{girardi2025gaussiantesting}. The proof of its soundness guarantee comes with a caveat: the input state must have a minimum fidelity (i.e., above some constant threshold) with some bosonic Gaussian state for the proof to claim dimension-independent testing. In the low-fidelity regime, the analysis only yields a polynomial scaling in system size (and mean energy).  The literature also contains a number of characterizations of Gaussianity \cite{hudson1974wigner,genoni2008quantifying,dai2026uncovering}, some of which can be converted to complete tests for bosonic Gaussianity \cite{hertz2019multicopy,springer2009conditions,cuesta2020stable,bu2025efficient, hahn2025measuring} (i.e., via \cite{harrow2013testing}), but none have known soundness.  Bosonic Gaussian states also have quantum tomography algorithms \cite{smithey1993measurement, mele2025learning, fanizza2411efficient, bittel2025energy}, but again these approaches scale with the system size.  

For multimode coherent states, the operator for the tester we analyze has been studied previously \cite{andersson2006experimentally, sedlak2007unambiguous, springer2009conditions, cuesta2020stable}, but no known soundness guarantees were established.  

\paragraph{Testing other quantum properties.} 

One of the earliest works to consider quantum property testing was that of Childs, Harrow, and Wocjan \cite{childs2007weak}, whose quantum collision problem gave rise to an early form of quantum spectrum and mixedness testing. O'Donnell and Wright \cite{odonnell2015quantum} formalized this problem and gave the optimal sample complexity: $\Theta(d)$ copies are necessary and sufficient to test if a $d$-dimensional state is maximally mixed, and $\Theta(r^2)$ to test if it has rank $\leq r$. A related problem is identity testing, wherein the concept class of states is a singleton set. B{\u{a}}descu, O'Donnell, and Wright \cite{buadescu2019quantum} showed that this problem can be solved in $O(d)$ copies. This can be exponentially improved if the target state is pure: Huang, Preskill, and Soleimanifar \cite{huang2025certifying} gave an algorithm with sample complexity $O(n^2)$ to certify a target $n$-qubit state, applicable to a $(1 - e^{-\Omega(n)})$-fraction of all pure states. This was later enhanced by Gupta, He, and O'Donnell \cite{gupta2026few}, who gave an $O(n)$-sample algorithm that provably worked for all pure states.

Arguably the first nontrivial constant-copy quantum property tester was given by Harrow and Montanaro \cite{harrow2013testing} for product states. That is, their tester can distinguish if a given pure state on $n$ qudits with local dimension $d$ is product or not, using only $O(1)$ copies independent of $n$ and $d$. Since this cornerstone result, follow-up works have sharpened the analysis of the Harrow-Montanaro tester, giving an optimal characterization of the acceptance probability for states far from product \cite{soleimanifar2022testing,beckey2026optimal}. This direction has been extended in various ways, for example testing for genuine multipartite entanglement \cite{harrow2017sequential}.

A natural conceptual extension of product states are matrix product states (MPS). Ostensibly strong candidates to demonstrate constant-copy testability, say for constant bond dimension, they were nevertheless shown by Soleimanifar and Wright \cite{soleimanifar2022testing} to require at least $\Omega(\sqrt{n})$ samples for any bond dimension $r \geq 2$. For upper bounds, they gave an $O(nr^2)$-sample algorithm to test $n$-qubit MPS of bond dimension $r$. Later work by Lovitz and Lowe \cite{lovitz2026nearly} nearly closed this gap in the superconstant bond dimension regime by considering a generalization of MPS, tree tensor network states (TTNS). There, they gave an upper bound also of $O(nr^2)$ copies to test TTNS, and a lower bound of $\Omega(nr^2/\log n)$ copies whenever $r \geq 2 + \log n$. For $r$ below this threshold, the best known lower bound is $\Omega(\sqrt{nr} + r)$ due to Chen, Wang, and Zhang \cite{chen2026local}.

Perhaps no class has been as well-studied in the context of property testing as stabilizer states.
The field of stabilizer testing was spearheaded by the seminal work of Gross, Nezami, and Walter \cite{gross2021schur}, who gave a six-copy test for stabilizerness .
A rich line of follow-up work has leveraged machinery from additive combinatorics to demonstrate that the Gross-Nezami-Walter test has strong tolerance properties \cite{arunachalam2025polynomial, grewal2024improved,bao2025tolerant, mehraban2025improved}. In particular, it can be shown that, for some constant $C > 1$, the test accepts with high probability if a state has fidelity at least $\eps$ with a stabilizer state and rejects with high probability if a state has fidelity at most $\eps^C$ with all stabilizer states, promised that one of the two cases holds.
These tolerant testing results have also been extended in a fairly black-box manner to the case of possibly mixed input states \cite{iyer2025tolerant}.

\section*{Acknowledgments}
\addcontentsline{toc}{section}{Acknowledgments}

We thank Carlos Bravo-Prieto, Arkopal Dutt, Vinayak Kumar, Michael Jaber, Hari Krovi, Antonio Anna Mele, Francesco Anna Mele, Geoffrey Mon, and Barak Nehoran for helpful conversations. Part of this work was completed while VI was supported by an NSF Graduate Research Fellowship, and while AZ was supported by the Laboratory Directed Research and Development program at Sandia National Laboratories under the Gil Herrera Fellowship in Quantum Information Science. AZ is also supported by the National Nuclear Security Administration’s Advanced Simulation and Computing Program.  OP and KT are supported by the U.S.\ Department of Energy, Office of Science, Accelerated Research in Quantum Computing, Fundamental Algorithmic Research toward Quantum Utility (FAR-Qu). 

This article has been authored by an employee of National Technology \& Engineering Solutions of Sandia, LLC under Contract No.\ DE-NA0003525 with the U.S. Department of Energy (DOE). The employee owns all right, title and interest in and to the article and is solely responsible for its contents. The United States Government retains and the publisher, by accepting the article for publication, acknowledges that the United States Government retains a non-exclusive, paid-up, irrevocable, world-wide license to publish or reproduce the published form of this article or allow others to do so, for United States Government purposes. The DOE will provide public access to these results of federally sponsored research in accordance with the DOE Public Access Plan \url{https://www.energy.gov/downloads/doe-public-access-plan}.

\paragraph{Concurrent work.}
Toward the completion of this manuscript, we became aware of concurrent and independent work by Yaghubi Rad, Puig, Hassanandani, Coffman,
Eisert, Bravo-Prieto, and Mele \cite{yaghubi2026testing}. Their work also proves dimension-independent sample complexity for testing the classes of Gaussian states that we consider here. We have planned for both manuscripts to appear on the arXiv simultaneously.

\paragraph{AI usage disclosure.}
The authors acknowledge the use of ChatGPT 5.6 Sol Ultra in the preparation of this manuscript. 
We provided the model with a large document of structural lemmas and calculations as well as a list of possible approaches to complete the argument, aimed squarely at the fermionic Gaussian instance.
In response, GPT 5.6 proposed a proof technique resembling the spectral gap technique outlined in \cref{sec:framework}, along with a candidate proof for the spectral gap for the fermionic Gaussian tester. The latter invoked techniques from algebraic geometry that are considerably foreign to the intended audience of this work.
In this work, we provide two alternative proofs of this spectral gap that we believe are much more accessible to the quantum information community: one using representation theory and the other using carefully chosen decompositions of fermionic Fock space.
Furthermore, we further developed and generalized the symmetrized spectral gap technique, and subsequently applied it to the other classes of states present in this manuscript (Slater determinants, bosonic Gaussian state, zero-mean bosonic Gaussian states).

All proofs were carefully verified and written down entirely by the authors, and we claim full responsibility for the correctness of the results presented in this work. 
Lighter models of ChatGPT 5.6 and Claude Opus 5 were used in sharpening calculations and refining proofs.

\section{Preliminaries}\label{sec:preliminaries}

\paragraph{Notation.} The set $\{1,2,\ldots,n\}$ is denoted by $[n]$. The Dirac delta function $\delta_{xy}$ is equal to $1$ if $x = y$ and $0$ otherwise. For a Boolean event $A$, $\indic{A}$ is the indicator function of $A$: it is $1$ when $A$ is true and $0$ when $A$ is false. For a polynomial $p$ in a formal variable $z$ we write $[z^d]p(z)$ to mean the coefficient of $z^d$ in $p$. $S^{d-1}$ denotes the $(d-1)$-dimensional unit sphere in $\R^{d}$.

The $n \times n$ identity matrix is denoted $\identity_n$. Often, we will drop the subscript when the dimension is apparent. We denote by $e_i$ the unit vector in $\R^n$ where the $i$th coordinate is $1$ and all other entries are $0$. The commutator of matrices $A$ and $B$ is denoted $[A,B] = AB - BA$. The $k$th order nested commutator of $A$ and $B$, denoted $[A,B]_k$, is defined as $[A,B]_0 = B$ and $[A,B]_{k+1} = [A, [A,B]_{k}]$ for $k \geq 0$. We denote the set of eigenvalues, or spectrum, of an operator $A$ as $\spec(A)$.

We use $\dist(\rho, \sigma) = \frac{1}{2} \|\rho - \sigma\|_1$ for the trace distance between two density operators $\rho, \sigma$. For pure states $\ket{\psi}$ we frequently write $\psi$ for its density operator, whence $\dist(\psi, \phi) = \sqrt{1 - \abs{\braket{\psi | \phi}}^2}$.

\paragraph{Multilinear algebra.} For a vector space $V$ we write $\Sym^k(V)$ and $\bigwedge^k V$ for the symmetric and antisymmetric subspaces of $V^{\otimes k}$, respectively. The Fock space on $V$ is the direct sum over all $k$:
\[
\bigvee V \equiv \bigoplus_{k=0}^\infty \Sym^k(V) \quad \text{and} \quad \bigwedge V \equiv \bigoplus_{k=0}^{\dim V} {\bigwedge}^k V.
\]
Fock basis elements are typically written in their occupation-number representation, e.g., the simple antisymmetric $k$-vector $e_{i_1} \wedge \cdots \wedge e_{i_k}$ corresponds to $\ket{x}$ where $x \in \{0,1\}^{\dim V}$ has $x_i = 1$ if and only if $i \in S \coloneqq \{i_1, \ldots, i_k\}$. Equivalently, we also write $\ket{S} \equiv \ket{x}$. The index order for such states is assumed to be increasing. Similar conventions hold for the symmetric setting $e_{i_1} \vee \cdots \vee e_{i_k}$ where now $S$ can be a weakly increasing multiset.

\subsection{Representation theory}

We follow the presentation of standard texts \cite{fulton2004representation,goodman2009symmetry}. Write $\genlingrp(d, \C), \unitgrp(d), \symplgrp(2d, \R)$, etc.\ for the classical (matrix) Lie groups and $\genlinalg(d, \C)$ etc.\ for their Lie algebras.
Our convention for symplectic group $\symplgrp(2d, \R)$ is that $S \in \symplgrp(2d, \R)$ if and only if $S^\transpose \Omega S = \Omega$ where
\begin{equation}\label{eq:symplectic-form}
    \Omega = \begin{pmatrix}
        0 & \identity_d \\ -\identity_d & 0
    \end{pmatrix}.
\end{equation}
$\spingrp(d)$ is the double cover of $\splorthgrp(d)$.

If $G$ is a group and $V$ is a vector space, a continuous map $\rho : G \to \genlingrp(V)$ is called a representation if it preserves group multiplication. A representation is said to be irreducible (an irrep) if it has no nontrivial subrepresentations. When the group and map are clear from context, we may refer to a representation simply by its vector space $V$ and use $g \cdot v \equiv \rho(g)v$, $g \in G$, for its group action.

\begin{fact}[Schur's lemma]
    If $V$ is an irrep and $T : V \to V$ is linear, then $T$ commutes with $V$ if and only if $T$ is a scalar multiple of the identity.
\end{fact}

A linear-algebraic group is a matrix group defined by polynomial equations in the matrix entries. A representation $\rho$ of a linear-algebraic group $G$ is said to be rational if it is finite-dimensional and, for all $g \in G$, the entries of $\rho(g)$ are polynomials in $g_{ij}$ and $\det(g)^{-1}$. A rational representation $V$ of $G$ is completely reducible if every $G$-invariant subspace $W \subset V$ has a some other invariant subspace $U \subset V$ such that $V = W \oplus U$. If every rational representation of such $G$ is completely reducible, then $G$ is said to be reductive.

Let $T$ be a maximal torus (connected abelian subgroup) of a reductive $G$. A weight of a rational representation $W$ is a character $\lambda : T\to\C^\times$ for which the corresponding weight space
\[
    W[\lambda] \coloneqq \{v\in W : t \cdot v = \lambda(t) v \text{ for all } t\in T\}
\]
is nonzero. Weights can be put into partial order and every irrep has a unique maximal weight, called its highest weight, which uniquely specifies the irrep up to isomorphism. We write $V_\lambda$ for the irrep of highest weight $\lambda$.

\begin{fact}[Highest-weight decomposition {\cite[Propositions 4.1.15 and 12.1.1]{goodman2009symmetry}}]\label{prop:highest-weight-decomp}
    Every rational representation $W$ of a reductive group $G$ decomposes as
    \[
    W \cong \bigoplus_\lambda V_\lambda \otimes M_\lambda
    \]
    where $\lambda$ runs over the irreducible types that occur and $M_\lambda$ is a multiplicity space on which $G$ acts trivially. 
\end{fact}

Finally, the determinant character on a matrix group $G$ is denoted by $\det : G \to \C^\times$, so that ${\det}^k$ is the one-dimensional representation $g \mapsto (\det g)^k$.

\subsubsection{General linear group}

Let $V$ be a representation of $\genlingrp(d, \C)$. For diagonal $t = \diag(t_1, \ldots, t_d) \in \genlingrp(d, \C)$, a vector $v \in V$ has weight $\lambda \equiv (\lambda_1, \ldots, \lambda_d)$ if $t \cdot v = t_1^{\lambda_1} \cdots t_d^{\lambda_d} v$. The weight space $V[\lambda]$ is the subspace of $V$ spanning only those vectors of weight $\lambda$. If $E_{ij}$ is the elementary matrix in the Lie algebra $\genlinalg(d, \C)$, then a weight vector $v$ is said to be of highest weight if $E_{ij} \cdot v = 0$ for all $i < j$. We write $V_\lambda$ for the unique irrep of $\genlingrp(d, \C)$ of highest weight $\lambda$; by the theorem of highest weight, it is generated by repeatedly applying $E_{ij}$ for all $i > j$ to a highest-weight vector of weight $\lambda$ (for example, see \cite[Proposition 12.11]{fulton2004representation}).

\begin{fact}[Highest weights for $\genlingrp(d)$ and $\spllingrp(d)$ {\cite[Theorem 5.5.22]{goodman2009symmetry}}]\label{prop:GL-highest-weight}
    The rational irreps $V_\lambda$ of $\genlingrp(d, \C)$ are labeled by weakly decreasing tuples of integers $\lambda = (\lambda_1, \ldots, \lambda_d)$. The restriction of $V_\lambda$ to $\spllingrp(d, \C)$ has highest weight $\mu = (\mu_1, \ldots, \mu_{d-1})$ where $\mu_i = \lambda_i - \lambda_{i+1}$.
\end{fact}

\begin{fact}[Branching rule for $\genlingrp(d) \downarrow \genlingrp(d-1)$ {\cite[Theorem 8.1.1]{goodman2009symmetry}}]\label{prop:GL-interlacing}
    Let $V_\lambda^{(d)}$ and $V_\mu^{(d-1)}$ be irreps of highest weight for $\genlingrp(d, \C)$ and $\genlingrp(d-1, \C)$, respectively. In the restriction of $V_\lambda^{(d)}$ to $\genlingrp(d-1, \C)$ (in the sense of the embedding $A \mapsto \diag(A, 1)$), the irrep $V_\mu^{(d-1)}$ occurs with multiplicity $1$ exactly when the interlacing condition
    \[
    \lambda_1 \geq \mu_1 \geq \lambda_2 \geq \cdots \geq \lambda_{d-1} \geq \mu_{d-1} \geq \lambda_d
    \]
    is met. Otherwise, it does not occur. 
\end{fact}

\begin{fact}[Determinant twist rule {\cite[Section 15.5]{fulton2004representation}}]\label{prop:det-twist}
    For any integer $k$,
    \[
    V_{(\lambda_1+k, \ldots, \lambda_d+k)} \cong V_{(\lambda_1, \ldots, \lambda_d)} \otimes {\det}^k.
    \]
\end{fact}

Finally, we will make particular use of the Lie algebra for the \emph{special} linear group in $d = 2$. We can characterize when its modules (vector spaces with an $\spllinalg(2, \C)$-action that preserves the Lie bracket $[A, B]$) are irreducible by using the following fact.

\begin{fact}[Characterization of irreducible $\spllinalg(2)$-modules {\cite[Section 11.1]{fulton2004representation}}]\label{prop:sl2-module}
    Every irreducible finite-dimensional $\spllinalg(2, \C)$-module contains, for some integer $k \geq 0$, only vectors of weights
    \[
    k, k-2, \ldots, -k+2, -k.
    \]
    Applying $E_{21}$ maps each weight vector to the next according to the string above, until hitting the lowest weight.
\end{fact}

Upon reaching the lowest-weight vector, the lowering operator $E_{21}$ then annihilates it. An easy consequence of this fact is that a weight-zero vector is killed by $E_{21}$ (and also $E_{12}$) if and only if the module has $k = 0$.

\subsubsection{Spin group}

\begin{fact}[Schur decompositions under $\spingrp$]\label{fct:schur-decompositions}
Let $n$ be a positive integer and $V_{\lambda}$ be the irrep of $\spingrp(2n)$ with highest weights $\lambda = (\lambda_1, \ldots, \lambda_n)$.
Set $S_{\pm} = V_{\left(\frac12, \ldots, \frac12, \pm\frac12\right)}$ and $S = S_+ \oplus S_-$.
Then
\begin{align}
S^{\otimes k} &\cong 
\bigoplus_{\substack{
    \frac{k}{2} \geq \lambda_1 \geq \cdots \lambda_{n-1} \geq \abs{\lambda_n} \\
    \lambda_1, \ldots, \lambda_n \in \mathbb Z + \frac{k}{2}
}} 
V_{\lambda} \otimes V_{\lambda'},
\end{align}
where
\begin{align}
\lambda' &= (\lambda'_1, \ldots, \lambda'_{\lfloor k/2 \rfloor}),\\
\lambda'_j &= \abs{\left\{
        \lfloor \abs{\lambda_i} \rfloor < \lfloor k /2 \rfloor + 1 -j
\right\}}
\end{align}

\begin{align}
S_+^{\otimes k} &\cong 
\bigoplus_{\substack{
    \frac{k}{2} \geq \lambda_1 \geq \cdots \lambda_{n-1} \geq \abs{\lambda_n} \\
    \lambda_1, \ldots, \lambda_n \in \mathbb Z + \frac{k}{2} \\
    \sum_{i=1}^n \lambda_i \in 2 {\mathbb Z} + \frac{nk}{2} \\
    \sum_{i=1}^{n-1} \lambda_i - \lambda_n \leq \frac{nk}{2} - k
}} V_{\lambda} \otimes M_{\lambda},
\end{align}
where $M_{\lambda}$ is the multiplicity space.

\begin{align}
\mathrm{Sym}^{k+1}S_+ &\cong
\bigoplus_{\substack{
    \frac{k+1}{2} \geq \lambda_1 \geq \cdots \geq 0 \\
    \lambda_1, \ldots, \lambda_n \in \mathbb Z + \frac{k+1}{2}
}} 
V_{\lambda} \otimes M_{\lambda},
\end{align}
where $M_{\lambda}$ is the multiplicity space.
\begin{align}
V_{(\frac{k}{2}, \ldots, \frac{k}{2})} \otimes S_+ &\cong 
\bigoplus_{l=0}^{\lfloor n / 2 \rfloor}
V_{\left(
    \underbrace{\frac{k+1}{2}, \ldots, \frac{k+1}{2}}_{n - 2l}, 
    \underbrace{\frac{k-1}{2}, \ldots, \frac{k-1}{2}}_{2l}
\right)}
\end{align}
\end{fact}

\begin{fact}[$\spingrp(2n)$ irrep dimensions]\label{fct:irrep-dims}
\begin{align}
\dim V_{{\left(\frac{k}{2}\right)}^n} &= \prod_{1 \leq i < j \leq n} \frac{k + 2n - i - j}{2n - i - j},\\
\dim V_{
{\left(\frac{k+1}{2}\right)}^{2(m-i)}
{\left(\frac{k-1}{2}\right)}^{2i}
} &= 
\dim V_{{\left(\frac{k}{2}\right)}^n}
\binom{2n}{2i}
\frac{
    {\left(\frac{k}{2}\right)}^{\overline{i}}
    {\left(\frac{k + n + 2i}{2}\right)}^{\overline{m-i}}
}{
    {\left(\frac{k+2i - 1}{2}\right)}^{\overline{i}}
    {\left(\frac{k + 4i + 1}{2}\right)}^{\overline{m-i}}
}
.
\end{align}
\end{fact}

\begin{fact}[$\spingrp(d)$ branching]\label{fct:branching}
    Let $\lambda$ be the highest weights of a $\spingrp(d)$ irrep. Then on restriction to $\spingrp(d-1)$,
    \begin{align}
        V_{\lambda}^{\spingrp(d)} \downarrow_{\spingrp(d-1)} &
        \overset{\spingrp(d-1)}{\cong}
        \bigotimes_{\lambda' \preceq \lambda} V_{\lambda'}^{\spingrp(d-1)},
    \end{align}
where $\lambda'$ is the highest weights of a $\spingrp(d-1)$ irrep and
   \begin{align}
  \lambda' \preceq \lambda & \Leftrightarrow \begin{cases}
      \lambda_1 \geq \lambda'_1 \geq \lambda_2 \geq \cdots \geq \lambda'_{n-1} \geq \lambda_n \geq \abs{\lambda'_n}
      , &d=2n+1,\\
      \lambda_1 \geq \lambda'_1 \geq \lambda_2 \geq \cdots \geq \lambda'_{n-1} \geq \lambda'_{n-1} \geq \lambda_n
      , &d=2n.
  \end{cases}
  \end{align}
\end{fact}

\subsubsection{Lie groups}

\begin{theorem}[{Weyl's dimension formula \cite[Theorem 5.84]{knapp1996lie}}]
    Let $V$ be a finite-dimensional irrep of the complex semisimple Lie algebra $\mathfrak{g}$ with highest weight $\lambda$. Then
    \[
    \dim V = \prod_{\alpha \in \Delta^+} \frac{ \braket{\lambda + \delta, \alpha}}{\braket{\delta, \alpha}},
    \]
    where $\Delta^+$ are the positive roots and $\delta$ is the Weyl vector.
\end{theorem}

\begin{fact}
For type $D_n$ (e.g., $\mathfrak{so}(2n+1)$),
\begin{align}
\Delta^+ &= \bracebar{e_i \pm e_j}{i < j} \cup \braces*{\pm e_i},\\
\delta &= \sum_{i=1}^{n} \paren*{n + \frac12 -i} e_i.
\end{align}
For type $D_n$ (e.g., $\mathfrak{so}(2n)$),
\begin{align}
\Delta^+ &= \bracebar{e_i \pm e_j}{i < j},\\
\delta &= \sum_{i=1}^{n-1} (n-i) e_i.
\end{align}
\end{fact}

\subsubsection{Orthogonal group}

\paragraph{Legendre polynomials.} The representation theory of $\splorthgrp(3)$ is intimately related to spherical harmonics and Legendre polynomials, which we review here first. The Legendre polynomials are functions $P_\ell : [-1,1] \to \R$, commonly defined via the Rodrigues formula:
\[
P_\ell(x) = \frac{1}{2^\ell \ell!} \dfrac{d^\ell}{dx^\ell}(1 - x^2)^\ell.
\]

The \emph{associated} Legendre polynomials $P_{\ell}^{(m)}$ for $0 \leq m \leq \ell$ are obtained from the Legendre polynomials as follows:
\[
    P_{\ell}^{(m)}(x) = (-1)^m (1 - x^2)^{m/2}\dfrac{d^m}{dx^m} P_\ell(x).
\]
In particular, $P_{\ell}^{(0)}(x) = P_\ell(x)$. For negative $m \geq - \ell$, we can define the associated Legendre polynomials
\[
P_{\ell}^{(m)}(x) = (-1)^m \frac{(\ell + m)!}{(\ell - m)!} P_{\ell}^{-m}(x).
\]

We give a useful identity obeyed by the Legendre polynomials.

\begin{proposition}\label{prop:legendre-inner-product-identity}
Let $v = (x, y, z) \in S^2$ and consider any $u,w \in S^2$. Then
\[
    \dfrac{2 \ell + 1}{4 \pi} \int_{S^2} P_\ell(u \cdot v)  P_\ell(w \cdot v) \ \rd x \, \rd y \, \rd z = P_\ell(u \cdot w),
\]
where $u \cdot v$ is the standard inner product in $\R^3$.
\end{proposition}

\paragraph{Representation theory of $\splorthgrp(3)$.} For integers $\ell,m$ such that $-\ell \leq m \leq \ell$, we define the spherical harmonics
\begin{equation}\label{eq:spherical-harmonics}
Y_{\ell,m}(\theta,\phi) = (-1)^m \sqrt{\dfrac{(2\ell + 1)(\ell - m)!}{4\pi(\ell + m)!}}P_\ell^{(m)}(\cos \theta)e^{\ri m \phi},
\end{equation}
where $P_\ell^{(m)}$ are the associated Legendre polynomials. Here, we are abusing notation slightly and taking $P_\ell^{(m)} = P_\ell^{(-m)}$ when $m < 0$.
For brevity, we will denote
\begin{equation}
    c_{\ell, m} = \sqrt{\dfrac{(2\ell + 1)(\ell - m)!}{4\pi(\ell + m)!}}
\end{equation} %
We define
\[
\calY_\ell = \linspan\{Y_{\ell,m}: -\ell \leq m \leq \ell\}
\]
to be a $(2 \ell + 1)$-dimensional subspace of $L^2$-integrable functions on $S^2$. It is a well-known fact that the $\calY_\ell$ form the irreps of $\splorthgrp(3)$ \cite{fulton2004representation}.
To be concrete, we describe the action of $R$ on these spherical harmonics $Y_{\ell,m}$. Let $v = (\sin \theta \cos \phi, \sin \theta \sin \phi, \cos \theta)$ be a vector in $S^2$. Any $R \in \splorthgrp(3)$ is a rotation about the origin that maps the angle pair $(\theta, \phi)$ to a pair $(\theta^\prime, \phi^\prime) = f_R(\theta, \phi)$, where $f_R$ is some (possibly nonlinear) function corresponding to the rotation $v \to R^{-1}v$. For a basis element $Y_{\ell, m}(\theta,\phi)$, the action $U_R$ of the rotation $R \in \splorthgrp(3)$ is simply:
\begin{equation}\label{eq:so-3-action-spherical-harmonics}
    U_R Y_{\ell,m}(\theta, \phi) = Y_{\ell,m}(\theta^\prime, \phi^\prime).
\end{equation}

We will be specifically interested in the case where $R$ is a rotation about the $z$-axis:
\[
R_z(\alpha) =
\begin{pmatrix}
    \cos \alpha & -\sin \theta & 0 \\
    \sin \alpha & \cos \alpha & 0 \\
    0 & 0 & 1
\end{pmatrix}.
\]
For this case, we have
\begin{equation}\label{eq:spherical-z-rotation-action}
U_{R_z(\alpha)} Y_{\ell, m}(\theta, \phi) = Y_{\ell, m}(\theta, \phi - \alpha) = e^{-\ri m\alpha}Y_{\ell, m}(\theta, \phi).
\end{equation}

\paragraph{Representation theory of $\orthgrp(3)$.} The irreps of $\orthgrp(3)$ are simply those of $\splorthgrp(3)$ with an adjoined parity component. That is, each irrep of $\orthgrp(3)$ is isomorphic to one of
\[
\calY_\ell,\quad \calY_\ell \otimes {\det},
\]
where ${\det}$ is the determinant representation.

\subsection{Fermions}
In this section we give a brief overview of fermionic systems. For a more thorough overview, we refer the reader to \cite{bravyi2004lagrangian}.
We consider fermionic systems of $n$ modes, each of which can be occupied by at most one fermion. As such, the corresponding Hilbert space is a Fock space $\bigwedge \C^n$ of dimension $2^n$. The fermionic creation and annihilation operators $c_i^\dagger$ and $c_i$ insert or remove a fermion in mode $i$, respectively.
These operators satisfy the canonical anticommutation relations:
\[
\{c_i, c_j\} = 0,\quad \{c_i, c_j^\dagger\} = \delta_{ij}.
\]
The algebra spanned by the $2n$ creation and annihilation operators is the space of complex $2^n \times 2^n$ matrices. One way to see is this by introducing the Majorana operators $\majorana_1,\ldots,\majorana_{2n}$:
\[
\majorana_{2i-1} =c_i + c_i^\dagger, \quad \majorana_{2i} = -\ri(c_i - c_i^\dagger).
\]
It can be shown that the canonical anticommutation relations in terms of Majorana operators are equivalent to $\{\majorana_i, \majorana_j\} = 2\delta_{ij}$. 

We often represent fermionic states in the language of Fock states. The vacuum state $\ket{0^n}$ is the state where every mode has occupation number $0$; all other states are formed by populating modes with creation operators:
\[
\ket{x_1, \ldots, x_n} = \prod_{i=1}^n (c_i^\dagger)^{x_i} \ket{0^n}, \quad x_i \in \{0, 1\}.
\]
An equivalent notation is to label Fock states by subsets $S \subseteq [n]$ where $S = (i : x_i = 1)$. We take the standard sign convention: $\ket{S} = c_{S_1}^\dagger \cdots c_{S_{|S|}}^\dagger \ket{0^n}$ and $S_1 < \cdots < S_{|S|}$.

By the fermionic parity superselection rule, all physical fermionic pure states must have fixed parity \cite{wick1997parity,streater2000pct}. That is, the Fock space has a natural grading between $\ket{S}$ with $|S|$ even or odd, and physical fermionic states cannot have superpositions across this grading. It can be shown that an operator commutes with parity if and only if it is a linear combination of even-degree polynomials in the Majorana operators.

Fermionic Gaussian unitaries are unitary operators generated by Hamiltonians that are quadratic in the creation and annihilation operators. These form a representation of the spin group $\spingrp(2n)$ on Fock space. Fermionic Gaussian states are the Gibbs states of such Hamiltonians.
Pure fermionic Gaussian states are the rank-$1$ zero-temperature states; equivalently, they are the union of the orbits of $\ket{0^n}$ and $\ket{1,0^{n-1}}$ under all fermionic Gaussian unitaries. We denote the set of pure fermionic Gaussian states by $\fermionicgaussianset$. Slater determinants $\slaterdeterminantset \subset \fermionicgaussianset$ are the subset of pure fermionic Gaussian states that are also eigenstates of the number operator $N = \sum_{i=1}^n c_i^\dagger c_i$.

\subsection{Bosons}

The $n$-mode bosonic Fock space $\bigvee \C^n$ is an infinite-dimensional symmetric subspace, spanned by Fock basis states of the form $\ket{x_1,\ldots,x_n}$ where $x_i \in \mathbb{N}_0$. Unlike the fermionic case, we are not restricted to a maximum occupation of $1$ per mode. We can similarly define raising and lowering operators $a_i^\dagger$ and $a_i$ for each mode, but the canonical relations are now \emph{commutation} relations:
\[
[a_i, a_j] = 0, \quad [a_i, a_j^\dagger] = \delta_{ij}.
\]
Otherwise, the Fock basis states are generated similarly (i.e., by applying raising operators on the vacuum $\ket{0^n}$).

Bosonic Gaussian unitaries are an important class of bosonic operations that are uniquely defined by their action on raising and lowering operators by conjugation. Concretely, define the vector of operators $v = (q_1, \ldots, q_n, p_1, \ldots, p_n)$ where the quadrature operators $q_i, p_i$ are
\[
q_i = \frac{a_i + a_i^\dagger}{\sqrt{2}}, \quad p_i = \frac{a_i - a_i^\dagger}{i\sqrt{2}}.
\]
A bosonic Gaussian unitary $U$ is specified by a matrix $S \in \symplgrp(2n, \R)$ and a vector $r \in \R^{2n}$ such that
\[
Uv_iU^\dagger = r_i + \sum_{j=1}^{2n} S_{ij} v_j.
\]
The vector $r = (r_1, \ldots, r_{2n})$ is known as the displacement vector of $U$. If $r = 0^n$, the unitary $U$ is said to have mean zero. The group of all such zero-mean unitaries forms a projective representation of the symplectic group $\symplgrp(2n, \R)$. It is known that any bosonic Gaussian unitary can be written as the product of a zero-mean bosonic Gaussian unitary $V_S$ and a displacement unitary $D_r$:
\begin{equation}\label{eq:bosonic-decomposition}
    U = D_{r}V_S.
\end{equation}
The zero-mean unitary $V_S$ is sometimes referred to as the symplectic action unitary, and the displacement unitary $D_r$ is given by
\[
    D_r = \exp\left(\ri \sum_{k=1}^n (\alpha_k a_k^\dagger + \alpha_k^* a_k)\right).
\]

Bosonic Gaussian states are defined to be the Gibbs states of bosonic Hamiltonians that are at most quadratic in the raising and lowering operators.
Analogous to the fermion case, the set of pure Bosonic Gaussian states can also be defined as the orbit of $\ket{0^n}$ under all bosonic Gaussian unitaries. Zero-mean Gaussian states are those in the orbit of only zero-mean Gaussian unitaries. The set of bosonic Gaussian states is denoted $\bosonicgaussianset$, and the set of zero-mean bosonic Gaussian states is $\zeromeanbosonicgaussianset \subset \bosonicgaussianset$.

Another important subset of $\bosonicgaussianset$ is the class of coherent states, also known as Glauber states, which are ubiquitous in quantum optics. For a given $\alpha \in \C$, there is a corresponding single-mode coherent state
\[
    \ket{\alpha} = e^{-\abs{\alpha}^2/2} \sum_{n = 0}^\infty \frac{\alpha^n}{\sqrt{n!}}\ket{n}
\]
which is essentially the vacuum displaced in phase space coordinates $(q, p)$ by $\alpha = \frac{q + \ri p}{\sqrt{2}}$.
The class of all single-mode coherent states is denoted $\coherentstateset_1$, or $\coherentstateset$ for short. The class of all $n$-mode coherent states is then $\coherentstateset_n$, which is given by the tensor product of $n$ single-mode coherent states:
\[
\ket{\alpha} = \ket{\alpha_1} \otimes \cdots \otimes \ket{\alpha_n}
\]
where now $\alpha \in \C^n$.
In the multimode setting, any coherent states is reachable from the vacuum state by applying some displacement operator:
\[
    \ket{\alpha} = D_\alpha \ket{0^n}.
\]

\begin{fact}\label{fact:coherent-identity-resolution}
The set of coherent states resolve to the identity:
    \begin{equation}
    \frac{1}{\pi}\int_{\C} \ketbra{\alpha}{\alpha}\ \rd^2\alpha = \identity.
    \end{equation}
\end{fact}

\subsection{Statistics}

We will need surprisingly few sophisticated statistical tools in this work; the following lemma controls the sample complexity for distinguishing between the two cases (close vs.\ far) in property testing.

\begin{lemma}[Distinguishing Bernoulli distributions]\label{lem:bernoulli-sample-complexity}
Let $0 \leq p_1 < p_2 \leq 1$ and $0 < \delta < 1$. Consider two Bernoulli distributions with means $p_1$ and $p_2$.
	Let $k = O\left(\frac{\log(1/\delta)}{{(\sqrt{p_2}-\sqrt{p_1})}^2}\right)$,
    and 
    $$
    \theta =\begin{cases} 
	\displaystyle
    \frac{
		\log\frac{1-p_1}{1-p_2}
	}{
		\log \frac{p_2(1-p_1)}{p_1(1-p_2)}
	}&\text{ if $p_1\neq 0$ and $p_2 \neq 1$}\\
    \, \\
    \displaystyle
    \frac{p_2+p_1}{2} &\text{ otherwise }.
    \end{cases}
    $$
	Taking the mean $\hat{p}$ of $k$ samples and returning $1$ if $\hat{p} \leq \theta$ and $2$ if $\hat{p} > \theta$ succeeds in distinguishing the two cases with probability at least $1-\delta
	$.
\end{lemma}
\begin{proof}
    First we examine the edge cases.  Let $X=\sum_{i=1}^k X_i $ be the sum of $k$ i.i.d.\ Bernoulli samples.  If $p_1=0$, then
    \begin{align*}
        \Pr_{p_1}\left[ X > k\frac{p_2}{2} \right]=0 \text{       and     }\Pr_{p_2} \left[ X \leq k\frac{p_2}{2} \right] \leq e^{-k p_2/8},
    \end{align*}
    where the inequality is due to the Chernoff-Hoeffding bound.  It follows that $k=O(\log(1/\delta)/p_2)=O(\log(1/\delta)/(\sqrt{p_2})^2$ many samples can drive the probability of misclassification below $\delta$.

    If $p_2=1$ then we have:
    \begin{align*}
        \Pr_{p_1}\left[ X > k\frac{1+p_1}{2} \right]\leq e^{-k(1-p_1)/8} \text{       and     }\Pr_{p_2} \left[ X \leq k\frac{1 + p_1}{2} \right] =0,
    \end{align*}
    where the inequality is again due to the Chernoff-Hoeffding bound.  It follows that $k \geq 8 \log(1/\delta)/(1-p_1)$ are sufficient to drive the probability of misclassification below $\delta$.  Since $\log(1/\delta)/(1-p_1) \leq \log(1/\delta)/(1-\sqrt{p_1})^2$, $k=O(\log(1/\delta)/(1-\sqrt{p_1})^2)$ many samples are also sufficient to drive the misclassification probability below $\delta$.

    Now let us assume $0 < p_1 < p_2 < 1$.  Define
    \begin{align*}
        t=\frac{1}{2}\log\left( \frac{p_2(1-p_1)}{p_1(1-p_2)}\right).
    \end{align*}
    Then,
    \begin{align*}
        e^{-t \theta}=\sqrt{\frac{1-p_2}{1-p_1}}.
    \end{align*}
    By Markov's inequality,
    \begin{align*}
        \Pr_{p_1} \left[X > \theta k \right]&= \Pr_{p_1} \left[e^{tX} > e^{t \theta k} \right] \leq \Ex_{p_1} e^{t X-t \theta k}.\\
        \Pr_{p_2} \left[X \leq \theta k \right]&=\Pr\left[-X \geq -\theta k\right]=\Pr_{p_2}\left[e^{-Xt} \geq e^{-\theta kt}\right] \leq \Ex_{p_2} e^{-Xt+\theta k t}.
    \end{align*}
    By the known moment generating function of Bernoulli random variables,
    \begin{align*}
        \Pr_{p_1} \left[X > \theta k \right] \leq \left(e^{-t \theta}(1-p_1+p_1e^t) \right)^k=\left(\sqrt{(1-p_1)(1-p_2)}+\sqrt{p_1 p_2}\right)^k.\\
         \Pr_{p_2} \left[X \leq  \theta k \right] \leq \left(e^{t\theta}(1-p_2+p_2e^{-t}) \right)^k =\left(\sqrt{(1-p_1)(1-p_2)}+\sqrt{p_1 p_2}\right)^k.
    \end{align*}
    Note that 
    \begin{align*}
    2-2\left( \sqrt{p_1 p_2} +\sqrt{(1-p_1)(1-p_2)}\right)=\left(\sqrt{p_1}-\sqrt{p_2}\right)^2+\left(\sqrt{1-p_1}-\sqrt{1-p_2}\right)^2\geq \left(\sqrt{p_1}-\sqrt{p_2}\right)^2,
    \end{align*}
    so
    \begin{align*}
        \sqrt{p_1 p_2} +\sqrt{(1-p_1)(1-p_2)} \leq 1-\frac{\left(\sqrt{p_1}-\sqrt{p_2}\right)^2}{2}\leq \exp\left(-\frac{\left(\sqrt{p_2}-\sqrt{p_1}\right)^2}{2} \right).
    \end{align*}
    Hence in this case taking $k \geq 2\log(1/\delta)/(\sqrt{p_2}-\sqrt{p_1})^2$ is sufficient to drive the error probability below $\delta$.
\end{proof}

\subsection{Geometry of quantum states}

Our gradient flow argument relies on the treatment of (projective) Hilbert space as a Riemannian manifold. The Fubini-Study metric is the standard choice here; a comprehensive introduction can be found in \cite{bengtsson2007geometry}.

\begin{definition}[Fubini-Study metric]\label{def:fubini-study}
    For a normalized state $\ket{\psi}$ and two elements $\ket{d\psi_1}, \ket{d\psi_2}$ of the tangent space, the \emph{Fubini-Study metric} is  
    \begin{equation*}
        \fubiniStudyMetric(\ket{d\psi_1}, \ket{d\psi_2}) = \Re\left[\braket{d\psi_1|d\psi_2}
            - \braket{d\psi_1 | \psi} \braket{\psi | d\psi_2}
        \right].
    \end{equation*}
    This induces the line element
    \begin{equation*}
        \fubiniStudyLineElement = \sqrt{\fubiniStudyMetric(\ket{d\psi}, \ket{d\psi})}
        =
        \sqrt{\norm{\ket{d\psi}}^2 - \abs{\braket{\psi|d\psi}}^2}
        .
    \end{equation*}
    and the distance
    \begin{equation*}
        \fubiniStudyDistance(\ket{\psi_1}, \ket{\psi_2}) = \arccos{\abs{\braket{\psi_1|\psi_2}}}.
    \end{equation*}
\end{definition}

The gradient over quantum states is induced by the directional derivative, which we define below.

\begin{definition}[Directional derivative]\label{def:directional_derivative}
    Consider a projective Hilbert space $\calH$ and a continuous function $f: \calH \to \R$. For states $\ket{\psi}, \ket{\eta} \in \calH$, 
    consider the perturbation of $\ket{\psi}$ in the direction of $\ket{\eta}$:
    \[
        \ket{\zeta(t)} = \frac{\ket{\psi} + t \ket\eta}{\norm{\ket{\psi} + t\ket{\eta}}}, \quad t \in \R.
    \]
    The \emph{$\eta$-directional derivative of $f$} evaluated at $\psi$ is the following limit, assuming it exists:
    \[
        D_\eta f(\psi) = \lim_{t \to 0} \frac{f(\zeta(t)) - f(\psi)}{t}.
    \]
\end{definition}

\section{The central framework: Soundness from spectral gaps}\label{sec:framework}

Here we outline our general framework for proving \Cref{thm:general-testing-upper-bound} for fermionic Gaussian states, Slater determinants, and bosonic Gaussian states.
For each class $\calC$, there is a $k$-copy permutation-invariant projector $\testproj$ such that $\testproj \ket{\psi}^{\otimes k} = \ket{\psi}^{\otimes k}$ if and only if $\ket{\psi} \in \calC$. That is, $\testproj$ is a tester with perfect completeness (i.e., accepts with probability $1$ if $\ket{\psi} \in \calC$). Such projectors are well-studied in the literature; our main contribution is a qualitative improvement on the soundness of these testers.

Write $\psi$ for the density matrix of $\ket{\psi}$ and let $q\paren*{\psi} \coloneqq 1 - \tr\paren*{\testproj \psi^{\otimes k}}$ be the rejection probability of the test. We want to relate it to the distance of $\psi$ from $\calC$,
\begin{equation}\label{eq:eps_min_dist}
    \eps \coloneqq \dist\paren*{\psi, \calC} \equiv \inf_{\ket{\phi} \in \calC} \dist(\psi, \phi)
\end{equation}
where $\dist(\psi, \phi) \coloneqq \frac{1}{2} \|\psi - \phi\|_1$ is the trace distance. Showing that $q(\psi) = O(\eps^2)$ is easy; since all states are pure, one has a clean expression
\[
\tr(\testproj \psi^{\otimes k}) \geq \abs{\braket{\phi | \psi}}^{2k} \implies q(\psi) \leq 1 - \abs{\braket{\phi | \psi}}^{2k}
\]
for a minimizing $\ket{\phi}$ to \cref{eq:eps_min_dist}. Since $\dist(\psi, \phi)^2 = 1 - \abs{\braket{\phi | \psi}}^2 = \eps^2$, we get
\begin{equation}\label{eq:q_upper_bound_easy}
    q(\psi) \leq 1 - (1-\eps^2)^k \leq k\eps^2.
\end{equation}

The more difficult direction is showing that $q = \Omega(\eps^2)$.
We achieve this by bounding the spectral gap $\Delta$ of a symmetrized version of $\testproj$ on $m > k$ copies.

\begin{definition}[Symmetrized projectors and spectral gaps]\label{def:sym_test_op}
Let $\testproj$ be a permutation-invariant projector on the Hilbert space $\calH^{\otimes k}$. Let $m > k$ and denote $\symproj^{(1,\ldots,m)}$ for the projector onto the symmetric subspace $\Sym^{m}(\calH) \subseteq \calH^{\otimes m}$. Define
\begin{equation}
    \Xi \coloneqq \symproj^{(1,\ldots,m)} \testproj^{(1,\ldots,k)} \symproj^{(1,\ldots,m)},
\end{equation}
where $\testproj^{(1,\ldots,k)}$ acts on the original $k$ copies of $\calH$.
The \emph{spectral gap} $\Delta$ of $\Xi$ is the smallest nonzero eigenvalue of $\identity - \Xi$.
\end{definition}

The utility of this spectral gap is that it provides us a useful lower bound on the magnitude of the gradient $\norm{\nabla q(\psi)}$ over the manifold of quantum states.
This allows us to lower bound $q(\psi)$ as a function of the distance $\eps$ from the set $\calC$ via a gradient flow argument. Note that this is a purely analytical tool; the tester does not actually consume any additional $m - k$ copies.

\subsection{Geometry of the landscape}

Let $R \coloneqq \identity - \testproj$ be the projector onto the rejection subspace within $\calH^{\otimes k}$. For a state $\ket{\psi}$ we define the following contracted vector on one register:
\begin{equation}
    \ket{u_\psi} \coloneqq (\identity \otimes \bra{\psi}^{\otimes (k-1)}) R \ket{\psi}^{\otimes k}.
\end{equation}
Intuitively, this contraction filters out the Gaussian component of $\ket{\psi}$ and compresses the non-Gaussianity into a useful form. Indeed, this resembles the three-copy convolution studied in the context of fermionic Gaussian testing \cite{lyu2024convolution,coffman2025nongaussian}. By definition, the overlap with $\ket{\psi}$ is the rejection probability,
\[
    \braket{\psi | u_\psi} = q(\psi).
\]
On the other hand, the component perpendicular to $\ket{\psi}$ measures how this rejection probability changes as a function of $\psi$. Let $\ket{\eta}$ be orthogonal to $\ket{\psi}$ and define the parametric family
\begin{equation}
    \ket{\psi(t)} \coloneqq \frac{\ket{\psi} + t\ket{\eta}}{\sqrt{1 + t^2}}.
\end{equation}
The infinitesimal change of $q(\psi)$ along the $\eta$ direction can be determined by differentiating with respect to $t$. By the product rule,
\begin{equation}
    \left. \frac{d}{dt} \ket{\psi(t)} \right|_{t=0} = \ket{\eta} \implies \left. \frac{d}{dt} \ket{\psi(t)}^{\otimes k} \right|_{t=0} = \sum_{i=1}^k \ket{\psi}^{\otimes(i-1)} \otimes \ket{\eta} \otimes \ket{\psi}^{\otimes(k-i)}.
\end{equation}
Let $|\dot{\Psi}^{(k)}\rangle = \left. \frac{d}{dt} \ket{\psi(t)}^{\otimes k} \right|_{t=0}$. Since $R\ket{\psi}^{\otimes k}$ is symmetric with respect to any permutation of the registers, every term of the inner product between $|\dot{\Psi}^{(k)}\rangle$ and $R\ket{\psi}^{\otimes k}$ is identical, hence
\begin{equation}
\begin{split}
    \left. \frac{d}{dt} q(\psi(t)) \right|_{t=0} &= \braket{\dot{\Psi}^{(k)} | R | \psi^{\otimes k}} + \braket{\psi^{\otimes k} | R | \dot{\Psi}^{(k)} }\\
    &= 2k \Re\left[ (\bra{\eta} \otimes \bra{\psi}^{\otimes(k-1)}) R \ket{\psi}^{\otimes k} \right]\\
    &= 2k \Re\braket{\eta | u_\psi}.
\end{split}
\end{equation}
This holds for every $\ket{\eta}$ orthogonal to $\ket{\psi}$. The gradient of $q(\psi)$ is therefore represented by the tangent vector
\[
\nabla q(\psi) = 2k (\ket{u_\psi} - q(\psi) \ket{\psi}),
\]
(note that $q(\psi) = \braket{\psi | u_\psi}$ is real), and so its norm with respect to the Fubini-Study metric, $\fubiniStudyMetric(\eta, \zeta) = \Re\braket{\eta | \zeta}$, is
\begin{equation}\label{eq:gradient_expression_u_psi}
\begin{split}
    \|\nabla q(\psi)\|_{\mathrm{FS}}^2 &= 4k^2 \left\|\ket{u_\psi} - q(\psi) \ket{\psi} \right\|^2\\
    &= 4k^2 (\|u_\psi\|^2 - 2q(\psi) \Re\braket{\psi | u_\psi} + q(\psi)^2 \|\psi\|^2)\\
    &= 4k^2 (\|u_\psi\|^2 - q(\psi)^2).
\end{split}
\end{equation}
Hence, for a lower bound on $\|\nabla q(\psi)\|_{\mathrm{FS}}^2$ we need a lower bound on $\|u_\psi\|^2$. We take $\Xi$ from \cref{def:sym_test_op} with the fixed choice of $m = 2k-1$. The following lemma shows how the spectral gap of $\Xi$ controls this gradient; this is a type of \lojasiewicz{} inequality on the class $\calC$, treated as a manifold equipped with the Fubini-Study metric.

\begin{lemma}[\lojasiewicz{} inequality]\label{lem:lojasiewicz-inequality}
Let $m > k \geq 2$ be integers.
Let $R$ be a permutation-invariant projector on a Hilbert space $\calH^{\otimes k}$ and $r\paren*{\psi} = \tr(R \ketbra{\psi}{\psi}^{\otimes k})$.
Let $\Delta$ be the spectral gap of $R^{(m)} = \symproj^{(1,\ldots,m)} R \symproj^{(1,\ldots,m)}$.
Then
\begin{align}
\norm{\nabla r}^2 & \geq 4 r \paren*{c_0 - c_1 r},
\label{eq:lojasiewicz-inequality}
\end{align}
where
\begin{align}
c_0 = \frac{m(m-1) \Delta - k(k-1)}{m-k}, \quad c_1 = (m + k - 1).
\end{align}
\end{lemma}

\begin{corollary}
Define
\begin{align}
\Delta^* = \frac{k (k -1)}{m(m-1)}, \quad c' = \frac{c_0}{c_1}.
\end{align}
Then \cref{eq:lojasiewicz-inequality} is trivial (i.e., the right-hand side is negative) when $\Delta < \Delta^*$ or $r \geq c'$.
For $k=2$ and $m=3$,
\begin{align}
    \norm{\nabla r}^2 & \geq 4 r \paren*{2\paren*{3 \Delta - 1} - 4 r}, &
    \Delta^* &= \frac13, &
    c' &= \frac{3 \Delta - 1}{2}.
\end{align}
For $k=3$ and $m=7$,
\begin{align}
    \norm{\nabla r}^2 & \geq 4 r \paren*{\frac32 \paren*{7 \Delta - 1} - 7 r}, &
    \Delta^* &= \frac{1}{7}, &
    c' &= \frac{7 \Delta - 1}{6}.
\end{align}
\end{corollary}

\begin{proof}[Proof of \cref{lem:lojasiewicz-inequality}]
For concision, let $T = R^{(m)}$.
First, write
\begin{align}
    \ket{\eta} = R \ket{\psi}^{\otimes k} &= \sum_{j=0}^k \ket{\eta_j},
\end{align}
where
\begin{align}
    \ket{\eta_j} &= \sum_{\substack{
            x \in \braces*{0, 1}^k \\
            \sum_l x_l = j
    }}
    \bigotimes_{l=1}^{k} 
    \paren*{\ket{\psi}\bra{\psi}}^{1 - x_l}
    \paren*{1 - \ketbra{\psi}{\psi}}^{x_l}
    \ket{\eta}.
\end{align}
Let $a_j = \norm{\ket{\eta_j}}^2$, so that
\begin{align}
\sum_{j=0}^k a_j &= \norm{R\ket{\psi}^{\otimes k}}^2 = r,\\
\ket{\eta_0} &= r \ket{\psi}^{\otimes k},\\
a_0 &= r^2,\\
\sum_{j=1}^k a_j &= r - r^2.
\end{align}
Let
\begin{align}
    \ket{\psi'} &= \paren*{1 - \ketbra{\psi}{\psi}}
    \paren*{\bra{\psi}^{\otimes(k-1)} \otimes \identity}
    \ket{\eta},
\end{align}
so that 
\begin{align}
    \nabla r &= 2k \ket{\psi'}.
\end{align}
Because $\braket{\psi | \psi'} = 0$,
\begin{align}
\ket{\eta_1} &= \sum_{l=1}^k \ket{\psi}^{\otimes (l-1)} \otimes \ket{\psi'} \otimes \ket{\psi}^{\otimes (k-l)}
\end{align}
and thus
\begin{align}
    \norm{\nabla r}^2 &= 4 k^2 \norm{\psi'}^2 = 4k a_1.
\end{align}
Next, we show that
\begin{align}
    \bra{\psi}^{\otimes m} T^2 \ket{\psi}^{\otimes m} &= \sum_{j=0}^k \frac{\binom{k}{j}}{\binom{m}{j}} a_j.
\end{align}
By symmetry, for each $j$ there is a $j$-register state $\ket{\upsilon_j}$ such that
\begin{align}
    \ket{\eta_j} &= \sum_{
        \substack{
            J \subset [k] \\
            \abs{J} = j
    }} \ket{\upsilon_j}_J,
\end{align}
where $\ket{\upsilon_j}_J$ has $\ket{\upsilon_j}$ on the registers in $J$ and $\ket{\psi}$ elsewhere.
We also have
\begin{align}
a_j &= \binom{k}{j} \norm{\ket{\upsilon_j}}^2.
\end{align}
Similarly,
\begin{align}
T \ket{\psi}^{\otimes m} &= \frac{1}{\binom{m}{k}}
\sum_{\substack{
        I \subset [m] \\
        \abs{I} = k
}}
\sum_j
\sum_{\substack{
        J \subset I \\
        \abs{J} = j
}}
\ket{\upsilon_j}_J
\\
&=
\sum_j
\frac{\binom{m-j}{k-j}}{\binom{m}{k}}
\sum_{\substack{
        J \subset [m] \\
        \abs{J} = j
}}
\ket{\upsilon_j}_J
=
\sum_j
b_j
\sum_{\substack{
        J \subset [m] \\
        \abs{J} = j
}}
\ket{\upsilon_j}_J,
\end{align}
where we define $b_j = \frac{\binom{k}{j}}{\binom{m}{j}}$.
Therefore
\begin{align}
\bra{\psi}^{\otimes m} T^2 \ket{\psi}^{\otimes m} &= 
\sum_j
\binom{m}{j} b_j^2 \norm{\ket{\upsilon_j}}^2
=
\sum_j
b_j a_j.
\end{align}
Because $b_{j+1}/b_j = (k-j)/(m-j) \leq 1$, $b_j \geq b_{j + 1}$ and therefore
\begin{align}
\Delta r & \leq 
\bra{\psi}^{\otimes m} T^2 \ket{\psi}^{\otimes m} \\
&\leq
b_0 a_0 + b_1 a_1 + b_2 \sum_{j=2}^k a_j
=r^2 + b_1 a_1 + b_2 (r - r^2 - a_1)
\\
&=b_2 r + (1 - b_2)r^2 + (b_1 - b_2) a_1
\\
&=b_2 r + (1 - b_2)r^2 + (b_1 - b_2) \frac{\norm{\nabla r}^2}{4k},
\end{align}
i.e.,
\begin{align}
    \norm{\nabla r}^2 
&\geq 
\frac{4k}{b_1 - b_2}\bracket*{
\Delta r - b_2 r - (1-b_2)r^2
}
\\
&=
\frac{4k m(m-1)}{k(m-k)}\bracket*{
    \Delta r - \frac{k(k-1)}{m(m-1)} r - \frac{(m-k)(m+k-1)}{m(m-1)}r^2
}
\\
&=
4r
\bracket*{
    \frac{m(m-1)\Delta  - {k(k-1)}}{m-k}  - (m+k-1)r
}
=
4r \paren*{c_0 - c_1r}. \notag \qedhere
\end{align}
\end{proof}

\subsection{Bounded gradients imply soundness}

Using this lower bound on $\|\nabla q\|_{\mathrm{FS}}^2$, we can establish $q = \Omega(\epsilon^2)$ via a gradient flow argument.

\begin{lemma}[Gradient descent distance]\label{lem:gradient-descent-distance}
Let $\calC$ be a closed set of pure states in a Hilbert space $\calH$. Regard $\calH$ as a projective Hilbert space.
Let $\eps\paren*{\psi} \coloneqq \dist(\psi, \calC)$.
Let $r: \calH \to [0, 1]$ be a smooth function such that 
\begin{itemize}
    \item $r\paren*{\psi} = 0$ if and only if $\ket{\psi} \in \mathcal C$, and
    \item $\norm{\nabla r(\psi)}_{\mathrm{FS}}^2 \geq 4 r(\psi) (c_0 - c_1 r(\psi))$ for all $\ket{\psi}$
\end{itemize}
for some constants $0 < c_0 \leq c_1$ with $c_1 \geq 1$.
Then for every $\ket{\psi}$ it holds that
\begin{equation}\label{eq:distance-lower-bound}
    r(\psi) \geq c'\eps(\psi)^2, \quad \text{where } c' \coloneqq \min\braces*{\frac12, \frac{c_0}{c_1}}.
\end{equation}
\end{lemma}

\begin{proof}
Abbreviate $\eps \equiv \eps(\psi)$. Let $f(\eps) \coloneqq c'\eps^2$.
If $r\paren*{\psi} = 0$, then $\ket{\psi} \in \calC$ and $\eps = 0$, so trivially $r(\psi) = 0 = f(0)$.
Similarly, if $r\paren*{\psi} \geq c'$, then $r\paren*{\psi} \geq c' = \max_{\eps' \in [0, 1]} f(\eps') \geq f(\eps)$.
Therefore, the remainder of the proof assumes that $0 < r\paren*{\psi} < c'$, which implies that $\norm{\nabla r}_{\mathrm{FS}} > 0$.

We will consider $\calH$ as both a projective Hilbert space and a Riemannian manifold with the Fubini-Study metric $\fubiniStudyMetric$.
Consider the family of states $\ket{\psi(t)}$ defined by
\[
\begin{aligned}
    \ket{\psi(0)} &= \ket{\psi},\\
    |\dot{\psi}(t)\rangle &= -\nabla r(t)
\end{aligned}
\]
where $r(t) \equiv r(\psi(t))$. The line element is given by 
\begin{equation}\label{eq:ds-dt}
    \fubiniStudyLineElement = \sqrt{\norm{d\psi}^2 - \abs{\braket{\psi|d\psi}}^2}
    =
    \sqrt{\norm{\dot{\psi}}^2 - \abs{\braket{\psi|\dot{\psi}}}^2} \, dt
    =\norm{\dot{\psi}}_{\mathrm{FS}} \, dt
    =\norm{\nabla r}_{\mathrm{FS}} \, dt.
\end{equation}
The differential of $r$ is given by
\begin{equation}
\label{eq:dq-dt}
    dr
    = g(\nabla r, d\psi)
    = g(\nabla r, \dot{\psi} \, dt)
    = -g(\nabla r, \nabla r) \, dt
    = - \norm{\nabla r}^2_{\mathrm{FS}} \, dt.
\end{equation}
Hence, 
$$
\frac{dr}{dt} =- \norm{\nabla r}^2_{\mathrm{FS}} \leq -4 r(t) (c_0 - c_1 r(0))
$$ 
with $r(0) \leq c_0/c_1$.  It follows that $r(t)\rightarrow 0$ as $t\rightarrow \infty$.  For $T > s$ define $r_T=r(\psi(T))$ and $r_s=r(\psi(s))$,
\begin{align}
    \fubiniStudyDistance\left(\ket{\psi(s)}, \ket{\psi(T)}\right) &\leq
    \int_{\left\{\ket{\psi(t')}\right\}_{s \leq t'\leq T}} \fubiniStudyLineElement(\ket{\psi(t)}) \tag{Distance is the minimal length of a curve}\\
    &\leq \int_s^T \norm{\nabla r}_{\mathrm{FS}} dt \tag{By \cref{eq:ds-dt}}\\
    &=\int_{r_s}^{r_T} \norm{\nabla r}_{\mathrm{FS}} \left(-\frac{dr}{\norm{\nabla r}_{\mathrm{FS}}^2}\right)\tag{By \cref{eq:dq-dt}}=\int_{r_T}^{r_s}  \frac{dr}{\norm{\nabla r}_{\mathrm{FS}}}\\
    &\leq \int_{r_T}^{r_s} \frac{dr}{2\sqrt{r\paren*{c_0-c_1 r}}} \tag{By assumption on $\norm{\nabla r}_{\mathrm{FS}}$}\\
    &\label{eq:dist_calc}=F(r_s)-F(r_T),
\end{align}
where 
$$
F(x)=\frac{1}{\sqrt{c_1}} \arcsin\sqrt{\frac{c_1 x}{c_0}}.
$$
Since $r(t)\rightarrow 0$ the flow is Cauchy and converges to a limit $\ket{\psi_\infty}=\lim_{T\rightarrow \infty}\ket{\psi(T)}$ satisfying $r(\psi_\infty)=0$.

Therefore:
\begin{align*}
    \arcsin\paren*{\eps}
    &=
    \arccos\sqrt{1 - \eps^2}
    =
    \min_{\ket{\phi} \in \calC}
    \arccos{\abs{\braket{\phi| \psi(0)}}}
    \\
    &\leq
    \arccos{\abs{\braket{\psi(0)| \psi_\infty}}}
    \\
    &=
    \fubiniStudyDistance\left(\ket{\psi(0)}, \ket{\psi_\infty}\right)
    \\
    &\leq F(r(\psi))-F(0) \tag{By \cref{eq:dist_calc}}\\
    &= \frac{1}{\sqrt{c_1}} \arcsin\sqrt{\frac{c_1 r}{c_0}}\\
    &\leq \arcsin\sqrt{\frac{r}{c'}}.\label{eq:stronger-trig-bound}
    \end{align*}
Inverting the $\arcsin$ yields the advertised claim.
\end{proof}

Combining \cref{lem:lojasiewicz-inequality,lem:gradient-descent-distance} gives the main conclusion of our framework.

\begin{theorem}[Soundness from spectral gaps]\label{thm:rejection-infidelity-bounds}
    Let $\calC$ be a nonempty class of pure states on a Hilbert space $\calH$, and let $\testproj$ be a permutation-invariant orthogonal projector on $\calH^{\otimes k}$ with $k \geq 2$ such that
\[
    \testproj \ket{\psi}^{\otimes k} = \ket{\psi}^{\otimes k}
    \quad \text{if and only if } \ket{\psi}\in\calC.
\]
Suppose that the symmetrized operator $\Xi$ on $m > k$ registers has spectral gap at least $\Delta$ (\cref{def:sym_test_op}), where
\begin{equation}\label{eq:gap_threshold}
    \frac{k(k-1)}{m(m-1)}< \Delta \leq 1.
\end{equation}
Then, for $\eps = \dist(\psi,\calC)$, it holds that
\begin{equation}\label{eq:framework-soundness}
    \alpha\eps^2 \leq q(\psi) \leq k\eps^2, \quad \text{where } \alpha = \min\braces*{\frac{1}{2}, \frac{m(m-1)\Delta - k(k-1)}{(m - k)(m+k-1)}}.
\end{equation}
\end{theorem}

\begin{proof}
The upper bound was derived in \cref{eq:q_upper_bound_easy}, so it remains to prove the lower bound. The rejection probability $q(\psi)$, as a function on $\calH$, satisfies the hypotheses of \cref{lem:gradient-descent-distance} with
\[
    c_0 = \frac{m(m-1)\Delta - k(k-1)}{m-k} \quad \text{and} \quad c_1 = m + k - 1,
\]
as established by \cref{lem:lojasiewicz-inequality}. Observe that the conditions $0 < c_0 \leq c_1$ hold, equivalently
\begin{align}
    c_1 - c_0 &= (m + k - 1) - \frac{m(m-1)\Delta - k(k-1)}{m-k} \\
    &= \frac{(m + k - 1)(m-k) - m(m-1)\Delta + k(k-1)}{m-k} \\
    &= \frac{m(m-1) - k(k-1) - m(m-1)\Delta + k(k-1)}{m-k} \\
    &= \frac{(1 - \Delta)m(m-1)}{m-k} \geq 0
\end{align}
as long as $\Delta$ obeys both bounds in equation \eqref{eq:gap_threshold} (the condition $c_1 \geq 1$ holds since $k \geq 2$). Then, \cref{lem:gradient-descent-distance} implies the lower bound of \eqref{eq:framework-soundness}.
\end{proof}

\subsection{The sample complexity of tolerant testing}

Finally, for completeness we show how this sandwich bound on $q(\psi)$ implies a sample complexity upper bound for tolerant testing of $\calC$. The arguments here are fairly standard. First, we depict the algorithmic framework. It requires various constants dependent on $\calC$, and notably a \say{base tester} $\calA$ which measures $\ket{\psi}^{\otimes k}$ according to $\testproj$.

\begin{algorithm}%
\caption{Tolerant quantum property tester.}\label{alg:constant-copy-tester}
\KwInput{
    Hilbert space $\calH$, 
    class $\calC \subset \calH$, constants $0 \leq \eps_1 < \eps_2 \leq 1$ such that $\eps_2 >  \beta \eps_1$ where $\beta = \sqrt{\frac{\alpha_1}{\alpha_2}}$ for $\alpha_1, \alpha_2$ from \cref{thm:framework-tolerant-testing}, failure probability $0 < \delta < 1/2$, copies of a state $\ket{\psi}$ with the promise that $\dist(\psi, \calC) \notin (\eps_1, \eps_2)$, and a base algorithm $\calA : \Sym^k(\calH) \to \{\accept\equiv 0, \reject\equiv 1\}$ for testing $\calC$.}

    \KwOutput{\accept if $\dist(\psi, \calC) \leq \eps_1$; \reject if $\dist(\psi, \calC) \geq \eps_2$ (correct with probability at least $1 - \delta$).}
    
    \BlankLine

    Let $t \leftarrow O\mathopen{}\left( \frac{\log(1/\delta)}{\alpha_2(\eps_2-\beta\eps_1)^2} \right)$\\
    
    \For{$i \in [t]$}{
        $s_i \leftarrow \calA(\ket{\psi}^{\otimes k})$
    }
    
    Let $\displaystyle s \leftarrow \frac{1}{t} \sum_{i=1}^t s_i$\\
    
    Let $(q_1, q_2) \leftarrow (\alpha_1 \eps_1^2, \alpha_2 \eps_2^2)$\\
    
    Let $\displaystyle\theta \leftarrow \frac{
        \log \frac{1 - q_1}{1 - q_2}
    }{
    \log \frac{q_2 (1-q_1)}{q_1(1-q_2)}
    }$ \hfill \tcp{$\theta \gets \frac{q_1+q_2}{2}$ if $q_1 = 0$ or $q_2=1$}
    
    \uIf{$s \leq \theta$}{\KwRet{\accept}}
    
    \Else{\KwRet{\reject}}
\end{algorithm}

\begin{theorem}[Tolerant testing sample complexity]\label{thm:framework-tolerant-testing}
Let $\calC \subset \calH$ be nonempty. Suppose the rejection probability $q(\psi)$ of a base property tester $\calA$ on $k$ copies satisfies, for all pure states $\ket{\psi} \in \calH$,
\begin{equation}
    \alpha_2 \, \dist(\psi, \calC)^2 \leq q(\psi) \leq \alpha_1 \, \dist(\psi, \calC)^2
\end{equation}
with $\alpha_1 \geq 1$ and $0 < \alpha_2 \leq 1/2$. Then \cref{alg:constant-copy-tester} is correct; that is, it is a $kt$-copy $(\eps_1, \eps_2, \delta)$-tester for $\calC$ (\cref{def:quantum-state-testing}), provided that $\eps_2 > \beta\eps_1$ where $\beta \equiv \sqrt{\alpha_1/\alpha_2}$, and
\[
    t = O\left( \frac{\log(1/\delta)}{\alpha_2(\eps_2-\beta\eps_1)^2} \right).
\]
\end{theorem}

\begin{proof}
On each $\ket{\psi}^{\otimes k}$ the tester returns a Bernoulli variable with mean $q(\psi)$, and we want to distinguish between $q(\psi) \leq \alpha_1 \eps_1^2$ and $q(\psi) \geq \alpha_2 \eps_2^2$.
The sample complexity follows by applying \cref{lem:bernoulli-sample-complexity}.
Note that an intolerant tester corresponds to $\eps_1 = 0$.
\end{proof}

At this point, the class-by-class analysis has been reduced to bounding the spectral gap $\Delta$ of the respective $\Xi$ operators (with the exception of coherent states, see \cref{sec:coherent-state-testing}). If it can be shown that $\Delta$ is a constant strictly above $\frac{k(k-1)}{m(m-1)}$, then by \cref{eq:framework-soundness} we have that
\begin{equation}\label{eq:alpha2_choice}
    \alpha_2 = \min\braces*{\frac{1}{2}, \frac{m(m-1)\Delta - k(k-1)}{(m - k)(m+k-1)}}
\end{equation}
is a positive constant. 
A common choice in our analyses takes $k=2$ and $m=3$, which simplifies the term above to
\begin{equation}
\min\left\{\frac{1}{2}, \frac{3\Delta - 1}{2}\right\}.
\end{equation}
Numerically, the threshold for $\Delta$ is $1/3$ in this case.
The only example that differs is the class of general bosonic Gaussian states, wherein we take $k = 3$ and $m = 7$. In that case, $\alpha_2$ is chosen to be
\begin{equation}
\min\left\{\frac{1}{2}, \frac{7\Delta - 1}{6}\right\},
\end{equation}
and the numerical threshold for $\Delta$ is $1/7$.
Note that the other constant $\alpha_1$ can always be taken to be $k$.

\section{Testing fermionic Gaussian states}\label{sec:fermionic-gaussian-tester}

The first class we study are fermionic Gaussian states, and we obtain the following result for property testing this class:
\begin{theorem}\label{thm:testing-fermionic-gaussians}
    Let $\beta = 4$. For all $0 \leq \eps_1 < \eps_2 < 1$ with $\eps_2 > \beta \eps_1$ and $0 < \delta < \frac{1}{2}$, there is an $O\paren*{\frac{\log(1/\delta)}{(\eps_2 - \beta\eps_1)^2}}$-copy $(\eps_1, \eps_2, \delta)$-tester for fermionic Gaussian states. The tester acts on two copies at a time.
\end{theorem}

At the core of our property tester is the following two-copy operator introduced by Tsirel'son \cite{tsirel1987quantum}:
\begin{equation}
    \Lambda = \sum_{i = 1}^n \left( c_i \otimes c_i^\dagger + c_i^\dagger \otimes c_i \right).
\end{equation}
Bravyi showed that $[\Lambda, \rho \otimes \rho] = 0$ if and only if $\rho$ is a fermionic Gaussian state \cite{bravyi2004lagrangian}. A simple calculation shows that $\ket{\psi} \otimes \ket{\psi}$ is a pure fermionic Gaussian state if and only if it lies entirely in the kernel of $\Lambda$. Thus we define $\fermionicgaussiantestproj$ to be the two-copy projector onto the kernel of $\Lambda$. We will prove that this projector can tolerantly test for fermionic Gaussianity with dimension-independent sample complexity.

Perfect completeness is immediate due to \cite{bravyi2004lagrangian}. To prove soundness, we apply the framework developed in \cref{sec:framework}, ultimately providing a drastically sharpened analysis compared to prior state-of-the-art bounds for the same tester \cite{haug2026practical,tarabunga2026computable,leone2026fermionicentropy}.

First, we define some notation. 
For a tuple $\mathsf{t}$ of registers let $\symproj^{(\mathsf{t})}$ be the projector onto the  symmetric subspace indexed by $\mathsf{t}$. For example,  $\symproj^{(1,2,3)}$ is the projector onto the symmetric subspace of the first $3$ registers.
Furthermore, let $\Lambda^{(a,b)}$ be the $\Lambda$ applied to registers $a$ and $b$, let $\fermionicgaussiantestproj^{(a,b)}$ be the corresponding projector onto the $0$-eigenspace of $\Lambda^{(a,b)}$, and let $\SWAP^{(a,b)}$ be the swap operator between registers $a$ and $b$.

Now, following the approach outlined in \cref{sec:framework}, we consider the symmetrized operator
\begin{equation}
    \Xi = \symproj^{(1,2,3)}\cdot\fermionicgaussiantestproj^{(1,2)}\cdot\symproj^{(1,2,3)}.
\end{equation}
The soundness of our tester will follow from the spectral gap of $\Xi$.
\begin{theorem}\label{lem:fermionic-spectral-gap}
    Let $\lambda_1$ be the largest eigenvalue of $\Xi$ and let $\lambda$ be any other (distinct) eigenvalue of $\Xi$. Then $\lambda_1 = 1$ and $\lambda \in [0, 7/12]$.
\end{theorem}

This is sufficient for the gradient flow argument to apply, since the gap $\Delta = 5/12 > 1/3$. We give two proofs of \cref{lem:fermionic-spectral-gap}.
The first, presented in \cref{sec:spectral-gap-concrete-proof}, uses linear-algebraic arguments and gives an explicit eigenbasis for the projector. It also makes a connection to intersection matrices that may be of independent interest.

The second, presented in \cref{sec:spectral-gap-abstract-proof}, uses the representation theory of fermionic Gaussian states to analyze the spectrum of $\Xi$. We give this alternative proof for unifying purposes: the other classes below, besides coherent states, also use a representation-theoretic approach.

\subsection{Linear-algebraic analysis of the spectral gap}\label{sec:spectral-gap-concrete-proof}

To analyze the eigenspectrum of $\Xi$, we can define the operator
\[
 \Theta = \fermionicgaussiantestproj^{(1,2)}\cdot\symproj^{(1,2,3)}\cdot\fermionicgaussiantestproj^{(1,2)}.
\]
Observe that for $M = \symproj^{(1,2,3)}\fermionicgaussiantestproj^{(1,2)}$, we have $\Xi = MM^\dagger$ and $\Theta = M^\dagger M$. Thus $\Xi$ and $\Theta$ possess identical eigenspectra. As it turns out, it will be easier to analyze the eigenspectrum of $\Theta$, as we can put it in a more tractable form.

\begin{proposition}\label{prop:Theta_identity}
On the symmetric subspace, $\Theta$ has the following form:
\begin{equation}
\Theta = \frac{\fermionicgaussiantestproj^{(1,2)} + 2 \fermionicgaussiantestproj^{(1,2)}\cdot \SWAP^{(2,3)} \cdot \fermionicgaussiantestproj^{(1,2)}}{3}
\end{equation}
\end{proposition}

\begin{proof}
    Observe that we can write
    \[
    \symproj^{(1,2,3)} = \frac{1}{3}\left(\identity + \SWAP^{(2,3)} + \SWAP^{(1,3)}\right)\symproj^{(1,2)}.
    \]
    We can directly compute
    \begin{align}
        \Theta &= \fermionicgaussiantestproj^{(1,2)}\symproj^{(1,2,3)}\fermionicgaussiantestproj^{(1,2)} \\
        &= \frac{1}{3}\left(\left(\fermionicgaussiantestproj^{(1,2)}\right)^3 + \fermionicgaussiantestproj^{(1,2)}\SWAP^{(2,3)} \fermionicgaussiantestproj^{(1,2)} + \fermionicgaussiantestproj^{(1,2)}\SWAP^{(1,3)} \fermionicgaussiantestproj^{(1,2)}\right) \\&= \frac{1} {3}\left(\fermionicgaussiantestproj^{(1,2)}+ \fermionicgaussiantestproj^{(1,2)}\SWAP^{(2,3)} \fermionicgaussiantestproj^{(1,2)} + \fermionicgaussiantestproj^{(1,2)}\SWAP^{(1,3)} \fermionicgaussiantestproj^{(1,2)}\right) \label{eq:expanded-theta-form}
    \end{align}
    Since $\fermionicgaussiantestproj^{(1,2)}$ is symmetric between the first and second register, we have 
    \[
    \fermionicgaussiantestproj^{(1,2)} \SWAP^{(1,2)} = \SWAP^{(1,2)}\fermionicgaussiantestproj^{(1,2)} = \fermionicgaussiantestproj^{(1,2)}.
    \]
    Furthermore, since $\SWAP^{(1,3)} = \SWAP^{(1,2)}\SWAP^{(2,3)}\SWAP^{(1,2)}$, we have that
    \begin{equation}
    \begin{split}
        \fermionicgaussiantestproj^{(1,2)}\SWAP^{(1,3)} \fermionicgaussiantestproj^{(1,2)} &= \fermionicgaussiantestproj^{(1,2)}\SWAP^{(1,2)}\SWAP^{(2,3)}\SWAP^{(1,2)} \fermionicgaussiantestproj^{(1,2)}\\
        &= \fermionicgaussiantestproj^{(1,2)}\SWAP^{(2,3)} \fermionicgaussiantestproj^{(1,2)}.
    \end{split}
    \end{equation}
    Applying this identity to \eqref{eq:expanded-theta-form} gives us our desired result.
\end{proof}

We will directly compute the eigenspectrum of the term $P = \fermionicgaussiantestproj^{(1,2)}\cdot \SWAP^{(2,3)} \cdot \fermionicgaussiantestproj^{(1,2)}$. To do so, we will fix disjoint $A,B \subseteq [n]$ and choose $C \subseteq B$. Let $B=\{b_1 < b_2 < \cdots < b_m\}$.  Define the following state on the first and second registers:
\begin{align*}
\ket{\chi(A,B,C)}_{1,2} = 2^{-|B|/2}\sum_{T \subseteq B} (-1)^{|T \cap C|+q(T)} \ket{A \cup T}_1\ket{A \cup (B \setminus T)}_2\\
=:2^{-|B|/2}\sum_{T \subseteq B} (-1)^{|T \cap C|+q_B(T)}\ket{T}_{1,2},
\end{align*}
where $q_B(T)=\sum_{r:b_r \in T} (r-1)$.  Let us first note that 
\[
(c_{b_r}^\dagger \otimes c_{b_r}+c_{b_r} \otimes c_{b_r}^\dagger)\ket{T}_{1,2}=(-1)^r \ket{T\Delta \{b_r\}}_{1,2},
\]
and hence
\begin{equation}
    \Lambda_{1,2} \ket{T}_{1,2}=2\sum_{r=1}^{m}(-1)^{r-1}\ket{T\Delta \{b_r\}}_{1,2}.
\end{equation}
We now show that these are eigenvectors of $\Lambda$.

\begin{proposition}\label{prop:lambda-eigendecomp}
The following facts hold about the $\ket{\chi(A,B,C)}$:
\begin{enumerate}
    \item Suppose $(A,B,C) \neq (A^\prime, B^\prime, C^\prime)$ and $C \subseteq B, C^\prime \subseteq B^\prime$. Then $\braket{\chi(A,B,C)|\chi(A^\prime,B^\prime,C^\prime)} = 0$.
    \item $\Lambda^{(1,2)} \ket{\chi(A,B,C)} = 2(|B| - 2|C|)\ket{\chi(A,B,C)}$
    \item  The set of states $\ket{\chi(A,B,C)}$ such that $|B|$ is even and $2|C|=|B|$ form a complete set of eigenstates for the null space of $\Lambda$
\end{enumerate}
\end{proposition}
\begin{proof}
    For (1), we first cover the case that $(A,B) \neq (A^\prime, B^\prime)$. If this is the case then the first and second registers of every term of $\ket{\chi(A,B,C)}$ and $\ket{\chi(A^\prime,B^\prime,C^\prime)}$ have either a different intersection or a different symmetric difference. Thus none of them can be equal and all bra-ket terms are annihilated. We now turn to the case where $(A,B) = (A^\prime, B^\prime)$ but $C \neq C^\prime$. In this case we have
    \begin{align}
    \braket{\chi(A,B,C)|\chi(A,B,C^\prime)} &= 2^{-|B|}\sum_{T,T^\prime \subseteq B} (-1)^{\abs{T \cap C} + \abs{T^\prime \cap C^\prime}} \braket{A \cup T|A \cup T^\prime} \braket{A \cup (B \setminus T)|A \cup (B \setminus T^\prime)} \\
    &= 2^{-|B|}\sum_{T\subseteq B} (-1)^{\abs{T \cap C} + \abs{T \cap C^\prime}} \\
    &= 2^{-|B|}\sum_{T\subseteq B} (-1)^{\abs{T \cap (C \triangle C^\prime)}} = 0,
    \end{align}
    where we have used the fact that
    \(
    \sum_{T \subseteq B} (-1)^{\abs{T \cap S}} = 0
    \)
    for any nonempty $S \subseteq B$, and $C \triangle C^\prime \neq \emptyset$ when $C \neq C^\prime$.
    
    Next, we show (2). Again let Let $B=\{b_1 < b_2 < \cdots < b_m\}$.  We can write
    \begin{align}
    \Lambda^{(1,2)} \ket{\chi(A,B,C)} &= 2^{-|B|/2}\sum_{T \subseteq B} (-1)^{|T \cap C|+q_B(T)} \Lambda^{(1,2)}\ket{T}_{1,2} \\
    &= 2^{1-|B|/2}\sum_{T \subseteq B} (-1)^{|T \cap C|+q_B(T)} \sum_{r=1}^m(-1)^{r-1}\ket{T\Delta b_r}_{1,2}.
    \end{align}
    If we define $U=T\Delta b_r$, then note that $q_B(T)=q_B(U)-(r-1) \mod 2$ and that $|T\cap C|=|(U\Delta b_r)\cap C|=|U\cap C|+|b_r \cap C|$.  Hence we can simplify $\Lambda^{(1,2)} \ket{\chi(A,B,C)} $ as
    \begin{align}
        \Lambda^{(1,2)} \ket{\chi(A,B,C)} &= 2^{1-|B|/2}\sum_{\substack{T \subseteq B\\ r}} (-1)^{|T \cap C|+q_B(T)+(r-1)} \ket{T\Delta b_r}_{1,2}\\
        &= 2^{1-|B|/2}\sum_{\substack{T \subseteq B\\ r}} (-1)^{|U\cap C|+|b_r \cap C|+q_B(U)} \ket{U}_{1,2},
    \end{align}
    where we are considering $U$ as a function of $T$ and $r$.  Note that each entry in the sum corresponding to $U$ and $r$ uniquely corresponds to some $T$, so we may rewrite the sum as:
    \begin{align}
        \Lambda^{(1,2)} \ket{\chi(A,B,C)} &= 2^{1-|B|/2}\sum_{\substack{U \subseteq B\\ r}} (-1)^{|U\cap C|+|b_r \cap C|+q_B(U)} \ket{U}_{1,2}\\
        &=2\sum_r (-1)^{|b_r \cap C|} \ket{\chi(A, B,  C)}\\
        &=2(m-2|C|)\ket{\chi(A, B,  C)}
    \end{align}

    For (3) note that $\majorana_i \otimes \majorana_i$ commutes with $\majorana_j \otimes \majorana_j$ for all $i,j$.  Hence, we can simultaneously diagonalize all terms in the summation of $\Lambda$.  The null space corresponds exactly to the set of states which are anti-stabilized by $n$ of the $\majorana_j \otimes \majorana_j$ and stabilized by $n$ of $\majorana_j \otimes \majorana_j$.  There are $\binom{2n}{n}$ such states since distinct $\majorana_i \otimes \majorana_i$ are algebraically independent.  We have seen that $\ket{\chi(A,B,C)}$ with $|B|$ even and $2|C|=|B|$ are in the null space and that distinct such $\ket{\chi(A,B,C)}$ are orthogonal.  We can further note that there are $\binom{2n}{n}$ tuples $(A, B, C)$ satisfying the assumption\footnote{See \cite{sun2019new} near the bottom of the second page.}, hence these vectors must span the null space.  
\end{proof}

Before proving \cref{lem:fermionic-spectral-gap}, we first show a useful combinatorial statement.

\begin{proposition}\label{prop:combinatorial-identity}
    Let $k$ be an even integer. Then
    \[
    \sum_{i=j}^{k} \binom{k - j}{i - j}\binom{2k-j-i}{k-i} (-2)^j = \begin{cases}
        (-2)^j(-1)^{(k-j)/2}\binom{k-j}{\frac{k-j}{2}} & k-j \textrm{ even} \\
        0  & k-j \textrm{ odd}
    \end{cases}.
    \]
\end{proposition}
\begin{proof}
    We introduce a change of variables $\ell = k-j, i = j + r$. Then
    \begin{align}
        \sum_{i=j}^{k} \binom{k - j}{i - j}\binom{2k-j-i}{k-i} (-2)^j &= (-2)^j \sum_{r=0}^{\ell} \binom{\ell}{r} \binom{2\ell - r}{\ell-r}(-2)^r \\
        &= (-2)^j[x^\ell]\sum_{r=0}^\ell \binom{\ell}{r}(-2)^r(1 + x)^{2\ell-r} \\
        &=(-2)^j[x^\ell](1 + x)^{2\ell}\left(1 - \frac{2}{1+x}\right)^\ell \\
        &= (-2)^j[x^\ell](1+x)^\ell(x-1)^\ell \\
        &= (-2)^j[x^\ell](x^2-1)^\ell.
    \end{align}

    Now, for any polynomial $p$, we have that
    \[
    [x^\ell]p(x) = \frac{1}{\ell!}\left.\dfrac{d^\ell}{dx^\ell}p(x)\right|_{x = 0}.
    \]
    The Rodrigues formula for Legendre polynomials tells us that
    \[
    P_\ell(z) = \frac{1}{2^\ell \ell!} \dfrac{d^\ell}{dx^\ell}(x^2-1)^\ell\,
    \]
    so $[x^\ell](x^2-1)^\ell = 2^{\ell}P_\ell(0)$. We conclude via standard evaluations of the Legendre polynomials:
    \[
    \sum_{i=j}^{k} \binom{k - j}{i - j}\binom{2k-j-i}{k-i} (-2)^j = (-2)^j2^{\ell}P_\ell(0) = (-2)^j(-1)^{(k-j)/2}\binom{k-j}{\frac{k-j}{2}}. \qedhere
    \]
\end{proof}

We have now assembled all the tools necessary to prove the spectral gap of $\Xi$.

\begin{proof}[Proof of \cref{lem:fermionic-spectral-gap}]
    A key observation is that that $P$ preserves the total occupation number of each mode $i \in [n]$ across the registers.
    Concretely, choose any $w \in \{0,1,2,3\}^n$ and write 
    \[
    V_w = \linspan\{\ket{S_1,S_2,S_3}: \forall i, \ \indic{i \in S_1} + \indic{i \in S_2} + \indic{i \in S_3} = w_i\}.
    \]
    Then $P V_w = V_w$, and to compute the eigenvalues of $P$, it suffices to compute the eigenvalues of $P_w$ for all $w$.
    
    Consider a total occupation vector $w \in \{0,1,2,3\}^n$ and an index $i$ such that $w_i = 3$. Because we can decompose $\Lambda = \sum_{i=1}^n (c_i \otimes c_i^\dagger + c_i^\dagger \otimes c_i)$, such a vector is annihilated by $c_i \otimes c_i^\dagger + c_i^\dagger \otimes c_i$ and is unaffected by $c_j \otimes c_j^\dagger + c_j^\dagger \otimes c_j$ where $j \neq i$.
    In other words, it lies in a subprojector of $\fermionicgaussiantestproj^{(1,2)}:$
    \[
    \sum_{A \ni i}\sum_{B \cap A = \emptyset} \sum_{\substack{C \subseteq B \\ 2|C| = |B|}} \ketbra{\chi(A,B,C)}{\chi(A,B,C)}
    \]
    Furthermore, it is clearly unchanged by the $\SWAP$ operator. By a similar argument, the same is true of indices $i$ such that $w_i = 0$.
    Thus the eigenvalues of $P_w$ depend on the number of $i \in [n]$ such that $w_i = 1$ and $w_i = 2$, which we denote as $k_1$ and $k_2$, respectively.

    Next, we argue a symmetry between the cases $v_i = 1$ and $v_i = 2$. Consider the unitary $J_i = \majorana_{2i-1}^{\otimes 3}$. A simple computation shows that each $J_i$ commutes with $\Lambda^{(1,2)}$ and $\SWAP^{(2,3)}$, so it commutes with $P$. For $w \in \{0,1,2,3\}$, $J_i$ converts a vector in $V_w$ to a vector in $V_{w^\prime}$, where $w^\prime = (w_1,\ldots, 3 - w_i,\ldots, w_n)$. Therefore we can reduce to the case where all $i$ such that $w_i \in \{1,2\}$ satisfy $w_i = 1$, as the eigenspectra of these states is identical. For all $k \in [n]$, it suffices to compute the eigenspectrum of $P_{w_k}$, where $w_k = (1^k,0^{n-k})$.

    Now since $w_i \leq 1$ for all $i$, $V_w$ is annihilated by any rank-$1$ projector $\ketbra{\chi(A,B,C)}{\chi(A,B,C)}$ such that $A \neq \emptyset$. The nonzero eigenspace of $P$ is spanned by
    \[
    \ket{\xi(B,C)} = 2^{-|B|/2}\sum_{T \subseteq B} (-1)^{|T \cap C|+q_B(T)} \ket{T} \ket{B \setminus T}\ket{[k] \setminus B},
    \]
    where we require further that $B \subseteq [k]$.
    Let us compute the matrix elements of $\SWAP^{(2,3)}$ in this basis. For the element $\braket{\xi(B^\prime,C^\prime)|\SWAP^{(2,3)}|\xi(B,C)}$ we obtain
    \begin{align*}
     &2^{-\frac{|B|+|B^\prime|}{2}}\sum_{\substack{T \subseteq B \\ T^\prime \subseteq B^\prime}} (-1)^{|T \cap C| + |T^\prime \cap C^\prime|+q_B(T)+q_{B'}(T')} \indic{T = T^\prime} \cdot \indic{B \setminus T = [k] \setminus B^\prime} \cdot \indic{B^\prime \setminus T^\prime = [k] \setminus B} \\ 
     &= 2^{-\frac{|B|+|B^\prime|}{2}}\sum_{\substack{T \subseteq B \cap B^\prime}} (-1)^{|T \cap (C\triangle C^\prime)|+q_B(T)+q_{B'}(T')} \cdot \indic{B \setminus T = [k] \setminus B^\prime} \cdot \indic{B^\prime \setminus T = [k] \setminus B}
    \end{align*}
    The two indicators above pose the conditions that $B \cup B^\prime = [k]$ and that $T = B \cap B^\prime$. Indeed, if there was an $i \in [k] \setminus (B \cup B^\prime)$ then $i \in [k] \setminus B$ but $i \not \in B \setminus T$. Similarly, if there was an $i \in B \cap B^\prime$ but $i \not \in T$ then $i \in B \setminus T$ but $i \not \in [k] \setminus B^\prime$. Thus the matrix element of $P_{w_k}$ corresponding to $(B,C,B^\prime,C^\prime)$ is simply
    \begin{equation}
        \indic{B \cup B^\prime = [k]}\cdot 2^{-\frac{|B|+|B^\prime|}{2}}\cdot (-1)^{\abs{B^\prime \cap C} + \abs{B \cap C^\prime}+q_B(B\cap B')+q_{B'}(B\cap B')}
    \end{equation}
    Now let $\sigma(S)$ denote the number of even $i \in S$ modulo $2$ and define the diagonal unitary $D$ such that $D\ket{S} = (-1)^{\sigma(S)}\ket{S}$. We observe that 
        \[
        (D \otimes D \otimes I)\ket{\xi(B,C)} = (-1)^{\sigma(B)}\ket{\xi(B,C)}.
        \]
    Using the facts that $B \cup B^\prime = [k]$, $|B| = 2|C|$, and $|B^\prime| = 2|C^\prime|$, we can compute
    \[
    q_B(B \cap B^\prime) + q_{B^\prime}(B \cap B^\prime) + \sigma(B) + \sigma(B^\prime) \equiv k + |C| + |C^\prime| \mod 2.
    \]
        Thus conjugating $P_{w_k}$ by $D \otimes D \otimes \identity$ gives a matrix whose entries in the $\xi(B,C)$ basis are
        \begin{equation}\label{eq:block-diagonalized-P-matrix}
        \indic{B \cup B^\prime = [k]}\cdot 2^{-\frac{|B|+|B^\prime|}{2}}\cdot (-1)^{\abs{C \setminus B^\prime} + \abs{C^\prime\setminus B} + k} 
        \end{equation}
        
        All that remains is to compute the eigenvalues of this matrix. To do this, we will write $(B,C)$ in a useful coordinatewise form. For each $i \in [n]$, let
    \[
    x_i^{(B,C)} = \begin{cases}
        -1 & i \in B \setminus C \\
        0 & i \not \in B \\
        1 & i \in C
    \end{cases}
    \]
    We can then write $(B,C)$ as a vector $\ket{x^{(B,C)}}$ such that $\sum_{i=1}^b x_i^{(B,C)} = 0$. We will further replace the ternary variables $x_i^{(B,C)}$ via the following isometry $M$:
    \[
    M\ket{-1} = \ket{00}, \quad M\ket{0} = \frac{\ket{01} + \ket{10}}{\sqrt{2}}, \quad M\ket{1} = \ket{11}
    \]
    The benefit of this isometry is that $M^{\otimes k}\ket{x^{(B,C)}}$ is a substring of $\F_2^{2k}$ of Hamming weight exactly $k$, by our requirement that $x^{(B.C)}$ is balanced. Let $\Pi_k$ be the projection onto this subspace, and let $\Pi_{2}$ be the projector onto the symmetric subspace of two qubit states. A calculation demonstrates that we can write the action of $P_{w_k}$ as
    \begin{equation}
    Q_{k} = (-1)^k\Pi_2^{\otimes k} \Pi_{k} H^{\otimes 2k} \Pi_{k} \Pi_2^{\otimes k},
    \end{equation}
    where $H$ is the standard Hadamard matrix. For subsets $S, T \subseteq [2k]$ of size $k$, we have that
    \begin{equation}\label{eq:intersection-form}
    \braket{S|Q_k|T} = (-1)^k2^{-k} (-1)^{S \cap T}. 
    \end{equation}
    $Q_k$ thus has the form of an intersection matrix \cite{ghareghani2011intersectionmatrices}. Concretely, \cite{ghareghani2011intersectionmatrices} define a set of matrices $F_{k,k,k}(z)$ over subsets of $[2k]$ of size $k$ as\footnote{In their notation, we swap $K \to T$ and have both $s = |S| = k$ and $|T| = k$.}
    \[
    (F_{k,k,k}(z))_{S,T} = \sum_{i=0}^k \binom{\abs{S \cap T}}{i}z^i = (1 + z)^{\abs{S \cap T}}.
    \]
    Choosing $z = -2$ we obtain exactly the matrix in \cref{eq:intersection-form}, without the factor of $(-1)^k2^{-k}$. By \cite[Theorem 9]{ghareghani2011intersectionmatrices}, the eigenvalues of $F_{k,k,k}(-2)$ are given by
    \[
    \mu_{j}(-2) = \sum_{i=j}^k \binom{k-j}{i-j}\binom{2k - j - i}{k-i}(-2)^{i} = \begin{cases}
        (-2)^j(-1)^{(k-j)/2}\binom{k-j}{\frac{k-j}{2}} & k-j \textrm{ even} \\
        0  & k-j \textrm{ odd}
    \end{cases},
    \]
    where the second equality invokes \cref{prop:combinatorial-identity}.
    For even $k-j$, we write $2r = k - j$ and compute the eigenvalues of $Q_k$ as
    \[
    \nu_r = (-1)^k 2^{-k} \mu_{k-2r}(-2) = (-1)^k 2^{-k} (-2)^{k-2r}(-1)^r \binom{2r}{r} = (-1)^r \frac{\binom{2r}{r}}{4^r}.
    \]
    For $r = 0$, we obtain the eigenvalue $1$. It remains to show that for $r \geq 1$, $\nu_r \in [-1/2,3/8]$. Observe that
    \[
    \abs*{\frac{\nu_{r+1}}{\nu_{r}}} = \frac{\binom{2r+2}{r+1}/4^{r+1}}{\binom{2r}{r}/4^r} =  \frac{(2r+2)(2r+1)}{4(r+1)^2} = \frac{2r+1}{2r+2} < 1.
    \]
    We conclude by computing $\nu_1 = -\frac{1}{2}$ and $\nu_2 = \frac{3}{8}$. \qedhere
\end{proof}

\subsection{Representation-theoretic analysis of the spectral gap}\label{sec:spectral-gap-abstract-proof}

Here we present an alternative proof of \cref{lem:fermionic-spectral-gap} using the representation theory of (three copies of) the spin group $\spingrp(2n)$. Define
\begin{align}
\Pi_{\mathrm{Sym}^{k}} &= \frac{1}{k!} \sum_{\sigma} U_{\sigma},\\
d^{(k)}_n &= \dim V^{\mathrm{Spin}(2n)}_{{\left(\frac{k}{2}\right)}^n},\\
d^{(k)}_{n, i} &= \dim V^{\mathrm{Spin}(2n)}_{
    {\left(\frac{k+1}{2}\right)}^{n-i}
    {\left(\frac{k-1}{2}\right)}^{i}
}.
\end{align}
Note that
\begin{align}
d^{(2)}_n &= d^{(1)}_{n, 0} =
\dim V^{\mathrm{Spin}(2n)}_{{\left(1\right)}^n}
=\frac12 \binom{2n}{n}
,\\
d^{(2)}_{n, i} &=
\dim V^{\mathrm{Spin}(2n)}_{
    {\left(\frac32\right)}^{n-i}
    {\left(\frac12\right)}^{i}
}
=\binom{2n}{n-i} - \binom{2n}{n-i-1}
=
{\left[t^n\right]} C_n {\left(t C_n^2\right)}^i.
\end{align}

The basic idea is as follows.
Let $P$ be the projector onto the top irrep $V$ of $k$ registers.
The decomposition of $V \otimes H$ (adding one more register) is multiplicity free, so on $H^{\otimes k}$
\begin{align}
    P \otimes I &= \sum_{\lambda} P_{\lambda},
\end{align}
where $P_{\lambda}$ acts on the irrep space for $\lambda$.
For different $n$, certain $\lambda$ appear or not, but when they appear $P_{\lambda}$ is the same.
Furthermore, $P_{\lambda}$ acts non-trivially only on the commutant irrep space.
Given how permutations of the registers acts on the commutant irrep space, we can calculate the symmetrized version of $P_{\lambda}$.
Taking traces and using the known dimensions of the irrep spaces, we get a linear system of equations whose solution yields the eigenvalues.
\cref{lem:spectral-sequence-conditions} shows how to derive the system of equations, and \cref{lem:spectral-sequence-solution} gives their solution.
Following their statement and proof, we conclude with the alternative proof of \cref{lem:fermionic-spectral-gap}.

\begin{lemma}\label{lem:spectral-sequence-conditions}
There is a sequence $\sigma_0, \sigma_1, \ldots \in \mathbb C$ such that the following holds.
Let $n, k$ be positive integers.
Let $P^{(n)}$ be the projector onto the top irrep of $\mathrm{Spin}(2n)$ in $S_+^{\otimes k}$.
Define the operator
\begin{align}
    A^{(n)} &= 
\Pi^{(n)}_{\mathrm{Sym}^{k+1}}
P^{(n)}_{[k]}
\Pi^{(n)}_{\mathrm{Sym}^{k+1}}
\end{align}
on $S_+^{\otimes (k + 1)}$, where $P_{[k]} = P \otimes I_{S_+}$.
Then the eigenvalues of $A$ are
\begin{align}
    \{0, 1\} \cup \left\{\sigma_i: i \in \{1, 2, \ldots, \lfloor n/2\rfloor\}\right\},
\end{align}
and the sequence $\sigma_1, \sigma_2, \ldots$ satisfies
\begin{align}
\sum_{i=0}^{\lfloor n/ 2 \rfloor} \sigma_i 
d^{(k)}_{n, 2i} &= d^{(k)}_{n} \frac{2^{n-1} + k}{1 + k}.
\label{eq:spectral-sequence-conditions}
\end{align}
\end{lemma}
\begin{proof}

Let
\begin{align}
J &= 
P_{[k]}
\Pi_{\mathrm{Sym}^{k+1}},\\
B &= 
J J^{\dagger}
=
P_{[k]}
\Pi_{\mathrm{Sym}^{k+1}}
P_{[k]};
\end{align}
note that
\begin{align}
A &= 
J^{\dagger} J
=
\Pi_{\mathrm{Sym}^{k+1}}
P_{[k]}
\Pi_{\mathrm{Sym}^{k+1}}.
\end{align}
We can write $J$ as 
\begin{align}
J &= \sum_{\lambda, \mu} J_{\lambda, \mu},
\end{align}
where $J_{\lambda, \mu}$ is the block of $J$ between the $V_{\lambda}$ subspace of $V_{\left(\frac{k}{2}, \ldots, \frac{k}{2}\right)} \otimes S_+$ and the $V_{\mu} \otimes M_{\mu}$ subspace of $\mathrm{Sym}^{k+1} S_+$ (see \cref{fct:schur-decompositions}.)
Because $J$ commutes with the action of $\mathrm{Spin}(2n)$,
\begin{align}
J {U(g)}^{\otimes(k+1)} = J {U(g)}^{\otimes(k+1)},
\end{align}
by Schur's lemma, $J_{\lambda, \mu}$ is proportional to the identity when $\lambda = \mu$ and 0 otherwise:
\begin{align}
J &\cong \sum_{\lambda} \sqrt{\sigma_{\lambda}} I_{V_{\lambda}} \otimes \bra{m_{\lambda}},
\end{align}
where $\sigma_{\lambda} \geq 0$.
For concision, let $V_l = V_{
    {\left(\frac{k + 1}{2}\right)}^{n-2l}
    {\left(\frac{k + 1}{2}\right)}^{2l}
}$, which covers all of the irreps appearing in $V_{{\left(\frac{k}{2}\right)}^n} \otimes S_+$.
Let $R_{\lambda}$ be the partial isometry between 
$V_{\lambda} \subset V_{{\left(\frac{k}{2}\right)}^n} \otimes S_+$ and 
$V_{\lambda} \subset \mathrm{Sym}^{k+1}S_+$.
Then
\begin{align}
    J &= \sum_{l=0}^{\lfloor n / 2\rfloor} \sqrt{\sigma_l} R_l,\\
    A &= J^{\dagger} J = \sum_{l=0}^{\lfloor n / 2\rfloor} \sigma_l 
    R_{\lambda}^{\dagger} R_{\lambda},\\
    B &= J J^{\dagger} = \sum_{l=0}^{\lfloor n / 2\rfloor} \sigma_l 
    R_{\lambda} R_{\lambda}^{\dagger}.
\end{align}
That is, $A$ and $B$ have the same (real, nonnegative) spectrum, with the same multiplicity for each nonzero eigenvalue.
Their traces are
\begin{align}
\sum_{l=0}^{\lfloor n / 2 \rfloor} \sigma_l d^{(k)}_{n, 2l}
&=
\sum_{l=0}^{\lfloor n / 2\rfloor} \sigma_l \dim V_l
=
\tr \, B \\
&=\tr \, A = 
\frac{1}{(k+1)!^2}
\sum_{\sigma, \sigma'} \tr\left[\sigma P_{[k]} \sigma'\right]
=
\frac{1}{(k+1)!}
\sum_{\sigma} \tr\left[\sigma P_{[k]}\right]
\\
&=
\frac{k!}{(k+1)!}
\tr P
\left[
    \tr \, \identity_{S_+}
    + k
\right]
=
\frac{
2^{n-1} + k
}{k+1}
\dim V_{\frac{k}{2}, \ldots, \frac{k}{2}}
=
d^{(k)}_n
\frac{
2^{n-1} + k
}{k+1},
\end{align}
i.e., \cref{eq:spectral-sequence-conditions}.
It remains to be proven that the sequence $\sigma_{i}^{(n)}$ is independent of $n$.

Let
\begin{align}
\Omega^{(m)}_{l} &= {\left(\ket{0}\!\bra{0}\right)}^{\otimes (m \times l)},\\
\tilde{\Omega}^{(m)}_l
                 &=
                 {\left[
 {\left(\ket{0}\!\bra{0}\right)}^{\otimes l}
 \otimes 
\identity_{2^{n-l}}
\right]}^{\otimes m}.
\end{align}
Note that
\begin{align}
\tilde{\Omega}_{k+1}^{(l)} 
P^{(n)}_{[k]}
&
=
P^{(n)}_{[k]}
\tilde{\Omega}_{k+1}^{(l)} ,\\
\tilde{\Omega}_{k+1}^{(l)} 
\Pi^{(n)}_{\mathrm{Sym}^{k+1}}
&
=
\Pi^{(n)}_{\mathrm{Sym}^{k+1}}
\tilde{\Omega}_{k+1}^{(l)}.
\end{align}
Furthermore,
\begin{align}
P^{(2l)}_{[k]} 
\otimes \Omega
&=
                 \tilde{\Omega} 
                 P^{(n)}_{[k]} 
                 \tilde{\Omega} ,\\
                 \Pi^{(2l)}_{\mathrm{Sym}^{k+1}}
\otimes \Omega
&=
                 \tilde{\Omega} 
                 \Pi^{(2l)}_{\mathrm{Sym}^{k+1}}
                 \tilde{\Omega}.
\end{align}
Therefore,
\begin{align}
\tilde{B}^{(2l)} &= B^{(2l)} \otimes \Omega^{(n-2l)}_{k+1}
\\
                 &=
                 \left(
                 P^{(2l)}_{[k]} 
                 \Pi^{(2l)}_{\mathrm{Sym}^{k+1}}
                 P^{(2l)}_{[k]} 
             \right) \otimes \Omega^{(n-2l)}_{k+1}\\
                 &=
                 \left(
                 P^{(2l)}_{[k]} 
                 \otimes \Omega
                 \right)
                 \left(
                 \Pi^{(2l)}_{\mathrm{Sym}^{k+1}}
                 \otimes \Omega
                 \right)
                 \left(
                 P^{(2l)}_{[k]} 
                 \otimes \Omega
                 \right)
                 \\
                 &=
                 \left(
                 \tilde{\Omega} 
                 P^{(n)}_{[k]} 
                 \tilde{\Omega} 
                 \right)
                 \left(
                 \tilde{\Omega} 
                 \Pi^{(n)}_{\mathrm{Sym}^{k+1}}
                 \tilde{\Omega} 
                 \right)
                 \left(
                 \tilde{\Omega} 
                 P^{(n)}_{[k]} 
                 \tilde{\Omega}
                 \right)
                 \\
                 &=
                 \tilde{\Omega} 
                 P^{(n)}_{[k]} 
                 \Pi^{(n)}_{\mathrm{Sym}^{k+1}}
                 P^{(n)}_{[k]} 
                 \tilde{\Omega}
                 =
                 \tilde{\Omega} B^{(n)} \tilde{\Omega}
\end{align}
Let $\ket{v^{(2l)}_i}$ be the highest-weight vector of $V^{\mathrm{Spin}(4l)}_{
    {\left(\frac{k+1}{2}\right)}^{2l-2i}
    {\left(\frac{k-1}{2}\right)}^{2i}
}$ and 
$\ket{v^{(n)}_i}$ be the highest-weight vector of $V^{\mathrm{Spin}(2n)}_{
    {\left(\frac{k+1}{2}\right)}^{2l-2i}
    {\left(\frac{k-1}{2}\right)}^{2i}
}$.
Set $\ket{\tilde{v}^{(2l)}_i} = 
\ket{0}^{\otimes (k+1)\times (n-2l)}
\otimes 
\ket{v^{(2l)}} 
$.
Consider $0 \leq i \leq l$.
Then, with respect to $\mathrm{Spin}(2n)$,  $\ket{\tilde{v}^{(2l)}}$ has weights 
${\left(\frac{k+1}{2}\right)}^{n-2i}{\left(\frac{k-1}{2}\right)}$.
Because $\ket{v^{(2l)}_i}$ is a highest-weight vector for $\mathrm{Spin}(2l)$, $\ket{\tilde{v}^{(2l)}_i}$ is annhilated by all of the positive roots of $\mathrm{Spin}(2n)$ involving modes with index greater than $n-2l$.
But it is also annhilated by any positive root with a first index at most $n-2l$, because the weights are already at their maximum value $\frac{k+1}{2}$ on $S_+^{(\otimes (k+1)}$.
Therefore $\ket{\tilde{v}^{(2l)}_i}$ is also a $\mathrm{Spin}(2n)$ highest-weight-vector, i.e., 
$\ket{\tilde{v}^{(2l)}_i} = \ket{v^{(n)}_i}$.
Thus
\begin{align}
\sigma_j^{(2l)}
&=
\sum_{i=0}^l
\sigma_i^{(2l)}
\braket{{v}^{(2l)}_j |
{R}_i^{(2l)} {\left({R}_i^{(2l)}\right)}^{\dagger}
| {v}^{(2l)}_j}
\\
&=
\sum_{i=0}^l
\sigma_i^{(2l)}
\braket{\tilde{v}^{(2l)}_j |
\tilde{R}_i^{(2l)} {\left(\tilde{R}_i^{(2l)}\right)}^{\dagger}
| \tilde{v}^{(2l)}_j}
\\
&=
\braket{\tilde{v}^{(2l)}_j |
    \tilde{B}^{(2l)}
| \tilde{v}^{(2l)}_j}
\\
&=
\braket{\tilde{v}^{(2l)}_j |
    \tilde{\Omega}
    B^{(n)}
    \tilde{\Omega}
| \tilde{v}^{(2l)}_j}
\\
&=
\braket{{v}^{(n)}_j |
    B^{(n)}
| {v}^{(n)}_j}
\\
&=\sigma^{(n)}_j.
\end{align}
That is,
\begin{align}
\sigma_j &= \sigma_j^{(2j)} = \sigma_{j}^{(2j+1)} = \cdots,
\end{align}
i.e., $\sigma_j$ is independent of $n$.

So the eigenvalues $\sigma_i$ for $i=1, 2, \ldots$ satisfy
\begin{align}
    \sum_{i=0}^{\lfloor n/2 \rfloor} \sigma_i \dim V_{
        {\left(\frac{k+1}{2}\right)}^{n-2i}
        {\left(\frac{k-1}{2}\right)}^{2i}
    }
    &=
    \frac{2^{n-1} + k}{k + 1} \dim V_{
        {\left(\frac{k}{2}\right)}^{n}
    }
\end{align}
for every integer $m \geq i$.
Applying \cref{lem:spectral-sequence-solution}
\begin{align*}
    \sum_{i=0}^m \sigma_i     
\binom{2n}{2i}
\frac{
    {\left(\frac{k}{2}\right)}^{\overline{i}}
    {\left(\frac{k + n + 2i}{2}\right)}^{\overline{m-i}}
}{
    {\left(\frac{k+2i - 1}{2}\right)}^{\overline{i}}
    {\left(\frac{k + 4i + 1}{2}\right)}^{\overline{m-i}}
}
    &=
    \frac{2^{2m-1} + k}{k + 1}. \qedhere
\end{align*}
\end{proof}

\begin{lemma}\label{lem:spectral-sequence-solution}
Let $k$ be a positive integer.
Suppose a sequence $\sigma_0, \sigma_1, \ldots \in \mathbb C$ satisfies
\begin{align}
\sum_{i=0}^{\lfloor n / 2 \rfloor} \sigma_i d^{(k)}_{n, 2i} &= d^{(k)}_{n} \frac{2^{n-1} + k}{1 + k}
\end{align}
for all integers $n \geq 2$.
Then
\begin{align}
\sigma_0 &= 1,\\
\sigma_1 &= 0,\\
\sigma_m &= \frac{
1 + k {(-1)}^m \prod_{j=0}^{m-1} \frac{1 + 2j}{k + 2j}}{k + 1}.
\end{align}

\end{lemma}

\begin{corollary}
\begin{align}
\max_{m\geq 2} \sigma_m &= \sigma_2 = \frac{k + 5}{(k+1)(k+2)} = \Theta(1/k).
\end{align}
\end{corollary}

\begin{proof}
    For $n = 2, 3$,
    \begin{align}
        \begin{pmatrix}
            d^{(k)}_{2, 0} & d^{(k)}_{2, 2} \\
            d^{(k)}_{3, 0} & d^{(k)}_{3, 2}
        \end{pmatrix}
        \begin{pmatrix}
            \sigma_0 \\ \sigma_1
        \end{pmatrix}
        &=
        \begin{pmatrix}
        d^{(k)}_2 \\
        d^{(k)}_3
        \end{pmatrix}
    \end{align}
    has the unique solution
    \begin{align}
        \begin{pmatrix}
            \sigma_0 \\ \sigma_1
        \end{pmatrix}
        &=
        \begin{pmatrix}
        1 \\
        0
        \end{pmatrix}.
    \end{align}
    Now define
    \begin{align}
        A &= {\left(A_{m, i}\right)}_{m, i \geq 1}, &
        A_{m, i} &= \begin{cases}
            \frac{d^{(k)}_{2m, 0}}{
            d^{(k)}_{2m, 2m}} & i = 1 < m,\\
            \frac{d^{(k)}_{2m, 2i}}{
            d^{(k)}_{2m, 2m}} & 1 < i < m,\\
            1, & i = m,\\
            0 & i > m,
        \end{cases}\\
            x &= {\left(x_m\right)}_{m\geq 1}, & x_m = 
            \begin{cases}
                1, & m = 1,\\
                \sigma_m, & m > 1,
            \end{cases}\\
            b' &= {\left(b'_m\right)}_{m\geq 1}, &
            b'_m &= 
            \begin{cases}
                1, & m = 1,\\
                \frac{d^{(k)}_{2m}}{d^{(k)}_{2m,2m}} , & m > 1,
            \end{cases}
            \\
            b'' &= {\left(b''_m\right)}_{m\geq 1}, &
            b''_m &= 2^{2m} b'_m,\\
            \\
            b &= {\left(b_m\right)}_{m\geq 1} = \frac{k b' +\frac12 b''}{k + 1}, &
            b_m &= \frac{d^{(k)}_{2m}}{d^{(k)}_{2m,2m}} \frac{2^{2m-1} + k}{1+k},
    \end{align}
    so that $Ax = b$ captures the initial conditions and the equations for $n=2, 4, \ldots,$ etc:
    \begin{align}
        {\left(Ax\right)}_1 &= b'_1 = 1,\\
        {\left(Ax\right)}_m &= \sum_{i=1}^m A_{m, i} x_i
        = 
        \frac{d^{(k)}_{2m, 0}}{d^{(k)}_{2m, 2m}} \cdot 1
        +
        \frac{d^{(k)}_{2m, 2}}{d^{(k)}_{2m, 2m}} \cdot 0
        +
        \sum_{i=2}^{m}
        \frac{d^{(k)}_{2m, 2i}}{d^{(k)}_{2m, 2m}} \sigma_i
        =
        \frac{1}{d^{(k)}_{2m, 2m}} 
        \sum_{i=0}^m d^{(k)}_{2m, 2i} \sigma_i
        \\
        &=
        \frac{1}{d^{(k)}_{2m, 2m}} 
        d^{(k)}_{2m} \frac{k + 2^{2m-1}}{k + 1}
        =
        b_m.
    \end{align}
    Let $A = I + \Delta$.
    Because $A$ is lower-triangular matrix with unit diagonal, $A$ is invertible.
    Then the inverse of $A$ is 
    \begin{align}
    A^{-1} &= \sum_{l=0}^{\infty} {(-1)}^l \Delta^l,
    \end{align}
    as confirmed by
    \begin{align}
    A A^{-1} &=
    A^{-1} + \Delta A^{-1}\\
             &=
\sum_{l=0}^{\infty} {(-1)}^l \Delta^l
+
\sum_{l=0}^{\infty} {(-1)}^l \Delta^{l+1}
\\
             &=
\sum_{l=0}^{\infty} {(-1)}^l \Delta^l
-
\sum_{l=1}^{\infty} {(-1)}^l \Delta^l
={(-1)}^0 \Delta^0 = I.
    \end{align}
    The solution to $Ax = b$ is
    \begin{align}
        x &= A^{-1}b =
        \frac{
            k x' + \frac12 x''
        }{k + 1}
        =
        \paren*{
        \frac{k + 2^{m-1}}{k + 1}
        }_{m=1}^{\infty}
    \end{align}
    where
    \begin{align}\label{eq:temp1}
        x' &= A^{-1} b' = \paren*{1}_{m=1}^{\infty},\\
        x'' &= A^{-1} b''  = \paren*{2^m}_{m=1}^{\infty}
    \end{align}
    are the solutions to $Ax'=b'$ and $Ax''=b''$, respectively.
    Because $A$ is invertible, the solution is unique and so it suffices to confirm that
    \begin{align}
      \paren*{A x'}_m
      &=
      b'_m,\\
      \paren*{A x''}_m
      &=
      b''_m.
    \end{align}
  
\end{proof}

\begin{proof}[Alternative proof of \cref{lem:fermionic-spectral-gap}]
For $k=2$ and $m>1$, 
\begin{align}
\sigma_m &=
\frac{
1 + {(-1)}^m  2^{1-2m} \binom{2m}{m}
}{3},
\end{align}
because
\begin{align}
\prod_{j=0}^{m-1} \frac{1 + 2i}{2 + 2i}
&=
\frac{(2m-1)!!}{2^m m!} = \frac{(2m)!}{2^{2m} m!} = 2^{-2m} \binom{2m}{m}.
\end{align}

Together \cref{lem:spectral-sequence-conditions,lem:spectral-sequence-solution} imply that for $k=2$ the eigenvalues are
$\braces{1} \cup \braces*{\sigma_m}_{m > 1}$.
The largest eigenvalue strictly less than 1 is therefore $\sigma_2 = 7/12$.
\end{proof}

\subsection{The algorithm}\label{sec:fermionic-gaussian-algorithm}

We are now ready to state and analyze our algorithm to property test fermionic Gaussian states. Recall that our tester is effectively the $\fermionicgaussiantestproj$, and that each measurement of this projector accepts with probability $\norm{\fermionicgaussiantestproj \ket{\psi^{\otimes 2}}}^2$. There is a simple Gaussian unitary that implements $\fermionicgaussiantestproj$. Indeed, define the unitary
\begin{equation}\label{eq:fermionic-lambda-diagonalizer}
    V = \prod_{i=1}^n \exp\left(\frac{\ri \pi}{4} \majorana_{2i} \otimes \majorana_{2i - 1}\right)
\end{equation}

This unitary acts upon $\Lambda$ in a convenient fashion:
\begin{proposition}
    Let $V$ be the Gaussian unitary defined  in equation \eqref{eq:fermionic-lambda-diagonalizer}. Then
    \begin{equation}
        V \Lambda V^\dagger = 2(\identity \otimes \hat{N} - \hat{N} \otimes \identity),
    \end{equation}
    where $\hat{N}$ is the total occupation operator $\sum_{i=1}^n c_i^\dagger c_i$.
\end{proposition}
\begin{proof}
    Each term $V_i = \exp\left(\frac{\ri \pi}{4} \majorana_{2i} \otimes \majorana_{2i - 1}\right)$ only interacts with $\majorana_{2i-1} \otimes \majorana_{2i-1} + \majorana_{2i} \otimes \majorana_{2i}$. Now, $(\majorana_{2i} \otimes \majorana_{2i - 1})^2 = \identity$, so 
    \[
    \exp\left(\frac{\ri \pi}{4}\majorana_{2i} \otimes \majorana_{2i - 1}\right) = \frac{\sqrt{2}}{2} \identity + \ri \frac{\sqrt{2}}{2} \majorana_{2i} \otimes \majorana_{2i - 1}.
    \]
    We can thus compute
    \begin{align}
        V_i(\majorana_{2i - 1} \otimes \majorana_{2i - 1})V_i^\dagger &= \frac{1}{2}(\identity + \ri \majorana_{2i} \otimes \majorana_{2i - 1}) \majorana_{2i - 1}^{\otimes 2} (\identity - \ri \majorana_{2i} \otimes \majorana_{2i - 1}) \\
        &= \frac{1}{2}\left(\majorana_{2i-1}^{\otimes 2} + \ri \majorana_{2i}\majorana_{2i-1} \otimes \identity - \ri \majorana_{2i-1}\majorana_{2i} \otimes \identity - \majorana_{2i-1}^{\otimes 2}\right) \\
        &= - \ri \majorana_{2i-1}\majorana_{2i} \otimes \identity = (\identity - 2 a_i^\dagger a_i) \otimes \identity,
    \end{align}
    Similarly,
    \begin{align}
        V_i(\majorana_{2i} \otimes \majorana_{2i})V_i^\dagger &= \frac{1}{2}(\identity + \ri \majorana_{2i} \otimes \majorana_{2i-1}) \majorana_{2i }^{\otimes 2} (\identity - \ri \majorana_{2i} \otimes \majorana_{2i - 1}) \\
        &= \frac{1}{2}\left(\majorana_{2i}^{\otimes 2} + \ri \identity \otimes \majorana_{2i-1}\majorana_{2i} - \ri \identity \otimes \majorana_{2i}\majorana_{2i-1} - \majorana_{2i}^{\otimes 2}\right) \\
        &= \ri \identity \otimes  \majorana_{2i-1}\majorana_{2i} =  \identity \otimes (2 a_i^\dagger a_i - \identity).
    \end{align}
    Thus
    \begin{align}
        V\Lambda V^\dagger &= \sum_{i=1}^n  \identity \otimes (2 a_i^\dagger a_i - \identity) - (2 a_i^\dagger a_i - \identity) \otimes \identity \\
        &= 2 \sum_{i=1}^n \identity \otimes a_i^\dagger a_i -  a_i^\dagger a_i\otimes \identity \\
        &= 2(\identity \otimes \hat{N} - \hat{N} \otimes \identity). \qedhere
    \end{align}
\end{proof}

Therefore, the $0$-eigenspace of $\Lambda$ corresponds to the equal-occupation sector of the Fock space after applying $V$.

\begin{algorithm}[H]
    \caption{Base tester for fermionic Gaussian states.}\label{alg:fermionic-gaussian-tester}
    \KwInput{Two copies of a fermionic pure state $\ket{\psi}$.}

    \BlankLine

    Apply the Gaussian unitary $V$ defined in \cref{eq:fermionic-lambda-diagonalizer} to $\ket{\psi}^{\otimes 2}$.
    
    Measure the occupation numbers $N_1, N_2$ of the first and second registers, respectively.

    \uIf{$N_1 = N_2$}{
        \KwRet{\accept}
    }
    \Else{
        \KwRet{\reject}
    }
\end{algorithm}

\begin{proof}[Proof of \cref{thm:testing-fermionic-gaussians}]
    The algorithm is \cref{alg:constant-copy-tester} with \cref{alg:fermionic-gaussian-tester} serving as the base algorithm $\calA$. Perfect completeness of the test is due to Bravyi \cite{bravyi2004lagrangian}. By \cref{thm:rejection-infidelity-bounds} with the spectral gap $\Delta = 5/12$ (\cref{lem:fermionic-spectral-gap}), $k=2$, and $m=3$, we have that the rejection probability for any state $\ket{\psi}$ obeys
    \begin{equation}
        \frac{1}{8} \, \dist(\psi, \fermionicgaussianset)^2 \leq q(\psi) \leq 2 \, \dist(\psi, \fermionicgaussianset)^2
    \end{equation}
    Hence \cref{thm:framework-tolerant-testing} yields the claim with $\beta = 4$.
\end{proof}

\begin{remark}
    One can also construct a property tester for fermionic Gaussian \emph{unitaries} $U$ by considering the fermionic Choi state of $U$:
    \[
    (U \otimes \identity) \prod_{i=1}^n \left(\frac{\identity + i \majorana_a \majorana_{a + 2n}}{2}\right) (U^\dagger \otimes \identity).
    \]
    It can easily be shown that applying our property tester to this Choi state can distinguish whether $U$ is fermionic Gaussian or $\eps$-far from all Gaussian unitaries in Frobenius distance with sample complexity $O(1/\eps^2)$, independent of $n$. This resolves an open question of \cite{iyer2025mildly}.
\end{remark}

\section{Testing Slater determinants}\label{sec:slater-dets-spectral-gap}
In this section, we prove the constant sample complexity upper bound for testing Slater determinants. The test is a specialization from pure fermionic Gaussian states to definite particle number. 

\begin{theorem}\label{thm:testing-slater}
    Let $\beta = \sqrt{12}$. For all $0 \leq \eps_1 < \eps_2 < 1$ with $\eps_2 > \beta \eps_1$ and $0 < \delta < \frac{1}{2}$, there is an $O\paren*{\frac{\log(1/\delta)}{(\eps_2 - \beta\eps_1)^2}}$-copy $(\eps_1, \eps_2, \delta)$-tester for Slater determinants. The tester acts on two copies at a time.
\end{theorem}

\subsection{Perfect completeness}

The following two-copy test has perfect completeness for Slater determinants.

\begin{definition}[Slater tester]
    Define the two-copy operator
    \[
        K \coloneqq G^\dagger G, \quad \text{where } G \coloneqq \sum_{j=1}^n c_j \otimes c_j^\dagger.
    \]
    Given two copies of the state $\rho$, the tester measures the observable $K$ and accepts if and only if $0$ is obtained.
\end{definition}

To analyze the test, first we record some basic definitions and facts.

\begin{definition}
    The $1$-particle reduced density matrix ($1$-RDM) of a state $\ket{\psi}$ is the Hermitian matrix $D \in \C^{n \times n}$ defined via
    \[
        D_{jk} \coloneqq \braket{\psi | c_j^\dagger c_k | \psi}.
    \]
\end{definition}

\begin{fact}\label{fact:rdm_properties}
    The $1$-RDM always obeys $0 \preceq D \preceq \identity$, and $\tr \, D = \eta$ if $\ket{\psi}$ is an $\eta$-particle eigenstate. Furthermore, $\ket{\psi}$ is an $\eta$-particle Slater determinant if and only if $D$ is a rank-$\eta$ orthogonal projector.
\end{fact}

Perfect completeness is implied by the following statement.

\begin{proposition}\label{prop:slater_completeness}
    A state $\ket{\psi}$ obeys $K\ket{\psi}^{\otimes 2} = 0$ if and only if $\ket{\psi}$ is a Slater determinant.
\end{proposition}

\begin{proof}
    We write $\ket{\Psi} \equiv \ket{\psi} \otimes \ket{\psi}$. Since the kernel of $K = G^\dagger G$ is also the kernel of $G$, it suffices to consider $G\ket{\Psi} = 0$. To do so, we examine the norm of $G\ket{\Psi}$:
    \begin{equation}
    \begin{split}
        \braket{\Psi | G^\dagger G | \Psi} &= \sum_{j,k=1}^n \braket{\psi | c_j^\dagger c_k | \psi} \braket{\psi | c_j c_k^\dagger | \psi}\\
        &= \sum_{j,k=1}^n D_{jk} (\delta_{jk} - D_{kj})\\
        &= \tr(D(\identity - D)).
    \end{split}
    \end{equation}
    By \cref{fact:rdm_properties}, $0 \preceq D \preceq \identity$ and so this vanishes if and only if $\ket{\psi}$ is a Slater determinant.
\end{proof}

\subsection{Spectral gap analysis}\label{sec:slater_spec_gap}

Throughout this subsection, we will work in the fixed particle sector $\calH = {\bigwedge}^\eta \C^n$ of $\eta$ fermions in $n$ modes. We will show that spectral gap is independent of $\eta$, thus applying to all sectors. Afterward in \cref{sec:slater_soundness} we show how to extend the soundness guarantee beyond number symmetry, thereby applying to the full fermionic Fock space $\bigwedge \C^n$.

Let $\slatertestproj$ be the projector onto the kernel of $K$ and $\Xi = \symproj^{(1,2,3)}\cdot\slatertestproj^{(1,2)}\cdot\symproj^{(1,2,3)}$ the operator whose spectral we need. As we did in \cref{sec:spectral-gap-concrete-proof}, we define the following alternative operator which shares the same nonzero spectra as $\Xi$:
\[
\Theta = \slatertestproj^{(1,2)}\cdot\symproj^{(1,2,3)}\cdot\slatertestproj^{(1,2)}.
\]
We shall prove that the eigenvalues of $\Theta$ (hence $\Xi$) are contained in $[0, 5/9] \cup \{1\}$, which gives the sufficient spectral gap bound $\Delta = 4/9 > 1/3$. As in \cref{prop:Theta_identity}, it holds that $\Theta = (\slatertestproj^{(1,2)} + 2P)/3$ on the symmetric subspace where $P = \slatertestproj^{(1,2)} \cdot \SWAP^{(2,3)} \cdot \slatertestproj^{(1,2)}$. Thus it suffices to analyze the spectrum of $P$, which we focus on for the remainder of this section.

\paragraph{Three-copy antisymmetric representation.}
The three-copy space is $\calH^{\otimes 3} = \left(\bigwedge \C^n\right)^{\otimes 3}$. However, it will be convenient to work with the space of $3n$ fermionic modes,
\[
W \coloneqq \bigwedge (\C^3 \otimes \C^n),
\]
and ``undo'' the extra antisymmetrization at the end. For each $\alpha \in [3]$ and $j \in [n]$, $c_{\alpha,j}^\dagger$ is the creation operator on the $j$th mode in the $\alpha$ block. For $S \subseteq [n]$, define
\[
A_S^{(\alpha)} \coloneqq \prod_{j \in S} c_{\alpha,j}^\dagger
\]
and consequently the isometry
\[
\begin{split}
    \iota : \calH^{\otimes 3} &\to W,\\
    \ket{S_1} \otimes \ket{S_2} \otimes \ket{S_3} &\mapsto A_{S_1}^{(1)} A_{S_2}^{(2)} A_{S_3}^{(3)} \ket{0^{3n}}.
\end{split}
\]
We also need a representation of $\genlingrp(3, \C)$ acting on the copy index inside the exterior product. Denote this by
\[
\begin{split}
    \Gamma : \genlingrp(3, \C) &\to \genlingrp(W),\\
    g &\mapsto \bigwedge (g \otimes \identity).
\end{split}
\]
Finally, let
\begin{equation}\label{eq:E_ab_W}
    E_{\alpha\beta} \coloneqq \sum_{j=1}^n c_{\alpha,j}^\dagger c_{\beta,j}
\end{equation}
and $s_{23} \in \genlingrp(3, \C)$ be the permutation matrix swapping coordinates $2$ and $3$ in the standard basis. Throughout, we will restrict to a fixed $\eta$th exterior power (physically, the $\eta$-particle sector) on each register:~$|S_1| = |S_2| = |S_3| = \eta$.

Recall that we are interested in analyzing the spectrum of $P = \slatertestproj^{(1,2)} \cdot \SWAP^{(2,3)} \cdot \slatertestproj^{(1,2)}$, where $\slatertestproj$ is the projector onto the kernel of $G = \sum_{j=1}^n c_j \otimes c_j^\dagger$. To map these operators between the spaces $\calH^{\otimes 3}$ and $W$, we have the following relations:
\begin{equation}\label{eq:iota_transform}
    \iota^{-1} \Gamma(s_{23}) \iota = (-1)^\eta \, \SWAP^{(2,3)}, \quad \iota^{-1} E_{21} \iota = (-1)^{\eta-1} (G \otimes \identity).
\end{equation}
Note that the signs are due to the canonical anticommutation relations between blocks of $W$ that are absent from the tensor-product structure of $\calH^{\otimes 3}$. For example, $\SWAP^{(2,3)} : \ket{S_1} \otimes \ket{S_2} \otimes \ket{S_3} \mapsto \ket{S_1} \otimes \ket{S_3} \otimes \ket{S_2}$ whereas $\Gamma(s_{23}) : A_{S_1}^{(1)} A_{S_2}^{(2)} A_{S_3}^{(3)} \ket{0^{3n}} \mapsto A_{S_1}^{(1)} A_{S_2}^{(3)} A_{S_3}^{(2)} \ket{0^{3n}}$. Bringing this final state's blocks into the appropriate order requires pulling $\eta$ creation operators through $\eta$ others, which incurs an overall sign of $(-1)^{\eta^2} = (-1)^{\eta}$.

We will analyze the representation $\Gamma : \genlingrp(3, \C) \to \genlingrp(W)$ in terms of the highest-weight decomposition (\cref{prop:highest-weight-decomp}):
\begin{equation}\label{eq:W_decomp}
W \cong \bigoplus_\lambda V_\lambda \otimes M_\lambda
\end{equation}
where $V_\lambda$ is the irrep of highest weight $\lambda \equiv (\lambda_1, \lambda_2, \lambda_3)$ and $M_\lambda$ is some multiplicity space. Each $\lambda_\alpha \in \{0, 1, \ldots, n\}$ counts the occupation number of the $\alpha$ block. Indeed, for diagonal $t \in \genlingrp(3, \C)$ on the standard basis $e_{\alpha,j} = e_\alpha \otimes e_j \in \C^3 \otimes \C^n$, we have
\[
(t \otimes \identity)e_{\alpha,j} = t_\alpha e_{\alpha,j} \implies \Gamma(t)(e_{\alpha_1,j_1} \wedge \cdots \wedge e_{\alpha_\eta,j_\eta}) = \left(\prod_{k=1}^\eta t_{\alpha_k}\right) e_{\alpha_1,j_1} \wedge \cdots \wedge e_{\alpha_\eta,j_\eta}.
\]
Putting $\lambda_\alpha = |\{k : \alpha_k = \alpha\}|$ for the above expression, the eigenvalue simplifies to $t_1^{\lambda_1} t_2^{\lambda_2} t_3^{\lambda_3}$ which is precisely the definition of a weight. Thus the weight space we are specifically interested in is $W[(\eta, \eta, \eta)]$, and so the tester's accept space $\ker G$ embeds into $W$ as
\[
\iota(\ker G) = \ker E_{21} \cap  W[(\eta, \eta, \eta)].
\]

In order to characterize the kernel of $G$, it suffices to analyze its support within each irrep $V_\lambda$. The following lemma identifies which highest weights intersect nontrivially.

\begin{lemma}\label{lem:nonzero_kernel_irreps}
    $\iota(\ker G)$ is nonzero in the $V_\lambda$ factor of \cref{eq:W_decomp} if and only if
    \[
    \lambda = (\eta + r, \eta, \eta - r) \quad \text{for } 0 \leq r \leq \min\{\eta, n-\eta\}.
    \]
    For such $\lambda$, this intersection with $V_\lambda$ is one-dimensional.
\end{lemma}

\begin{proof}
    First we establish that $V_\lambda$ decomposes into irreducible $\spllinalg(2, \C)$-modules;~this will later allow us to apply \cref{prop:sl2-module} to each. To see this, observe that $E_{12}$, $E_{21}$, and $E_{11} - E_{22}$ generate a copy of $\spllinalg(2, \C)$ acting on $W$ (we show this in \cref{claim:sl2_on_W} below). Thus $V_\lambda$ is an $\spllinalg(2, \C)$-module, although generically reducible. Decomposing its $\genlingrp(2, \C)$-restriction into irreps,
    \begin{equation}\label{eq:V_irrep_GL2_decomp}
        V_\lambda \big\rvert_{\genlingrp(2, \C)} = \bigoplus_\mu U_\mu,
    \end{equation}
    it holds that each $U_\mu$ remains irreducible when restricted to $\spllingrp(2, \C)$. Indeed by \cref{prop:GL-highest-weight}, if $\mu = (\mu_1, \mu_2)$ then $U_\mu \big\rvert_{\spllingrp(2, \C)}$ has highest weight $k = \mu_1 - \mu_2$.
    
    Recall that we are only interested in the intersection with $\iota(\ker G)$, which itself is contained in $W[(\eta, \eta, \eta)]$. Let $v \in V_\lambda \cap \iota(\ker G)$ and decompose $v = \sum_\mu v_\mu$ according to \cref{eq:V_irrep_GL2_decomp}. Since $v \in W[(\eta, \eta, \eta)]$, each nonzero vector $v_\mu$ has $\genlingrp(2, \C)$-weight $(\eta, \eta)$ and hence $\spllinalg(2, \C)$-weight $0$. On the other hand, since $v \in \iota(\ker G)$ we have $E_{21}v = 0$. Hence every $E_{21}v_\mu = 0$. We can now use \cref{prop:sl2-module}:~$U_\mu$ restricts to an irreducible $\spllingrp(2, \C)$-module of highest weight $k = \mu_1 - \mu_2$, which means it contains weights $k, k-2, \ldots, -k$. The lowering operator $E_{21}$ moves a vector down these weights and kills the bottom vector, but since $E_{21}v_\mu = 0$ it must be the case that $-k = 0$ is already the bottom of the string. Thus $\mu = (s, s)$ for some $s$, implying that $U_{(s, s)} \cong {\det}^s$ is one-dimensional (this can be seen from \cref{prop:det-twist} with $U_{(0, 0)} \cong \C$ being trivial). As such, it can only support vectors of weight $(s, s)$;~but we already know that $v_\mu \in U_\mu$ has weight $(\eta, \eta)$, which forces $s = \eta$.

    Now we apply the branching rule (\cref{prop:GL-interlacing}) to compute the $\genlingrp(2, \C)$-weights that can occur within the $\genlingrp(3, \C)$-irrep $V_\lambda$. The weight $\mu = (\mu_1, \mu_2)$ occurs only if the interlacing condition
    \[
    \lambda_1 \geq \mu_1 \geq \lambda_2 \geq \mu_2 \geq \lambda_3
    \]
    is met, and if so, it has multiplicity $1$. As established above, $\mu_1 = \mu_2 = \eta$ so $\lambda_2 = \eta$. Furthermore, since $V_\lambda$ contains the weight $(\eta, \eta, \eta)$, we know that $|\lambda| = 3\eta$ and so $\lambda_1 + \lambda_3 = 2\eta$. Writing $\lambda_1 = \eta + r$ gives $\lambda_3 = \eta - r$. Constraining the weight coordinates within $\{0, 1, \ldots, n\}$ gives the restriction $r \leq \min\{\eta, n-\eta\}$. Finally, since $\mu$ is multiplicity-free under the interlacing theorem, the one-dimensional $U_{(\eta, \eta)}$ can only occur once. Thus the intersection with $V_\lambda$ is also one-dimensional.

    The converse is shown by recognizing that for every such $\lambda = (\eta + r, \eta, \eta - r)$, the required $\genlingrp(2, \C)$-weight $\mu = (\eta, \eta)$ interlaces $\lambda$ and occurs with multiplicity $1$. Uniqueness and one-dimensionality then follow from the same arguments as above.
\end{proof}

\begin{claim}\label{claim:sl2_on_W}
    The operators $E_{12}$, $E_{21}$, and $H \coloneqq E_{11} - E_{22}$ are generators for $\spllinalg(2, \C)$.
\end{claim}

\begin{proof}
    The standard generators of $\spllinalg(2, \C)$ obey the commutation relations
    \[
    [H, E_{12}] = 2E_{12}, \quad [H, E_{21}] = -2E_{21}, \quad [E_{12}, E_{21}] = H.
    \]
    Straightforward calculations using the (anti)commutation relations of the $c_{\alpha,j}$ operators show that these hold.
\end{proof}

\paragraph{Expression for the eigenvalues.}
This characterization of $\ker G$ using one-dimensional subspaces of irreps of $\genlingrp(3, \C)$ reveals a compact expression for the eigenvalues of $P = \slatertestproj^{(1,2)} \cdot \SWAP^{(2,3)} \cdot \slatertestproj^{(1,2)}$.

\begin{lemma}
    Restrict to any fixed-particle sector $\left({\bigwedge}^\eta \C^n\right)^{\otimes 3}$. There, the eigenvalues of $P$ take the form
    \[
    p_r = \frac{\langle X_r, s_{23} X_r \rangle}{\langle X_r, X_r \rangle}, \quad 0 \leq r \leq \min\{\eta, n-\eta\},
    \]
    where $X_r$ is any nonzero vector in $V_{(r, 0, -r)}[(0, 0, 0)] \cap \ker E_{21}$. (By abuse of notation, we write $s_{23}$ for its representation on $V_{(r, 0, -r)}$.)
\end{lemma}

\begin{proof}
    Write $\lambda^{(r)} \coloneqq (\eta + r, \eta, \eta - r)$. By \cref{lem:nonzero_kernel_irreps}, we know that $\slatertestproj^{(1,2)}$ is nontrivial only on the irreps $V_{\lambda^{(r)}}$. The determinant twist rule (\cref{prop:det-twist}) allows us to reduce to a highest weight of $(r, 0, -r)$, since
    \[
    V_{(\eta+r,\eta,\eta-r)} \cong V_{(r, 0, -r)} \otimes {\det}^\eta.
    \]
    Recall from \cref{eq:iota_transform} that $\SWAP^{(2,3)} \cong (-1)^\eta \, \Gamma(s_{23})$, and so since $\det(s_{23}) = -1$, the signs cancel. (The sign $(-1)^{\eta-1}$ appearing on $G$ does not matter because we only consider its kernel projector.) \cref{lem:nonzero_kernel_irreps} also established that $\slatertestproj^{(1,2)}$ in $V_{\lambda^{(r)}}$ is a one-dimensional projector. Thus there can only be one eigenvector with nonzero eigenvalue (modulo the multiplicity space $M_{\lambda^{(r)}}$, on which everything acts trivially), which is $X_r \in V_{(r,0,-r)}$. Clearly, $X_r \in \ker E_{21}$ as well since $E_{21}$ is nothing but a representation of $G$. Finally, we know that $X_r$ has $\genlingrp(3, \C)$-weight $(0, 0, 0)$, again because of the determinant twist:~if $v$ is a vector in $V_\lambda$ with weight $\mu$, then under the correspondence $V_\lambda \cong V_{\lambda - (\eta, \eta, \eta)} \otimes {\det}^\eta$ it is a vector in $V_{\lambda - (\eta, \eta, \eta)}$ of weight $\mu - (\eta, \eta, \eta)$. In our case, we only consider vectors of weight $\mu = (\eta, \eta, \eta)$ at the onset.
\end{proof}

\paragraph{Computing the eigenvalues.}
Now we evaluate $p_r$. Our goal is to find a useful expression for the eigenvectors $X_r$ given the three constraints:
\begin{enumerate}
    \item $X_r \in V_{(r,0,-r)}$,
    \item $X_r$ has weight $(0, 0, 0)$,
    \item $E_{21} X_r = 0$.
\end{enumerate}
This third point is particularly powerful:~viewing $E_{21}$ as the lowering operator in $\spllinalg(2, \C)$, it says that $X_r$ lies in the trivial $\spllinalg(2, \C)$-submodule of $V_{(r,0,-r)}$. Note also that the second point is not logically independent, but rather a consequence of the first and third points combined with the branching rule (\cref{prop:GL-interlacing}).

It will be highly convenient to express elements of $\genlingrp(3, \C)$-irreps $V_\lambda$ using the following polynomial model. First, the irrep of highest weight $(r, 0, 0)$ is $\Sym^r(\C^3)$, and its dual $\Sym^r(\C^3)^*$ is the irrep of highest weight $(0, 0, -r)$ \cite[Theorems 3.2.13 and 9.3.10]{goodman2009symmetry}. By \cite[Proposition 5.5.19]{goodman2009symmetry}, the irrep of their summed highest weight $(r, 0, -r)$ is contained in the tensor product of $V_{(r,0,0)}$ and $V_{(0,0,-r)}$, i.e.,
\begin{equation}\label{eq:polynomial_model}
    V_{(r,0,-r)} \subseteq \Sym^r(\C^3) \otimes \Sym^r(\C^3)^*.
\end{equation}
Realize $\calP_{r,r} \coloneqq \Sym^r(\C^3) \otimes \Sym^r(\C^3)^*$ as the space of homogeneous degree-$(r, r)$ polynomials in variables $x_1, x_2, x_3$ and dual variables $y_1, y_2, y_3$. We occasionally adopt multi-index notation:~for $\alpha = (\alpha_1, \alpha_2, \alpha_3)$, write $x^\alpha \coloneqq x_1^{\alpha_1} x_2^{\alpha_2} x_3^{\alpha_3}$ and $\alpha! \coloneqq \alpha_1! \alpha_2! \alpha_3!$. This is particularly convenient for defining the Fischer inner product on multivariate polynomials,
\begin{equation}\label{eq:fischer_inner_product}
    \langle x^\alpha y^\beta, x^{\alpha'} y^{\beta'} \rangle = \alpha! \beta! \, \delta_{\alpha\alpha'} \delta_{\beta\beta'}.
\end{equation}
Overloading $x_i$ to also indicate the multiplication operation by $x_i$, we have that $\partial_{x_i}$ is the adjoint of $x_i$ with respect to this inner product (similarly for $y_i$ and $\partial_{y_i}$). Indeed, one can verify this by checking the univariate case:
\[
\begin{split}
    &\textit{Multiplication: } \langle x^k, x \cdot x^\ell \rangle = \langle x^k, x^{\ell+1} \rangle = k! \, \delta_{k,\ell+1},\\
    &\textit{Differentiation: } \langle \partial_x x^k, x^\ell \rangle = k \langle x^{k-1}, x^\ell \rangle = k \cdot (k-1)! \, \delta_{k-1,\ell}.
\end{split}
\]
Because these operations commute over different variables, the multivariate case follows.

\begin{proposition}\label{prop:polynomial_model_dictionary}
    We have the following dictionary between polynomial model $\calP_{r,r}$ and the representation theory of $\genlingrp(3, \C)$.
    \begin{enumerate}
        \item The $\genlingrp(3, \C)$-weight of $x^\alpha y^\beta$ is $\alpha - \beta$.

        \item The operators $E_{ij} \coloneqq x_i \partial_{x_j} - y_j \partial_{y_i}$ generate a copy of $\genlinalg(3, \C)$ inside of $\calP_{r,r}$ (hence it is a $\genlingrp(3, \C)$-module).

        \item The Fischer inner product agrees with the standard inner product on $V_{(r,0,-r)}$, up to a normalization convention.
    \end{enumerate}
\end{proposition}

\begin{proof}
    \textit{1. $\genlingrp(3, \C)$ weights.} By definition the standard basis vectors $e_i$ of $\C^3$ have themselves as weights, whereas the dual basis $e_i^*$ has $-e_i$ as weights. Then by the identification $e_i \leftrightarrow x_i$, $e_i^* \leftrightarrow y_i$, the group action of $t = \diag(t_1, t_2, t_3) \in \genlingrp(3, \C)$ on $\calP_{r,r}$ is
    \[
    t \cdot x_i = t_i x_i, \quad t \cdot y_i = t_i^{-1} y_i.
    \]
    Group actions are homomorphisms, so the product preserves this:
    \[
    t \cdot (x^\alpha y^\beta) = \prod_{i=1}^3 (t_i x_i)^{\alpha_i} \prod_{j=1}^3 (t_j^{-1} y_j)^{\beta_j} = \left( \prod_{i=1}^3 t_i^{\alpha_i - \beta_i} \right) x^\alpha y^\beta.
    \]

    \textit{2. The algebra $\genlinalg(3, \C)$.} Recall that the abstract action of $\genlinalg(3, \C)$ on the standard basis should be $E_{ij} e_k = \delta_{jk} e_i$. We verify this under the polynomial identification:
    \[
    E_{ij} x_k = x_i \partial_{x_j} x_k - y_j \partial_{y_i} x_k = \delta_{jk} x_i.
    \]
    The dual representation of $\genlinalg(3, \C)$, $E_{ij} \cdot e_k^* = -\delta_{ik} e_j^*$, is also easily checked:
    \[
    E_{ij} y_k = 0 - y_j \partial_{y_i} y_k = -\delta_{ik} y_j.
    \]
    
    \textit{3. Inner product preservation.} Consider $\unitgrp(3) \subset \genlingrp(3, \C)$. By definition, $V_{(r,0,-r)}$ is a $\genlingrp(3, \C)$-module, so in particular its inner product is invariant under $\unitgrp(3)$. Meanwhile we have just shown that $\calP_{r,r}$ is also a $\genlingrp(3, \C)$-module, and its Fischer inner product is indeed invariant under $\unitgrp(3)$ (implied by the fact that $x_i$ is the adjoint of $\partial_{x_i}$). Take the subspace of $\calP_{r,r}$ isomorphic to $V_{(r,0,-r)}$, which is known to exist due to the argument preceding \cref{eq:polynomial_model};~this is irreducible and inherits the Fischer inner product. By Schur's lemma, any two unitarily invariant inner products on an irrep are equivalent up to some scalar.
\end{proof}

In general, $\calP_{r,r}$ is larger than the irrep $V_{(r,0,-r)}$ we are interested in. We seek a characterization of the irrep;~it turns out that $V_{(r,0,-r)}$ is precisely a harmonic space sitting within $\calP_{r,r}$.

\begin{lemma}\label{lem:harmonic_space}
Define the harmonic space
\[
\calK_r \coloneqq \ker \Delta \cap \calP_{r,r}, \quad \text{where } \Delta \coloneqq \sum_{i=1}^3 \partial_{x_i} \partial_{y_i}
\]
is the Laplacian on $\calP_{r,r}$. Then $\calK_r \cong V_{(r,0,-r)}$.
\end{lemma}

\begin{proof}
    To simplify notation, we assume throughout that $\Delta$ denotes its restriction to $\calP_{r,r}$. First observe that $\ker\Delta$ is a $\genlinalg(3, \C)$-module. To see this, we first compute the commutator $[\Delta, E_{ij}]$. Let $f$ be an arbitrary polynomial. Take the $x$- and $y$-part of $E_{ij}$ separately:
    \[
    \begin{split}
        \Delta(x_i \partial_{x_j} f) = \sum_{k=1}^3 \partial_{x_k} \partial_{y_k} (x_i \partial_{x_j} f) = \sum_{k=1}^3 \partial_{x_k} (x_i \partial_{x_j} \partial_{y_k} f) &= \sum_{k=1}^3 ( \delta_{ik} \partial_{x_j} \partial_{y_k} f + x_i \partial_{x_k} \partial_{x_j} \partial_{y_k} f )\\
        &= \partial_{x_j} \partial_{y_i} f + x_i \partial_{x_j} \Delta f,
    \end{split}
    \]
    \[
    \begin{split}
        \Delta(-y_j \partial_{y_i} f) = -\sum_{k=1}^3 \partial_{x_k} \partial_{y_k} (y_j \partial_{y_i} f) = -\sum_{k=1}^3 \partial_{x_k} (\delta_{jk} \partial_{y_i} f + y_j \partial_{y_k} \partial_{y_i} f) = -(\partial_{x_j} \partial_{y_i} f + y_j \partial_{y_i} \Delta f).
    \end{split}
    \]
    Combining, we get
    \[
    \begin{split}
        \Delta (E_{ij} f) = x_i \partial_{x_j} \Delta f - y_j \partial_{y_i} \Delta f = E_{ij}(\Delta f),
    \end{split}
    \]
    which implies $[\Delta, E_{ij}] = 0$. Consequently for every $v \in \ker\Delta$, it holds that $E_{ij} v \in \ker\Delta$. Since every $E_{ij}$ preserves $\ker\Delta$, it is a $\genlinalg(3, \C)$-module.

    Next we show that $\ker\Delta$ contains a copy of $V_{(r,0,-r)}$. Consider a particular element $v_r = x_1^r y_3^r \in \ker\Delta$, which has weight $(r, 0, -r)$. By inspection, all raising operators on $\genlinalg(3, \C)$ kill $v_r$:
    \[
    E_{12}v_r = E_{13}v_r = E_{23}v_r = 0.
    \]
    Thus $v_r$ is in fact a vector of highest weight. By the theorem of highest weight, this vector precisely generates the irrep $V_{(r, 0, -r)}$ under the image of the lowering operators $E_{21}, E_{31}, E_{32}$. Thus $V_{(r,0,-r)} \subseteq \ker\Delta$.

    It remains to show that this inclusion is saturated. We prove this by dimension counting. Since $\Delta : \calP_{r,r} \to \calP_{r-1,r-1}$, by the rank--nullity theorem we have
    \[
    \dim\ker\Delta = \dim\calP_{r,r} - \rank\Delta.
    \]
    The dimension of $\calP_{r,r}$ is straightforward:~homogeneous degree-$r$ polynomials in three variables are spanned by $\binom{r+2}{2}$ monomials, so
    \[
    \dim \calP_{r,r} = \binom{r+2}{2}^2.
    \]
    To compute the rank of $\Delta$, it will be easier to analyze its adjoint
    \[
    \Delta^* = \sum_{i=1}^3 x_i y_i
    \]
    and use the fact that $\rank\Delta = \rank\Delta^*$. We show that $\Delta^* : \calP_{r-1,r-1} \to \calP_{r,r}$ is one-to-one. Indeed, since $\Delta^*$ is the linear map $f \mapsto (x_1 y_1 + x_2 y_2 + x_3 y_3)f$, it suffices to determine when $\Delta^* f = 0$. But $f = 0$ is the only possible solution, because polynomial rings over $\C$, in this case $\C[x_1, x_2, x_3, y_1, y_2, y_3]$, have no zero divisors. Hence
    \[
    \rank\Delta^* = \dim\calP_{r-1,r-1} = \binom{r+1}{2}^2,
    \]
    which is also $\rank\Delta$, and so
    \[
    \dim\ker\Delta = \binom{r+2}{2}^2 - \binom{r+1}{2}^2 = (r+1)^3.
    \]
    Now compare this to the dimension of $V_{(r,0,-r)}$. It is known that \cite[Theorem 6.3]{fulton2004representation}
    \[
    \dim V_\lambda = \prod_{1\leq i < j \leq d} \frac{\lambda_i - \lambda_j + j - i}{j - i}.
    \]
    Thus for $\lambda = (r, 0, -r)$,
    \[
    \dim V_{(r,0,-r)} = \frac{(r+1)(2r+2)(r+1)}{2} = (r+1)^3,
    \]
    implying that $V_{(r,0,-r)}$ and $\ker\Delta$ within $\calP_{r,r}$ are isomorphic.

    Technically, the above argument only holds for $r \geq 1$ so that the codomain $\calP_{r-1,r-1}$ makes sense. The case of $r = 0$ is trivial however, since $\Sym^0(\C^3)$ and $V_{(0,0,0)}$ are both just $\C$.
\end{proof}

We now have sufficient constraints to completely characterize the eigenvector $X_r$. After that it will be straightforward to evaluate the eigenvalue $p_r = \frac{\langle X_r, s_{23} X_r \rangle}{\langle X_r, X_r \rangle}$.

\begin{lemma}
    The one-dimensional subspace $V_{(r,0,-r)}[(0, 0, 0)] \cap \ker E_{21} \equiv \C X_r$ can be described in terms of the polynomial model $\calP_{r,r}$ as
    \[
    X_r = \sum_{\substack{i,j,k \geq 0 \\ i+j+k=r}} \chi_{ijk} \, x_1^i x_2^j x_3^k y_1^i y_2^j y_3^k, \quad \chi_{ijk} = \frac{(-1)^{k}}{i! j! k!} \binom{r+1}{k}.
    \]
\end{lemma}

\begin{proof}
    Recall that the $\genlingrp(3, \C)$-weight of the monomial $x^\alpha y^\beta$ is $\alpha - \beta$. Thus $X_r$ having weight $(0, 0, 0)$ implies it can only be supported by $x^\alpha y^\alpha$ terms. Write $\alpha = (i, j, k)$. By definition of $\calP_{r,r} = \Sym^r(\C^3) \otimes \Sym^r(\C^3)$, the degree must be exactly $i + j + k = r$.
    
    Now, we use the two central facts:~$E_{21} X_r = 0$ and $\Delta X_r = 0$. The first condition expands to
    \[
    \begin{split}
        E_{21} X_r &= (x_2 \partial_{x_1} - y_1 \partial_{y_2}) \sum_{\substack{i,j,k \geq 0 \\ i+j+k=r}} \chi_{ijk} \, x_1^i x_2^j x_3^k y_1^i y_2^j y_3^k\\
        &= \sum_{\substack{i,j,k \geq 0 \\ i+j+k=r}} \chi_{ijk} \left[ (1-\delta_{i,0})i \, x_1^{i-1} x_2^{j+1} x_3^k y_1^i y_2^j y_3^k - (1-\delta_{j,0})j \, x_1^i x_2^j x_3^k y_1^{i+1} y_2^{j-1} y_3^k \right],
    \end{split}
    \]
    where the Kronecker delta is due to the boundary cases where there are no $x_1$ or $y_2$ factor. By linear independence of monomials, we match degrees to get the recurrence relation
    \[
    i \, \chi_{ijk} - (j+1) \chi_{i-1,j+1,k} = 0.
    \]
    By inspection, the choice
    \[
    \chi_{ijk} = \frac{d_k}{i!j!}
    \]
    satisfies this recurrence, where the $k$ dependence $d_k$ is yet unspecified. Note that we can handle the boundary cases by demanding $\chi_{ijk} = 0$ whenever either any of the indices are negative, or they do not sum to $r$.

    We get $d_k$ from the second kernel condition $\Delta X_r = 0$:
    \[
    \begin{split}
        \Delta X_r &= (\partial_{x_1} \partial_{y_1} + \partial_{x_2} \partial_{y_2} + \partial_{x_3} \partial_{y_3}) \sum_{\substack{i,j,k \geq 0 \\ i+j+k=r}} \chi_{ijk} \, x_1^i x_2^j x_3^k y_1^i y_2^j y_3^k\\
        &= \sum_{\substack{i,j,k \geq 0 \\ i+j+k=r}} \chi_{ijk} \left[ (1-\delta_{i,0})i^2 \, x_1^{i-1} x_2^j x_3^k y_1^{i-1} y_2^j y_3^k + (1-\delta_{j,0})j^2 \, x_1^i x_2^{j-1} x_3^k y_1^i y_2^{j-1} y_3^k \right.\\
        &\phantom{\sum_{\substack{i,j,k \geq 0 \\ i+j+k=r}} \chi_{ijk} []} \left. {} + (1-\delta_{k,0})k^2 \, x_1^i x_2^j x_3^{k-1} y_1^i y_2^j y_3^{k-1} \right]\\
        &= \sum_{\substack{i,j,k \geq 0 \\ i+j+k=r}} d_{k} \left[ (1-\delta_{i,0}) \frac{i}{(i-1)!j!} \, x_1^{i-1} x_2^j x_3^k y_1^{i-1} y_2^j y_3^k + (1-\delta_{j,0}) \frac{j}{i!(j-1)!} \, x_1^i x_2^{j-1} x_3^k y_1^i y_2^{j-1} y_3^k \right.\\
        &\phantom{\sum_{\substack{i,j,k \geq 0 \\ i+j+k=r}} d_{k} []} \left. {} + (1-\delta_{k,0}) \frac{k^2}{i!j!} \, x_1^i x_2^j x_3^{k-1} y_1^i y_2^j y_3^{k-1} \right].
    \end{split}
    \]
    This expression vanishes by satisfying a different recurrence relation,
    \[
    d_k \frac{k^2}{i!j!} + d_{k-1} \frac{i+1}{i!j!} + d_{k-1} \frac{j+1}{i!j!} = 0, \quad k \geq 1.
    \]
    It will be convenient to renormalize the unknown as $d'_k \coloneqq k! \, d_k$. Using this, the left-hand side becomes
    \[
    d_k \frac{k^2}{i!j!} + d_{k-1} \frac{i+j+2}{i!j!} = d'_k \frac{k}{i!j!(k-1)!} + d'_{k-1} \frac{i+j+2}{i!j!(k-1)!}.
    \]
    Setting to zero and recalling that $i + j + k = r$ yields
    \[
    d'_k = - d'_{k-1} \frac{r + 2 - k}{k}.
    \]
    This recurrence relation is solved by explicit computation:
    \[
    \begin{split}
        d'_k &= d'_0 (-1)^k \prod_{q=1}^{k} \frac{r + 2 - q}{q} = d'_0 (-1)^k \frac{(r+1)r(r-1) \cdots (r+2-k)}{k!} = d'_0 (-1)^k \binom{r+1}{k}.
    \end{split}
    \]
    The choice of $d'_0$ is just normalization, so set $d'_0 = 1$. Then $\chi_{ijk} = \frac{d'_k}{i!j!k!}$ gives the claimed result.
\end{proof}

\begin{theorem}
    The eigenvalues $p_r$ of $P = \slatertestproj^{(1,2)} \cdot \SWAP^{(2,3)} \cdot \slatertestproj^{(1,2)}$ take the form
    \[
    p_r = \frac{(-1)^r}{r+1}.
    \]
\end{theorem}

\begin{proof}
    We compute the inner products in the expression
    \[
    p_r = \frac{\langle X_r, s_{23} X_r \rangle}{\langle X_r, X_r \rangle}.
    \]
    Let us start with the denominator. Recalling the Fischer inner product from \cref{eq:fischer_inner_product},
    \[
        \langle X_r, X_r \rangle = \sum_{\substack{i,j,k \geq 0 \\ i+j+k=r}} \chi_{ijk}^2 (i!j!k!)^2 = \sum_{\substack{i,j,k \geq 0 \\ i+j+k=r}} \binom{r+1}{k}^2 = \sum_{k=0}^r (r+1-k) \binom{r+1}{k}^2
    \]
    where $r+1-k$ is the number of possible $(i, j)$ pairs for fixed $k$ under the constraint $i + j + k = r$. This sum admits a closed form by first rewriting $(r+1-k) \binom{r+1}{k} = (r+1)\binom{r}{k}$, from which we get
    \begin{equation}\label{eq:pr_den}
    \sum_{k=0}^r (r+1-k) \binom{r+1}{k}^2 = (r+1) \sum_{k=0}^r \binom{r}{k} \binom{r+1}{k} = (r+1) \binom{2r+1}{r}
    \end{equation}
    where the final equality follows from a standard combinatorial identity \cite[Table 169]{graham1994concrete}.

    Next we compute the numerator of $p_r$. The SWAP operator $s_{23}$ merely swaps $j$ and $k$ in the expansion of $X_r$, so
    \begin{equation}\label{eq:pr_num}
    \begin{split}
        \langle X_r, s_{23} X_r \rangle &= \sum_{\substack{i,j,k \geq 0 \\ i+j+k=r}} \chi_{ikj} \chi_{ijk} (i!j!k!)^2\\
        &= \sum_{0 \leq j+k \leq r} (-1)^{j+k} \binom{r+1}{j} \binom{r+1}{k}\\
        &= \sum_{q=0}^r (-1)^q \sum_{k=0}^{r-q} \binom{r+1}{q-k} \binom{r+1}{k}\\
        &= \sum_{q=0}^r (-1)^q \binom{2r+2}{q}\\
        &= (-1)^r \binom{2r+1}{r}
    \end{split}
    \end{equation}
    where the last two lines again can be found in \cite{graham1994concrete}. Dividing \cref{eq:pr_num} by \cref{eq:pr_den} yields $p_r = \frac{(-1)^r}{r+1}$.
\end{proof}

\begin{corollary}\label{cor:slater_spectral_gap}
    Every nonzero eigenvalue of $\Xi$ must be of the form
    \[
    \theta_r = \frac{1}{3} \left( 1 + \frac{2(-1)^r}{r+1} \right) \quad \text{for some } 0 \leq r \leq \min\{\eta, n-\eta\}.
    \]
    Consequently,
    \[
    \spec(\Xi) \subset [0, 5/9] \cup \{1\}.
    \]
\end{corollary}

\begin{proof}
    Recall that the nonzero spectra of $\Xi$ and $\Theta = (\slatertestproj^{(1,2)} + 2P)/3$ agree, and we have just proven that the eigenvalues of $P$ are $p_r = \frac{(-1)^r}{r+1}$. Then the expression for the eigenvalues $\theta_r$ of $\Theta$ are immediate. One also checks that $\theta_0 = 1$, $0 \leq \theta_r < 1/3$ for odd $r$, and $1/3 < \theta_r \leq 5/9$ for all even $r \geq 2$.
\end{proof}

\subsection{Soundness for indeterminate particle number}\label{sec:slater_soundness}

The analysis above computes the eigenvalues of $\Xi$ within the fixed particle-number sectors (i.e., for weight spaces $(\eta, \eta, \eta)$). By \cref{thm:rejection-infidelity-bounds} this implies soundness for all number eigenstates. Here, we extend this soundness to the entire Fock space.

First, let us establish the fixed-number soundness.

\begin{corollary}[Soundness for fixed-number states]\label{cor:slater_soundness_fixed_number}
    For every integer $0 \leq \eta \leq n$, if $\ket{\psi} \in \bigwedge^\eta \C^n$, then the two-copy Slater tester has rejection probability at least
    \begin{equation}\label{eq:slater_soundness_fixed}
        q(\psi) \geq \frac{1}{6} \, \dist(\psi, \slaterdeterminantset)^2.
    \end{equation}
\end{corollary}

\begin{proof}
    Immediately follows from \cref{thm:rejection-infidelity-bounds,cor:slater_spectral_gap}, taking $k = 2$, $m = 3$, and $\Delta = 4/9$.
\end{proof}

Next we reduce to \cref{cor:slater_soundness_fixed_number} from arbitrary pure states. The following lemma gives a bound in terms of the mass of $\psi$ on each number sector.

\begin{lemma}\label{lem:slater_soundness_reduction}
    For any $\ket{\psi} \in \bigwedge \C^n$, let
    \[
        \ket{\psi} = \sum_{\eta=0}^n \sqrt{p_\eta} \ket{\psi_\eta}
    \]
    be an orthogonal decomposition into number sectors, i.e., $\ket{\psi_\eta}\in \bigwedge^\eta \C^n$ and $\sum_\eta p_\eta = 1$. The two-copy Slater tester on $\ket{\psi}$ has rejection probability at least
    \begin{equation}\label{eq:slater_indefinite_soundness_bound}
        q(\psi) \geq \frac{1}{2} \left(1 - \sum_{\eta=0}^n p_\eta^2\right) + \sum_{\eta=0}^n p_\eta^2 \, q(\psi_\eta).
    \end{equation}
\end{lemma}

\begin{proof}
    Let $\calH_a = \bigwedge^a \C^n$. Decompose the two-copy space as
    \[
    \bigwedge \C^n \otimes \bigwedge \C^n = \bigoplus_{a,b=0}^n \calH_a \otimes \calH_b.
    \]
    On each $(a, b)$ sector, the test operator $G = \sum_j c_j \otimes c_j^\dagger$ maps
    \begin{equation}\label{eq:G-and-Gdag-spaces}
    \begin{split}
        G &: \calH_{a} \otimes \calH_{b} \to \calH_{a-1} \otimes \calH_{b+1},\\
        G^\dagger &: \calH_{a-1} \otimes \calH_{b+1} \to \calH_{a} \otimes \calH_{b}.
    \end{split}
    \end{equation}
    Therefore $K = G^\dagger G$ preserves every $(a, b)$ sector, and hence so too does the $0$-eigenspace projector $\slatertestproj$. Thus it suffices to analyze the components of $\ket{\psi}^{\otimes 2}$ independently.

    Now we show that $\ker G \cap (\calH_a \otimes \calH_b)$ is trivial for $a > b$. This will imply that the tester rejects all such $(a, b)$ sectors with probability $1$. Define $N_1 \coloneqq N \otimes \identity$, $N_2 \coloneqq \identity \otimes N$ where $N \coloneqq \sum_{j=1}^n c_j^\dagger c_j$. Observe that
    \begin{equation}
    \begin{split}
        GG^\dagger &= \sum_{j,k} c_j c_k^\dagger \otimes c_j^\dagger c_k\\
        &= \sum_{j,k} (\delta_{jk} - c_k^\dagger c_j) \otimes c_j^\dagger c_k\\
        &= N_2 - \sum_{j,k} c_k^\dagger c_j \otimes c_j^\dagger c_k.
    \end{split}
    \end{equation}
    By an identical calculation, $G^\dagger G = N_1 - \sum_{j,k} c_j^\dagger c_k \otimes c_k^\dagger c_j$. Thus we have the identity
    \begin{equation}
        [G, G^\dagger] = N_2 - N_1.
    \end{equation}
    Now take any $\ket{v} \in \calH_a \otimes \calH_b$ and suppose $\ket{v} \in \ker G$. Since $N_1\ket{v} = a\ket{v}$ and $N_2\ket{v} = b\ket{v}$,
    \[
    \braket{v | [G, G^\dagger] | v} = (b - a) \braket{v | v}.
    \]
    But by hypothesis $G\ket{v} = 0$, so
    \[
    \begin{split}
        \braket{v | [G, G^\dagger] | v} &= \braket{v | GG^\dagger | v} - \braket{v | G^\dagger G | v}\\
        &= \|G^\dagger \ket{v}\|^2 - \|G\ket{v}\|^2\\
        &= \|G^\dagger \ket{v}\|^2 \geq 0.
    \end{split}
    \]
    Hence we must have $(b-a) \braket{v | v} \geq 0$. But if $a > b$, this is only possible if $\ket{v} = 0$.

    Now we proceed to derive the main claim. Write
    \[
    \ket{\psi}^{\otimes 2} = \sum_{a,b} \sqrt{p_a p_b} \ket{\psi_a} \otimes \ket{\psi_b}.
    \]
    As established by \cref{eq:G-and-Gdag-spaces}, $\slatertestproj$ is block-diagonal in the $(a, b)$ sectors so there are no cross terms in
    \[
    \begin{split}
        q(\psi) &= \sum_{a,b,c,d} \sqrt{p_a p_b p_c p_d} \braket{\psi_a \otimes \psi_b | (\identity - \slatertestproj) | \psi_c \otimes \psi_d}\\
        &= \sum_{a,b} p_a p_b \braket{\psi_a \otimes \psi_b | (\identity - \slatertestproj) | \psi_a \otimes \psi_b}.
    \end{split}
    \]
    Since $\ker G \cap (\calH_a \otimes \calH_b) = \{0\}$ for $a > b$, all such summands are exactly $1$. For $a = b$, the summand is $q(\psi_a)$. And for $a < b$, $\braket{\psi_a \otimes \psi_b | (\identity - \slatertestproj) | \psi_a \otimes \psi_b} \geq 0$ trivially. Therefore
    \[
    q(\psi) \geq \sum_{a>b} p_a p_b + \sum_a p_a^2 q(\psi_a).
    \]
    We get \cref{eq:slater_indefinite_soundness_bound} by rearranging indices and using $\sum_a p_a = 1$:
    \[
    \sum_{a>b} p_a p_b = \frac{1}{2} \sum_{a\neq b} p_a p_b = \frac{1}{2} \left[ \left( \sum_a p_a \right)^2 - \sum_a p_a^2 \right]. \qedhere
    \]
\end{proof}

We now prove soundness for the entire Fock space. The constant we get is identical to \cref{eq:slater_soundness_fixed}.

\begin{theorem}[Soundness for arbitrary states]\label{thm:slater_soundness_arbitrary}
    For every $\ket{\psi} \in \bigwedge \C^n$, the two-copy Slater tester has rejection probability at least
    \begin{equation}
        q(\psi) \geq \frac{1}{6} \, \dist(\psi, \slaterdeterminantset)^2.
    \end{equation}
\end{theorem}

\begin{proof}
    Let $p_\eta$, $\ket{\psi_\eta}$ be as in \cref{lem:slater_soundness_reduction}. Choose $\eta_* = \argmax_{0 \leq \eta \leq n} p_\eta$ and define
    \[
    p \coloneqq p_{\eta_*}, \quad \delta \coloneqq \dist(\psi_{\eta_*}, \slaterdeterminantset).
    \]
    Then we bound:
    \begin{align*}
        q(\psi) &\geq \frac{1}{2} \left( 1 - \sum_{\eta=0}^n p_\eta^2 \right) + \sum_{\eta=0}^n p_\eta^2 q(\psi_\eta) \tag{By \cref{lem:slater_soundness_reduction}}\\
        &\geq \frac{1}{2} \left( 1 - p \sum_{\eta=0} p_\eta \right) + p^2 q(\psi_{\eta_*}) \tag{Since $0 \leq p_\eta \leq p$}\\
        &\geq \frac{1 - p}{2} + \frac{p^2 \delta^2}{6}. \tag{By \cref{cor:slater_soundness_fixed_number} and $\sum_\eta p_\eta = 1$}
    \end{align*}
    Let $\ket{\phi} \in \slaterdeterminantset$ be a closest Slater determinant to $\ket{\psi_{\eta_*}}$; it has overlap
    \[
    \abs{\braket{\phi | \psi}}^2 = \abs{\braket{\phi | \sqrt{p} | \psi_{\eta_*}}}^2 = p(1 - \delta^2),
    \]
    hence $\dist(\psi, \slaterdeterminantset)^2 \leq 1 - \abs{\braket{\phi | \psi}}^2 = 1 - p + p\delta^2$. Compare $\frac{1}{6} \, \dist(\psi, \slaterdeterminantset)^2$ to the bound on $q(\psi)$ derived above; the difference is
    \[
        \left( \frac{1 - p}{2} + \frac{p^2 \delta^2}{6} \right) - \frac{1 - p + p\delta^2}{6} = \frac{(1 - p)(2 - p\delta^2)}{6},
    \]
    which is nonnegative since $p, \delta \in [0, 1]$.
    This implies that
    \begin{equation}
        q(\psi) \geq \frac{1 - p}{2} + \frac{p^2 \delta^2}{6} \geq \frac{1 - p + p\delta^2}{6} \geq \frac{1}{6} \, \dist(\psi, \slaterdeterminantset)^2,
    \end{equation}
    which proves the claim.
\end{proof}

\subsection{The algorithm}

Now we show that there exists an efficient quantum circuit to implement a POVM corresponding to the eigenspectrum of $K$. Afterwards we formalize the sample complexity of the tolerant tester.

For $S_j = \prod_{k<j} Z_k \otimes Z_k$, observe that under Jordan--Wigner,
\[
G = \sum_{j=1}^n S_j (\sigma^{-}_j \otimes \sigma^{+}_j).
\]
Let $\CNOT_j$ be the $\CNOT$ gate controlled on qubit $j$ of the first register and targeting qubit $j$ of the second register. Let $\CZ_{j,k}$ be the gate between qubit $j$ of the first copy and qubit $k$ of the second copy. Then define
\begin{equation}\label{eq:slater_V_circuit}
    V = RC, \quad \text{where } R = \prod_{1 \leq j < k \leq n} \CZ_{j,k} \quad \text{and} \quad C = \prod_{j=1}^n \CNOT_j.
\end{equation}
Then $V$ diagonalizes the second register of $G$. Indeed, first observe that
\[
C G C^\dagger = \sum_{j=1}^n \left( \prod_{k<j} \identity \otimes Z_k \right) \sigma_j^{-} \otimes n_j, \qquad n_j = \frac{\identity - Z_j}{2}.
\]
The $Z$ string is removed with the all-to-all $\CZ$ circuit $R$:
\[
R (CGC^\dagger) R^\dagger = \sum_{j=1}^n \sigma_j^{-} \otimes n_j.
\]
Hence
\[
V K V^\dagger = \sum_{j,k=1}^n \sigma_j^{-} \sigma_k^{+} \otimes n_j n_k.
\]

We can measure in the eigenbasis by first register by the following procedure. First we measure the second copy, which returns some bit string $b \in \{0,1\}^n$. Let $S = \{j : b_j = 1\}$. The operator projected to the attendant subspace is
\[
K_S = \sum_{j,k \in S} \sigma_j^{-} \sigma_k^{+} \equiv J_S^{-} J_S^{+}, \quad \text{where } J_S^{\pm} = \sum_{j \in S} \sigma_j^{\pm}.
\]
This defines a collective spin-$1/2$ algebra. In particular, we only need to focus on the qubits in $S$ on the first register; the resulting operators $J_S^{\pm}$ along with $J_S^z = \frac{1}{2} \sum_{j \in S} Z_j$ generate a copy of $\splunitalg(2)$ on those qubits. This is exactly the premise of the quantum Schur transform \cite{bacon2006efficient}, which can implemented to $\eps$ precision in $\poly(|S|, \log\eps^{-1})$ gates. Indeed, using the standard identity $J_S^{-} J_S^{+} = J_S^2 - (J_S^{z})^2 + J_S^{z}$, we see that $K_S$ is diagonal in the angular-momentum basis $\{\ket{J,m}\}_{J,m}$ with corresponding eigenvalues $J(J+1) - m^2 + m$. Note that $J$ and $m$ take half-integer values because we are working with the spin-$1/2$ representation of $\splunitalg(2)$.

\begin{algorithm}[H]
    \caption{Base tester for Slater determinants.}\label{alg:slater-tester}
    \KwInput{Two copies of a fermionic pure state $\ket{\psi}$.}

    \BlankLine

    Apply the circuit $V$ defined in \cref{eq:slater_V_circuit}.
    
    Measure the occupation numbers of second register and set $S\subseteq[n]$ to be the modes which returned $1$.

    Apply the quantum Schur transform on the $S$ modes in the first register and measure, obtaining angular-momentum eigenstate $\ket{J,m}$.

    \uIf{$J(J+1) - m^2 + m = 0$}{
        \KwRet{\accept}
    }
    \Else{
        \KwRet{\reject}
    }
\end{algorithm}

Because the quantum Schur transform is efficient to implement \cite{bacon2006efficient}, so too is \cref{alg:slater-tester}. Inserting this base procedure into \cref{alg:constant-copy-tester} along with the soundness guarantee of \cref{thm:slater_soundness_arbitrary} yields the constant sample complexity claim of \cref{thm:testing-slater}.

\begin{proof}[Proof of \cref{thm:testing-slater}]
    By \cref{prop:slater_completeness}, the Slater tester has perfect completeness. Let $q(\psi)$ be the rejection probability for any pure state $\ket{\psi}$. By \cref{thm:slater_soundness_arbitrary,eq:q_upper_bound_easy} we have
    \[
    \frac{1}{6} \, \dist(\psi, \slaterdeterminantset)^2 \leq q(\psi) \leq 2 \, \dist(\psi, \slaterdeterminantset)^2.
    \]
    Hence by \cref{thm:framework-tolerant-testing} the claim follows with $k = 2$ and $\beta = \sqrt{12}$.
\end{proof}

\section{Testing bosonic Gaussian states}

In this section, we prove that there is a dimension-independent tester for bosonic Gaussian states.

\begin{theorem}\label{thm:bosonic-gaussian-tester}
Let $\beta = \sqrt{405 / 2}$. For all $0 \leq \eps_1 < \eps_2 < 1$ with $\eps_2 > \beta \eps_1$ and $ 0 < \delta \leq \frac12$, there is an $O\paren*{\frac{\log (1/\delta)}{\paren*{\eps_2 - \beta \eps_1}^2}}$-copy $(\eps_1, \eps_2, \delta)$ tester for bosonic Gaussian states.
The tester acts on three copies at a time.
\end{theorem}

We start, in \cref{sec:bosonic-projector}, by constructing a three-copy projector $\bosonictestproj$ with perfect completeness.
Then, in \cref{sec:bosonic-gaussians-spectral-gap}, we show that $\bosonictestproj^{(7)}$, the symmetrization of $\bosonictestproj$ on seven registers, has a sufficiently large spectral gap.
In \cref{sec:bosonic-tester-proof}, we use this to prove \cref{thm:bosonic-gaussian-tester}.

\subsection{Three-copy tester with perfect completeness}\label{sec:bosonic-projector}

Recall that the $n$-mode bosonic gaussian states $\bosonicgaussianset(n) = \bracebar*{U(g) \ket{0}}{g \in \Jacobi(n)}$ are the orbit of an irreducible representation of the Jacobi group $\Jacobi(n) = \Heisenberg(n) \rtimes \Mp(2n)$.
Tensor powers of this representation have the following decomposition into irreps; the commutant of $\Jacobi(n)$ is $\Orth(k-1)$ (cf. $\Orth(k)$ for $\Mp(2n)$).

\begin{lemma}[Irreps of powered Jacobi representation \cite{goodman2004multiplicity}]\label{lem:jacobi-decomposition}
Let $k$ be a positive integer and $H$ the $n$-mode bosonic Fock space.
Then
\begin{align}
H^{\otimes k} & \overset{\Jacobi(n)}{\cong}
\bigoplus_{\lambda \in \mathcal I_{k}^{\Jacobi(n)}} V^{\Jacobi(n)}_{\lambda + \paren*{\frac{k}{2}, \ldots, \frac{k}{2}}} \otimes V_{\lambda}^{\Orth(k-1)},
\end{align}
where
\begin{align}
\mathcal I_{k}^{\Jacobi(n)}
&= \bracebar{\lambda}{\lambda_1 \geq \lambda_2 \geq \cdots \geq \lambda_n, \paren*{\lambda^{\transpose}}_1 + \paren*{\lambda^{\transpose}}_2 \leq k - 1}.
\end{align}
\end{lemma}

\begin{corollary}\label{cor:jacobi-decomposition-special-cases}
We have the special cases
\begin{align}
H^{\otimes 2} & \overset{\Jacobi(n)}{\cong}
\paren*{
V^{\Jacobi(n)}_{1, 1, \ldots, 1}
\otimes
V^{\Orth(1)}_{0}
}
\oplus
\paren*{
V^{\Jacobi(n)}_{2, 1, \ldots, 1}
\otimes
V^{\Orth(1)}_{1}
},
\label{eq:jacobi-decomposition-2}
\\
H^{\otimes 3} & \overset{\Jacobi(n)}{\cong}
\paren*{
V^{\Jacobi(n)}_{\frac{3}{2}, \ldots, \frac{3}{2}}
\otimes
V_{0}^{\Orth(2)}
}
\oplus
\paren*{
\bigoplus_{l=1}^{\infty}
V_{\frac32 + l, \frac32, \ldots, \frac32}^{\Jacobi(n)}
\otimes
V^{\Orth(2)}_{l}
}
\oplus
\paren*{
V_{\frac52, \frac52, \frac32, \ldots, \frac32}^{\Jacobi(n)}
\otimes
V^{\Orth(2)}_{1, 1}
}
\end{align}
(For $n < 2$, some of these irreps do not appear.)
\end{corollary}
Let $a_{i, j} = I_{H}^{\otimes (j-1)} \otimes a_i \otimes I_H^{\otimes (k-j)}$ be the bosonic annihilation operator $a_i$ acting on the $j$-th (of $k$) registers. From these we can construct explicit generators of the Heisenberg group $\Heisenberg(n)$, the metaplectic group $\Mp(2n)$, and the orthogonal group $\Orth(k-1)$.
\paragraph{Heisenberg generators}
Recall that $I$, $\braces*{a_i}_{i=1}^n$, $\braces*{a_i^{\dagger}}_{i=1}^n$ are the generators of $\Heisenberg(n)$.
Define
\begin{align}
    A_{i} &= \frac{1}{\sqrt{k}}\sum_{l=1}^k a_{i, l}, & 1 \leq i \leq n.
\end{align}
These satisfy the same commutation relations as $\braces*{a_i}_i$:
\begin{align}
\bracket*{A_{i}, A_j} &= 0, &
\bracket*{A_{i}, A_j^{\dagger}} &= \delta_{i, j} I,
\end{align}
and represent the generators of $\Heisenberg(n)$ on $H^{\otimes k}$.
\paragraph{Metaplectic generators}
Recall that the symplectic algebra corresponding to the metaplectic group has generators 
$\braces*{a_i^{\dagger} a_j, a_i a_j, a_i^{\dagger} a_j^{\dagger}}_{1 \leq i,  j \leq n}$.
Define
\begin{align}
    E_{i, j} &= \frac12 \sum_{l=1}^k \paren*{
        a_{i, l}^{\dagger} a_{j, l} +
         a_{j, l} a_{i, l}^{\dagger}
    }, &
    F_{i, j} &= \sum_{l=1}^k a_{i, l} a_{j, l}
\end{align}
for $1 \leq i, j \leq n$.
Then $\bracebar*{E_{i, j}, F_{i, j}, F_{i, j}^{\dagger}}{1 \leq i, j \leq n}$ satisfy the same commutation relations as the corresponding single-register generators and represent the generators of $\Mp(2n)$ on $H^{\otimes k}$.

\paragraph{Orthogonal generators.}
Let $\paren*{v_i}_{i=1}^k$ be an orthonormal basis of $\mathbb R^k$.
Define
\begin{align}
b_{i, p} &= \sum_{q=1}^k v_{p, q} a_{i, q}, & 1 \leq i \leq n, 1 \leq p \leq k\\
M_{p, q} &=  \sum_{i=1}^n \paren*{b_{i, p}^{\dagger} b_{i, q} - b_{i, q}^{\dagger} b_{i, p}}, & 1 \leq p < q \leq k,\\
Z &= \exp\paren*{\ri \pi N}, & N &= \sum_{i=1}^n b_{i, 1}^{\dagger} b_{i, 1}.
\label{eq:orthogonal-reflection-generator}
\end{align}
Then $\braces*{M_{p, q}}$ represent the generators of the special orthogonal algebra, and $Z$ the reflection, which together generate the orthogonal group.
If we take $v_k = \paren*{1, \ldots, 1} / \sqrt{k}$ then the remaining $\braces*{M_{p, q}}_{1 \leq p < q < k}$ (with $Z$) generate $\Orth(k-1)$.

For $k\geq 3$, but not $k = 2$,  the ``top'' irrep completely captures tensor powers of bosonic Gaussian states.
Our tester will be based on projecting on to the top irrep for $k=3$, and we denote this projector as $\bosonictestproj$.

\begin{lemma}\label{lem:top-irrep-bosonic-gaussians}
Let $k \geq 2$ be a positive integer and consider the decomposition from \cref{lem:jacobi-decomposition}.
\begin{align}
V_{\frac{k}{2}, \ldots, \frac{k}{2}}^{\Jacobi(n)} &\cong 
\overline{\Span}\bracebar{\ket{\phi}^{\otimes k}}{\text{$\ket{\phi}$ is bosonic Gaussian}}.
\label{eq:top-irrep-span-bosonic-gaussians}
\end{align}
For $k=2$,
\begin{align}
    V_{1, \ldots, 1}^{\Jacobi(n)} &=
    \overline{\Span}\bracebar{\ket{\phi}^{\otimes 2}}{\text{$\ket{\phi}$ is bosonic Gaussian}}
=
\Sym^2(H).
\label{eq:top-irrep-sym}
\end{align}
For $k\geq 3$, for any state $\ket{\psi} \in H$, $\ket{\psi}^{\otimes k} \in V^{\Jacobi(n)}_{\frac{k}{2}, \ldots, \frac{k}{2}}$ iff $\ket{\psi}$ is bosonic Gaussian.
\end{lemma}

\begin{corollary}
There is no two-copy symmetric projector $\Pi$ such that $\Pi\ket{\psi, \psi} = \ket{\psi, \psi}$ if and only if $\ket{\psi}$ is a bosonic Gaussian state.
\end{corollary}

\begin{fact}[Properties of Wigner function]\label{fct:wigner-function-properties}
Consider an $n$-mode bosonic Fock space $H$ and a normalized state $\ket{\psi} \in H$.
The Wigner function $W_{\psi}: \mathbb R^{2n} \to \mathbb R$,
\begin{align}
    W_{\psi}(x, p) &= \frac{1}{\paren{2\pi}^n} \int_{\mathbb R^n} \psi^*\paren*{x + \frac12 \xi} \psi\paren*{x - \frac12 \xi} e^{i p \cdot \xi} d^n\xi
\end{align}
has the following properties.
\begin{itemize}
    \item $W_{\psi}$ is continuous\cite{littlejohn1986semiclassical}.
    \item $\int_{\mathbb R^{2n}} dx W_{\psi}(x)=1$\cite{weedbrook2012gaussian}.
    \item Hudson's Theorem: a pure state $\ket{\psi}$ is gaussian iff $W_{\psi}(x) \geq 0$ for all $x \in \mathbb R^{2n}$\cite{hudson1974wigner,soto1983when}.
    \item $W_{\bigotimes_l \ket{\psi_l}}\paren*{x_1, \ldots, x_k} = \prod_l W_{\psi_l}(x_l)$\cite{weedbrook2012gaussian}.
    \item For an orthogonal $R \in \Orth(k)$, let $T(R)$ be its representation on $H^{\otimes k}$.
Then for $\ket{\Phi} \in H^{\otimes k}$ and $X = \paren*{x_i}_{i=1}^k \in \mathbb R^{k \times 2n}$,
$W_{T(R) \ket{\Phi}}(X) = W_{\ket{\Phi}}(R^{-1}X)$\cite{littlejohn1986semiclassical}.
\end{itemize}
\end{fact}

\begin{proof}[Proof of \cref{lem:top-irrep-bosonic-gaussians}]
For concision, let 
$V = V^{\Jacobi(n)}_{\frac{k}{2}, \ldots, \frac{k}{2}}$
and
$S = \overline{\Span}\bracebar{\ket{\phi}^{\otimes k}}{\text{$\ket{\phi}$ is bosonic Gaussian}}$.
To start, note that the (powered) vacuum state $\ket{0}^{\otimes k} \in V$, and indeed is the lowest-weight vector.
Because $V$ is the irrep space of a representation $U$ of bosonic Gaussian unitaries, any bosonic state $\ket{g}^{\otimes k} = {U(g)}^{\otimes k} \ket{0}^{\otimes k} \in V$, as is any linear combination thereof. Therefore, $\ket{\psi}^{\otimes k} \in V$ for Gaussian $\ket{\psi}$ and $S \subset V$.
But $S$ is also invariant under $U$: $U(g)S = S$ for all $g \in \Jacobi(n)$. Because $V$ is irreducible, it follows by definition that $V = S$.
This proves \cref{eq:top-irrep-span-bosonic-gaussians}.

Recall from \cref{eq:jacobi-decomposition-2} that under $\Jacobi(n)$, $H^{\otimes 2}$ decomposes into two subspaces, each associated with a 1-dimensional irrep of $\Orth(1)$. In particular, the top irrep $V$ is associated with the trivial representation.
The group $\Orth(1)$ has a single non-identity element corresponding to the $\Orth(2)$ reflection $R = \begin{pmatrix} 0 & 1 \\ 1 & 0 \end{pmatrix}$.
Let $T(R)$ be the representation $R \in \Orth(2)$ on $H^{\otimes 2}$ and consider the action of $R$ on the Wigner function of a state $\ket{\Psi} \in H^{\otimes 2}$:
\begin{align}
    W_{T(R)\ket{\Psi}}\paren*{x_1, x_2}
    &= W_{\ket{\Psi}}\paren*{x_2, x_1}.
\end{align}
That is $T(R)$ is the swap operator between the two registers.
The symmetric subspace $\Sym^2(H)$ is the $+1$-eigenspace of the swap operator, i.e., the trivial representation of $\Orth(1)$. The top irrep $V$ is the just the multiplicity space of this trivial representation in $H^{\otimes 2}$ and therefore $V = \Sym^2(H)$.
This proves \cref{eq:top-irrep-sym}.

Finally, we show that, for $ k\geq 3$ and any $\ket{\psi} \in H$, $\ket{\psi}^{\otimes k} \in V$ iff $\ket{\psi}$ is Gaussian.
We've already established, \emph{a fortiori}, that if $\ket{\psi}$ is Gaussian then $\ket{\psi}^{\otimes k} \in V$.
To establish the reverse direction, we will show that if $\ket{\psi}^{\otimes k} \in V$ then 
the Wigner function $W_{\psi}: \mathbb R^{2n} \to \mathbb R$ of $\ket{\psi}$ is non-negative, which by Hudson's Theorem implies that $\ket{\psi}$ is Gaussian.

Now recall that $\Orth(k-1) \cong \bracebar*{M \in \Orth(k)}{M\paren*{1, \ldots, 1} = \paren*{1, \ldots, 1}}$ acts trivially on $V$. (See \cref{lem:jacobi-decomposition}.)
So for any $\ket{\Phi} \in V$ and $R \in \Orth(k-1)$,
\begin{align}
T(R) \ket{\Phi} &= \ket{\Phi}, & W_{\ket{\Phi}}(X) &= W_{\ket{\Phi}}(RX).
\end{align}
In particular, define
\begin{align}
R &= \begin{pmatrix}
    -\frac13 & \frac23 & \frac23 & 0 & 0 & \cdots \\
    \frac23 & -\frac13 & \frac23 & 0 & 0 & \cdots \\
    \frac23 & \frac23 & -\frac13 & 0 & 0 & \cdots \\
    0 & 0 & 0 & 1 & 0 & \cdots \\
    0 & 0 & 0 & 0 & 1 & \cdots \\
    \vdots & \vdots & \vdots & \vdots & \vdots & \ddots
\end{pmatrix}.
\end{align}
Note that $R$ is orthogonal, self-inverse ($R^{\transpose} = R^{-1} = R$), and $R\paren*{1, \ldots, 1} = \paren*{1, \ldots, 1}$, i.e., $R \in \Orth(k-1)$; it acts non-trivially only on the first three registers. 
(This is also where we require $k\geq 3$.)

Consider any $x, y \in \mathbb R^{2n}$ such that $W_{\psi}(y) > 0$; there must always be some $y$ (because $\int W_{\psi} = 1$).
Define 
\begin{align}
x^{(p)} &= y + \paren*{-\frac13}^p \paren*{x - y}, &
y^{(p)} &= y + \frac23 \paren*{x^{(p)} - y}.
\end{align}
Then 
\begin{align}
&
W_{\psi}\paren*{x^{(p)}}
W_{\psi}\paren*{y}^{k-1}
=
W_{\psi^{\otimes k}}(x^{(p)}, y, \ldots, y)
=
W_{T(R)\psi^{\otimes k}}(x^{(p)}, y, \ldots, y)
\\
&=
W_{\psi^{\otimes k}}\paren*{
x^{(p + 1)}, y^{(p)}, y^{(p)}, y, \ldots, y}
=
W_{\psi}\paren*{x^{(p + 1)}}
W_{\psi}\paren*{
    y^{(p)}
}^2
W_{\psi}\paren*{y}^{k-3}.
\end{align}
For $W_{\psi}\paren*{x^{(p)}} \neq 0$, 
\begin{align}
    \sgn\paren*{W_{\psi}\paren*{x^{(p)}}}
    &=
    \sgn\paren*{W_{\psi}\paren*{x^{(p)}}
        W_{\psi}(y)^{k-1}
    }
    \\
    &=
    \sgn\paren*{
W_{\psi}\paren*{x^{(p + 1)}}
W_{\psi}\paren*{
    y^{(p)}
}^2
W_{\psi}\paren*{y}^{k-3}
}
=
\sgn\paren*{W_{\psi}\paren*{x^{(p + 1)}}}.
\end{align}
Therefore, if $W_{\psi}\paren*{x^{(p)}} < 0$, then $W_{\psi}\paren*{x^{(p + 1)}} < 0$.

Because $W_{\psi}$ is continuous,
\begin{align}
\lim_{p \to \infty}
W_{\psi}\paren*{x^{(p)}}
=
W_{\psi}\paren*{\lim_{p\to\infty} x^{(p)}}
=
W_{\psi}\paren*{y}.
\end{align}
If $W_{\psi}\paren*{x} = W_{\psi}\paren*{x^{(0)}} < 0$, then $W_{\psi}\paren*{x^{(1)}}, W_{\psi}\paren*{x^{(2)}}, \ldots < 0$ by induction.
The limit is $W_{\psi}\paren*{y} \leq 0$. 
But $W_{\psi}\paren*{y} > 0$ by construction, so it must be that $W_{\psi}(x) \geq 0$ for all $x$.
\end{proof}

\subsection{Spectral gap analysis}\label{sec:bosonic-gaussians-spectral-gap}

Now we show that the symmetrization of the 3-copy $\bosonictestproj$ on 7 copies has a spectral gap.

\begin{lemma}[Spectral gap for bosonic Gaussians]\label{lem:bosonic-gaussians-spectral-gap}
Let $H$ be the $n$-mode bosonic Fock space.
There is a symmetric projector $\bosonictestproj$ on $H^{\otimes 3}$ such that $1 - \bosonictestproj^{(7)}$ has a spectral gap of $\frac{7}{45}$.
\end{lemma}

\begin{proof}
    Follows from \cref{lem:bosonic-bound} with $m=7$, where $\bosonictestproj^{(7)} \cong \overline{\ReynoldsOp}$.

\end{proof}

\subsubsection{Spherical harmonics}

This section is based on \cite{dai2013approximation}.

\begin{definition}[Reynolds operator]
Let $\rho$ be a representation of a group $G$ with normalized Haar measure $dg$.
The \emph{Reynolds operator} is
\begin{align}
\ReynoldsOp_G(\rho) &= \int_G \rho(g) dg.
\end{align}
Where the representation $\rho$ is clear from context, we may write simply $\ReynoldsOp_G = \ReynoldsOp_G(\rho)$.
\end{definition}

\begin{definition}[Stabilizers]
Let $v_1, \ldots, v_m \in {\mathbb R}^d$.
The stabilizer of $\braces*{v_i}_i$ in $\Orth(d)$ is
\begin{align}
\Stab_{\Orth(d)}\paren*{v_1, \ldots, v_m} &= \bracebar*{R \in \Orth(d)}{\text{$Rv_i = v_i$ for $i=1,\ldots,m$}}.
\end{align}
\end{definition}
Note that $\Stab_{\Orth(d)}\paren*{v_1, \ldots, v_m} \cong \Orth\paren*{d - \dim \Span\braces*{v_1, \ldots, v_m}}$.

\begin{fact}\label{lem:spherical-harmonic-projectors}
Let $x \in {\mathbb R}^d$ be a unit vector and $(\rho, V)$ be an irreducible representation of $\Orth(d)$.
Then $\rank \ReynoldsOp_{\Stab_{\Orth(d)}\paren*{x}} \in \{0, 1\}$.
If $\rank \ReynoldsOp_{\Stab_{\Orth(d)}\paren*{x}} = 1$, then there is a function $\phi: {\mathcal S}^{d-1} \to V$ such that for all normalized $y \in {\mathbb R}^d$,
\begin{align}
\ReynoldsOp_{\Stab_{\Orth(d)}\paren*{y}} &= \ket{\phi(y)}\bra{\phi(y)}, &
\braket{\phi(x) | \phi(y)} &= \frac{
    C_l^{(d-2)/2}\paren*{x \cdot y}
}{
    C_l^{(d-2)/2}\paren*{1}
}.
\label{eq:spherical-harmonic-overlap}
\end{align}
\end{fact}

\begin{proof}
For concision, let 
\begin{align}
G &= \Stab_{\Orth(d)}(x),\\
T &= \ReynoldsOp_{G}\paren*{\rho},\\
W &= \bracebar*{v \in V}{
    \text{
        $\rho(R) v = v$ for all 
        $R \in G$}}.
\end{align}
$T$ is invariant under $\rho$ ($\rho T = T$), Hermitian, and the identity on $W$. 
Therefore, $T$ is the projector onto $W$.

For each non-zero $y \in W$ and some $l$, define a function $f: V \to \mathcal Y_l$ by
\begin{align}
    f(w): Rx & \mapsto \braket{\rho(R) y | w} \label{eq:spherical-harmonic-intertwiner}
\end{align}
for all $w \in V$ and $R \in \Orth(d)$,
where $\mathcal Y_l$ is the space of degree-$l$ homogoneous spherical harmonics on $\mathcal S^{d-1}$; $\mathcal Y_l$ is an irrep of $\Orth(d)$\cite{dai2013approximation}.
Consider some $x' \in \mathcal S^{d-1}$ and $R_1, R_2 \in \Orth(d)$ be such that $x' = R_1 x = R_2 x$.
Then $R_2^{-1} R_1 \in G$ and thus $\rho\paren*{R_2^{-1} R_1} y = y$, i.e., $\rho(R_1)y = \rho(R_2)y$, so $f$ is well defined.
Because $f$ is an intertwiner between the the irreps $(\rho, V)$ and $\mathcal Y_l$, it is either 0 or an isomorphism.
Because $f(w) = \norm{w}^2 = 1 \neq 0$ there is some $l$ such that $V \cong \mathcal Y_l$.

Any degree$-l$ spherical harmonic $Y$ that is invariant under $G$ satisfies~\cite[Lemma 7.1]{dai2013approximation}
\begin{align}
Y(y) & \propto C^{\paren*{\frac{d-2}{2}}}_l \paren*{x \cdot y}.
\end{align}
Therefore, $\dim W \in \{0, 1\}$.
When $\dim W = 1$, set $\ket{\phi(x)}$ to the corresponding spherical harmonic.
Then $\braket{\phi(x) | \phi(y)} \propto C^{\paren*{\frac{d-2}2}}_l \paren*{x \cdot y}$\cite[Theorem 2.6]{dai2013approximation}.
Normalization yields \cref{eq:spherical-harmonic-overlap}.
\end{proof}

\subsubsection{One-step bounds}

Now we consider symmetrizations of the lowest-irrep-projector over \emph{one} more register. (To get a sufficiently strong spectral gap, we'll need to symmetrize over several registers; we will use the one-step bound iteratively to get that in the next section.)

\begin{lemma}\label{lem:bosonic-one-step-bound}
Let $d\geq 2$ be a positive integer and $v_1, \ldots, v_{d+1} \in \mathbb R^d$ be unit vectors such that $v_i \cdot v_j = -1/d$ for $i \neq j$, and $(\rho, V)$ a representation of $\Orth(d)$. 
Define
\begin{align}
\overline{\ReynoldsOp}
&=
\frac{1}{d + 1}
\sum_{i=1}^{d + 1} \ReynoldsOp_{\Stab_{\Orth(d)}(v_i)}\paren*{\rho}.
\end{align}
Then
\begin{align}
    I - \overline{\ReynoldsOp} &\succeq c \paren*{I - \ReynoldsOp_{\Orth(d)}(\rho)}, &
    c &= \frac{(d-2)(d+1)}{d^2}.
\label{eq:bosonic-one-step-bound}
\end{align}
\end{lemma}

\begin{proof}
If the statement is true when $\rho$ is irreducible, then it is true in general; it holds uniformly on each of the orthogonal irrep spaces. 
If $\ReynoldsOp_{\Stab_{\Orth(d)}(v_i)}=0$ for some $i$, the same holds for all $i$ by symmetry and thus $\overline{\ReynoldsOp} = 0$, in which case the inequality is trivially true because $c, \norm{\ReynoldsOp_{\Orth(d)}} \leq 1$. 
Henceforth, we will assume that $\rho$ is irreducible, and thus corresponds to spherical harmonics of some degree $l$.
For concision, let
\begin{align}
P_i = \ReynoldsOp_{\Stab_{\Orth(d)}(v_i)}(\rho);
\end{align}
there is some $P$ such that $P_i = \rho(R_i) P \rho(R_i)^{\dagger}$ with $R_i \in \Orth(d)$.
By \cref{lem:spherical-harmonic-projectors}, 
\begin{align}
\overline{\ReynoldsOp} \paren*{\rho}
&=
\frac{1}{d + 1}
\sum_{i=1}^{d + 1} \ket{\phi(v_i)} \bra{\phi(v_i)}
=
\frac{1}{d + 1} L L^{\dagger},
\end{align}
where
\begin{align}
L &= \begin{pmatrix} \phi(v_1) & \cdots & \phi(v_{d+1}) \end{pmatrix}.
\end{align}
For $i \neq j$, let
\begin{align}
x_l &= \braket{\phi(v_i) | \phi(v_j)} = 
\frac{
    C^{\paren*{\frac{d-2}{2}}}_l \paren*{-1/d}
}{
    C^{\paren*{\frac{d-2}{2}}}_l \paren*{1}
}.
\end{align}
be the ($i,j$-independent) overlaps.
The Gram matrix is
\begin{align}
G &= L^{\dagger} L, & G_{i, j} &= (1 - x_l) \delta_{i, j} + x_l.
\end{align}
The eigenvalues of $G$ are $1 + dx_l$ and $1-x_l$, so the non-zero eigenvalues of $\overline{\ReynoldsOp}$ are
\begin{align}
    \lambda_{l, +} &= \frac{1 + dx_l}{d + 1}, &
    \lambda_{l, -} &= \frac{1 - x_l}{d + 1}.
\end{align}
Because the Gram matrix $G$ is positive-definite, we know that $x_l \geq -1 / d$.
Therefore $\lambda_{l, -} \leq 1 / d$.
Now we'll upper bound 
\begin{align}
    x_l &= 
    \frac{\Gamma\paren*{\alpha + 1}}{\pi^{1/2} \Gamma\paren*{\alpha + \frac12}}
    \int_0^{\pi} \paren*{-\frac{1}{d} + i \frac{\sqrt{d^2 - 1}}{d} \cos\phi}^l \paren*{\sin\phi}^{2\alpha} d\phi
    \\
    &=
    d^{-l}
    \mathbb E_z\bracket*{
    \paren*{-1 + i \sqrt{d^2-1} z}^l}
\end{align}
where $\alpha = (d-3)/2$, $\cos\theta = -1/d$, and $z = \cos \phi$ is distributied according to 
\begin{align}
    p(z) &= 
    \frac{\Gamma\paren*{\alpha + 1}}{\pi^{1/2} \Gamma\paren*{\alpha + \frac12}}
    \paren*{\sin\phi}^{2\alpha}.
\end{align}
The first few values are
\begin{align}
x_1 & = x_2 = -\frac1d, & x_3 &= \frac{2 + 3d}{d^3}.
\end{align}
The first two moments of $z^2$ are
\begin{align}
\mathbb E_z\bracket*{z^2} &= \frac{1}{d-1}, &
\mathbb E_z\bracket*{z^4} &= \frac{3}{(d-1)(d+1)}.
\end{align}
For $l \geq 4$,
\begin{align}
    \abs*{x_l} &\leq 
    \mathbb E_z\bracket*{
        \paren*{\frac{1 + \paren*{d^2 - 1} z^2}{d^2}}^{l/2}
    }
    \\
    &\leq
    \mathbb E_z\bracket*{
        \paren*{\frac{1 + \paren*{d^2 - 1} z^2}{d^2}}^{2}
    }
    \\
    &=
    \frac{1}{d^4} + 2 \frac{d+1}{d^4} + 3 \frac{d^2-1}{d^4} = 
    \frac{3d+2}{d^3} = x_3.
\end{align}
Define
\begin{align}
\lambda_{\star} &= \frac{d + 2}{d^2}.
\end{align}
Then
\begin{align}
    \lambda_{l, +} & \leq \frac{1 + d(3d + 2) / d^3}{d + 1} = \lambda_{\star}, &
    \lambda_{l, -} & \leq \frac{1}{d} < \lambda_{\star}. 
\end{align}
That is $\overline{\ReynoldsOp} \preceq \lambda_{\star} I$ for nontrivial $\rho$.

Let $W \subset V$ be the support of $\ReynoldsOp_{\Orth(d)}(\rho)$.
On $W$,
\begin{align}
    \left. P_i \right|_{W} &= I_W, &
        \left.\overline{\ReynoldsOp}\right|_W &= I_W.
\end{align}
On $W^{\perp}$ (as over all $V$),
\begin{align}
\left.\overline{\ReynoldsOp}\right|_{W^{\perp}} &\preceq \lambda_{\star} I_{W^{\perp}}.
\end{align}
Therefore
\begin{align}
\overline{\ReynoldsOp} &\preceq \ReynoldsOp_{\Orth(d)} + \lambda_{\star} \paren*{I - \ReynoldsOp_{\Orth(d)}},
\end{align}
i.e.,
\begin{align}
    I - \overline{\ReynoldsOp} & \succeq c \paren*{I - \ReynoldsOp_{\Orth(d)}},
\end{align}
where $c=1 - \lambda_{\star}$.
\end{proof}

\subsubsection{General bound}

For a subspace $V \subset W \cong \mathbb R^d$, we write $\Orth\paren*{V} \cong \Orth\paren*{\mathbb R^{\dim V}} = \Orth(\dim V)$ to mean the orthogonal group that acts trivially on $V^{\perp} \subset W$.

\begin{lemma}\label{lem:bosonic-bound}

    For integer $m \geq 3$, define
    \begin{align}
        W_{m, I} &= 
        \bracebar*{x \in \mathbb R^m}{\sum_{i=1}^m x_i = 0; \text{$x_i = 0$ for $i \notin I$} }, \\
        W_m &= W_{m,[m]},\\
        W_{m, -i} &= W_{m, [m]\setminus\braces{i}}, \text{ for $i \in [m]$}.
    \end{align}
 Let $\rho$ be a representation of $\Orth(W_m)$ and define
    \begin{align}
    \overline{\ReynoldsOp}
    &= \frac{1}{\binom{m}{3}} \sum_{\substack{I \subset [m] \\ \abs{I} = 3}} \ReynoldsOp_{\Orth\paren*{W_{m, I}}}(\rho).
    \end{align}
Then
\begin{align}
I - \overline{\ReynoldsOp} & \succeq c_m \paren*{I - \ReynoldsOp_{\Orth\paren*{W_m}}(\rho)}, & c_m = \frac{2m}{3(m-1)(m-2)}.
\end{align}
\end{lemma}

\begin{proof}
We will prove the bound by induction.
For $m=3$, $\overline{\ReynoldsOp} = \ReynoldsOp_{\Orth\paren*{W_3}}$ and $c_3 = 1$, so the two sides of the inequality are equal.
Let
\begin{align}
    \overline{\ReynoldsOp}_i &= \frac{1}{\binom{m-1}{3}} \sum_{\substack{
            I \subset [m] \setminus \{i\} \\
            \abs{I} = 3
    }} \ReynoldsOp_{\Orth\paren*{W_{m, I}}}.
\end{align}
Suppose the inequality holds for $m-1$ when $m \geq 4$:
\begin{align}
    I - 
   & \overline{\ReynoldsOp}_i
   \succeq c_{m-1} \paren*{I - \ReynoldsOp_{\Orth\paren*{W_{m, -i}}}},
\end{align}
because $W_{m, -i} \cong W_{m -1}$.
Then
\begin{align}
    I - \overline{\ReynoldsOp} &= I - \frac{1}{m} \sum_{i=1}^m \overline{\ReynoldsOp}_i
\succeq c_{m-1} \paren*{I - \frac{1}{m} \sum_i \ReynoldsOp_{\Orth\paren*{W_{m, -i}}}},
\end{align}
where we used the fact that
\begin{align}
    \overline{\ReynoldsOp}
    &= \frac{1}{\binom{m}{3}} \sum_{\substack{
            I \subset [m] \\ \abs{I} = 3
    }} \ReynoldsOp_{\Orth\paren*{W_{m,I}}}
    =
\frac{1}{\binom{m}{3}} \sum_{\substack{
            I \subset [m] \\ \abs{I} = 3
    }} 
    \frac{1}{m-3} \sum_{
    j \in [m] \setminus I} 
    \ReynoldsOp_{\Orth\paren*{W_{m,I}}}
    \\
    &=
    \frac{1}{m} \sum_{i \in [m]}
    \frac{1}{\binom{m-1}{3}}
    \sum_{\substack{I \subset [m] \setminus \{i\} \\ \abs{I} = 3}}
\ReynoldsOp_{\Orth\paren*{W_{m,I}}}
= \frac1m \sum_{i \in [m]} \overline{\ReynoldsOp}_i.
\end{align}
For $j \in [m]$, define
\begin{align}
u_j &= \paren*{
    u_{j, l}
}_{l=1}^m, &
    u_{j, l} &= \sqrt{\frac{m}{m-1}} \paren*{\delta_{j, l} - \frac{1}{m}}.
\end{align}
Each $u_j \in W_m$ and is normalized, and $u_j \cdot u_{j'}= -1/(m-1)$ for $j \neq j'$.
Note that $\Orth\paren*{W_{m, -j}} = \Stab_{\Orth\paren*{W_m}}(u_j)$.
We then apply \cref{lem:bosonic-one-step-bound} with $d=m-1$, 
$\Orth(d) \cong \Orth\paren*{W_m}$:
\begin{align}
    I - \frac{1}{m} \sum_{i=1}^m \ReynoldsOp_{\Orth\paren*{W_{m,-i}}} & \succeq 
    \frac{(m-3)m}{\paren{m-1}^2} \paren*{I - \ReynoldsOp_{\Orth\paren*{W_m}}},
\end{align}
where the present $\ReynoldsOp_{\Orth\paren*{W_{m, -i}}}$ and $\ReynoldsOp_{\Orth\paren*{W_m}}$ are respectively the Reynolds operators on the left and right sides of \cref{eq:bosonic-one-step-bound}.
Therefore
\begin{align}
    I - \overline{\ReynoldsOp} & \succeq 
    c_{m-1} \frac{m(m-3)}{\paren{m-1}^2} \paren*{I - \ReynoldsOp_{\Orth\paren*{W_m}}} = 
    c_{m}  \paren*{I - \ReynoldsOp_{\Orth\paren*{W_m}}}.
\end{align}

\end{proof}

\subsection{The algorithm}\label{sec:bosonic-tester-proof}

\begin{algorithm}[H]
    \caption{Base tester for general bosonic Gaussian states.}\label{alg:bosonic-gaussian-state-tester}
    \KwInput{Three copies of a bosonic pure state $\ket{\psi}$.}

    \BlankLine

    \BlankLine

    Apply the bosonic gaussian unitary $U$ that rotates into a basis in which $M$ (see \cref{eq:rotation-gen-eigenbasis}) is diagonal.

    Measure in the first two registers in the Fock basis $\paren*{
    n_{i, j}}_{\substack{ i \in [n] \\ j \in \{1, 2\}}}$ to get outcomes $\paren*{x_{i, j}}_{i, j}$.
    \uIf{$\sum_{i=1}^n \paren*{x_{i, 1} - x_{i,2}} = 0$}{
        \KwRet{\accept}
    }
    \Else{
        \KwRet{\reject}
    }
\end{algorithm}

\begin{proof}[Proof of \cref{thm:bosonic-gaussian-tester}]
Let $\bosonictestproj$ be the three-copy projector (with perfect completeness), 
$R = I - \bosonictestproj$ and 
$r\paren*{\psi} = \braket{\psi, \ldots, \psi | R | \psi, \ldots, \psi}$.
Let $\bosonicgaussianset$ be the set of pure $n$-mode bosonic Gaussian states.
By \cref{lem:bosonic-bound}, the spectral gap of $R^{(7)}$ is $7/45$.

By \cref{thm:rejection-infidelity-bounds} with $k=3$ and $m=7$,
\begin{align}
    \alpha \epsilon^2 &\leq r \leq 3 \epsilon^2, &
\eps &=
\dist(\ket{\psi}, \bosonicgaussianset), &
\alpha &=\frac{2}{135},
\end{align}
and by \cref{thm:framework-tolerant-testing}, running
\cref{alg:constant-copy-tester} using \cref{alg:bosonic-gaussian-state-tester} as the subroutine $\calA$ yields a $t$-copy $(\epsilon_1, \epsilon_2, \delta)$-tester for the pure bosonic Gaussian states $\bosonicgaussianset$.
\end{proof}

\paragraph{Implementing the projector}
The projector $\bosonictestproj$ corresponds to the trivial representation of $\Orth(2) \subset \Orth(3)$.
Let 
\begin{align}
v_+ = v_1 &= \begin{pmatrix} 1 & -1 & 0 \end{pmatrix} / \sqrt{2},\\
v_- = v_2 &= \begin{pmatrix} 1 & 1 & -2 \end{pmatrix} / \sqrt{6},\\
v_3 &= \begin{pmatrix} 1 & 1 & 1 \end{pmatrix} / \sqrt{3}.
\end{align}
The commutant acts on $\Span\braces*{v_+, v_-}$.
Define
\begin{align}
    b_{i, +} &= \frac{1}{\sqrt{2}} \paren*{a_{i, 1} - a_{i, 2}},\\
b_{i, -} &= \frac{1}{\sqrt{6}} \paren*{a_{i, 1} + a_{i, 2} - 2 a_{i, 3}},
\end{align}
so that the generators of the commutant $\Orth(2)$ are
\begin{align}
    M &= \sum_{i=1}^n \paren*{b_{i,+}^{\dagger} b_{i, -} - b_{i,-}^{\dagger} b_{i, +}}, &
    Z &= \exp\paren*{\frac{\ri \pi}{2} \sum_{i=1}^n 
        \paren*{a_{i,1}^{\dagger} - a_{i, 2}^{\dagger}}
        \paren*{a_{i,1} - a_{i, 2}}
        } = \SWAP_{1, 2}.
        \label{eq:bosonic-commutant-gen-2}
\end{align}
Let $R = \ri M$ be the Hermitian observable.
On the symmetric subspace, $Z$ acts as the identity, so we only need to observe $R$.
Define
\begin{align}
\tilde{a}_{i, 1} &= \paren*{b_{i, +} + i b_{i, -}}/\sqrt{2},\\
\tilde{a}_{i, 2} &= \paren*{b_{i, +} - i b_{i, -}}/\sqrt{2},\\
\tilde{a}_{i, 3} &= \paren*{a_{i, 1} + a_{i, 2} + a_{i, 3}}/\sqrt{3},
\end{align}
so that
\begin{align}
    R &= \sum_{i=1}^{n} \paren*{
        {\tilde{a}}_{i, 1}^{\dagger} \tilde{a}_{i, 1} - 
    {\tilde{a}}_{i,2}^{\dagger} \tilde{a}_{i,2} 
}.
\end{align}
There is a (passive) gaussian change of basis $U$ (over all three registers, taking $\paren*{\tilde{a}_{i, j}}_{i, j}$ to $\paren*{a_{i, j}}_{i, j}$) such that 
\begin{align}
    U R U^{\dagger} &= \sum_{i=1}^n \paren*{ a_{i, 1}^{\dagger} a_{i, 1} - a_{i, 2}^{\dagger} a_{i, 2}}.
\label{eq:rotation-gen-eigenbasis}
\end{align}
The intersection of the kernel of $R$ and the even subspace of $Z$ is the trivial representation of $\Orth(2)$ and thus also the ``top'' irrep space of $\Jacobi(n)$.
The entire symmetric subspace is in the even subspace of $Z$, so it suffices to just check $R$.

\section{Testing zero-mean bosonic gaussian states}\label{sec:zero-mean-boson}

In this section, we show that bosonic Gaussian states with zero displacement can be tested in constantly many copies.

\begin{theorem}\label{thm:testing-zero-mean-bosonic-gaussians}
    Let $\beta = 4$. For all $0 \leq \eps_1 < \eps_2 < 1$ with $\eps_2 > \beta \eps_1$ and $0 < \delta < \frac{1}{2}$, there is an $O\paren*{\frac{\log(1/\delta)}{(\eps_2 - \beta\eps_1)^2}}$-copy $(\eps_1, \eps_2, \delta)$-tester for zero-mean bosonic Gaussian states. The tester acts on two copies at a time.
\end{theorem}

To prove this, we follow a similar strategy used to analyze the fermionic Gaussian property tester from \cref{sec:fermionic-gaussian-tester}. 
We first describe our tester, which is essentially a bosonic analogue of the Bravyi operator. 
Specifically, we will use the following operator introduced by Girardi, Witteveen, Mele, Bittel, Oliviero, Gross, and Walter \cite{girardi2025gaussiantesting}:
\begin{equation}\label{eq:bosonic-gaussian-test-projector}
    \Gamma = \ri \sum_{i=1}^n a_i\otimes a_i^\dagger - a_i^\dagger \otimes a_i = \ri \sum_{i=1}^n \Gamma_i.
\end{equation}

Throughout this section, it will be convenient to denote by $a_{s,i}$ the lowering operator for the $i$th mode of register $s \in [3]$. In this notation, $\Gamma^{(1,2)}$ becomes
\[
\Gamma^{(1,2)} = \ri \sum_{i=1}^n a_{1,i}a_{2,i}^\dagger - a_{1,i}^\dagger a_{2,i}.
\]

Let $\Pi_{\ker \Gamma}$ be the projector onto the $0$-eigenspace of $\Gamma$. Our property tester for zero-mean bosonic Gaussian states will be the projector onto the symmetrization of $\Pi_{\ker \Gamma}$:
\[
\zeromeanbosonictestproj^{(1,2)} = \symproj^{(1,2)}\Pi_{\ker \Gamma^{(1,2)}}.
\]

This projector has a desirable attribute central to many property testers: it lies in the second-order commutant of the group of zero-mean bosonic Gaussian unitaries.

\begin{proposition}\label{prop:zero-mean-bosonic-projector-commutant}
    Let $U$ be any bosonic Gaussian unitary with no displacement component. Then
    \begin{equation}
        [U \otimes U, \zeromeanbosonictestproj] = 0.
    \end{equation}
\end{proposition}
\begin{proof}
It suffices to show that $[U \otimes U, \Gamma] = 0$. Indeed, if this holds, then $[U \otimes U, \Pi_{\ker \Gamma}] = 0$, and since $U \otimes U$ is symmetric, $[U \otimes U, \symproj^{(1,2)}] = 0$.

We will show the equivalent statement
\begin{equation}
(U \otimes U)\Gamma(U \otimes U)^\dagger = \Gamma.
\end{equation}
There exists a matrix $S \in \symplgrp(2n)$ such that, for all $i \in [2n]$
\[
U v_i U^\dagger = \sum_{j=1}^{2n} S_{ij} v_j,
\]
where $v = (a_1,
\ldots,a_n,a_1^\dagger,\ldots,a_n^\dagger)$.  Now, because $S \in \symplgrp(2n)$, it satisfies $S^\transpose \Omega S = \Omega$, where $\Omega$ is the symplectic form in equation \eqref{eq:symplectic-form}.
Thus, we compute
\begin{align}
    (U \otimes U)\Gamma(U \otimes U)^\dagger &= \ri \sum_{i=1}^n Ua_{i}U^\dagger \otimes Ua_{i}^\dagger U^\dagger - Ua_{i}^\dagger U^\dagger \otimes Ua_{i} U^\dagger \\
    &= \ri \sum_{i=1}^n \sum_{r,s=1}^{2n} S_{i,r}S_{n+i,s} v_r \otimes v_s  -  S_{n+i,r} S_{i,s} v_r \otimes v_s \\
    &= \ri \sum_{r,s=1}^{2n} \left(\sum_{i=1}^n S_{i,r}S_{n+i,s}  -  S_{n+i,r} S_{i,s} \right) v_r \otimes v_s \\
    &= \ri \sum_{r,s=1}^{2n} (S^\transpose \Omega S)_{r,s} v_r \otimes v_s \\
    &= \ri \sum_{r,s=1}^{2n} \Omega_{r,s} v_r \otimes v_s \\
    &=
    \ri \sum_{r = 1}^{n} v_{r} \otimes v_{n + r} - v_{n + r} \otimes v_{r} \\
    &= \ri \sum_{r = 1}^{n} a_{r} \otimes a_{r}^\dagger - a_{r}^\dagger \otimes a_{r} = \Gamma
\end{align}
and we are done.
\end{proof}

\subsection{Perfect completeness}

We first prove perfect completeness of our property tester. This was originally shown in \cite{girardi2025gaussiantesting}, but we give a self-contained proof here. 

\begin{lemma}\label{lem:perfect-completeness-zero-mean}
    If $\ket{\psi} \in \zeromeanbosonicgaussianset$ then $\norm{\zeromeanbosonictestproj \ket{\psi^{\otimes 2}}}^2 = 1$. In other words, our property tester always accepts.
\end{lemma}

\begin{proof}
First, we consider the vacuum state $\ket{0^n}$.
    We observe that since each term of $\Gamma$ has a lowering operator $a_i$, it annihilates $\ket{0^n}\otimes \ket{0^n}$, and so $\ket{0^n}\otimes \ket{0^n}$ is in the symmetric $0$-eigenspace of $\Gamma$. Thus the test always accepts.
    Now let $\ket{\psi}$ be any Gaussian pure state with mean $0$. There exists a Gaussian unitary $U$ with zero displacement such that $\ket{\psi} = U \ket{0^n}$. Thus the acceptance probability of the test is
    \begin{align}
        \norm{\zeromeanbosonictestproj \ket{\psi^{\otimes 2}}}^2 &= \norm{\zeromeanbosonictestproj (U \otimes U) \ket{0^n}\!\ket{0^n}} \\
        &= \braket{0^n,0^n| (U \otimes U)^\dagger \zeromeanbosonictestproj (U \otimes U)|0^n,0^n} \\
        &= \braket{0^n,0^n| (U \otimes U)^\dagger (U \otimes U) \zeromeanbosonictestproj |0^n,0^n} \label{eq:apply-zero-mean-bosonic-test-commutant} \\
         &= \braket{0^n,0^n| \zeromeanbosonictestproj |0^n,0^n} = 1,
    \end{align}
    where in equation \eqref{eq:apply-zero-mean-bosonic-test-commutant} we have applied \cref{prop:zero-mean-bosonic-projector-commutant}.
    Thus, when $\ket{\psi} \in \zeromeanbosonicgaussianset$, we accept with certainty.
\end{proof}

\subsection{Spectral gap analysis}
Our approach will mirror that of the previous two sections. We define the operators
\[
\Xi = \symproj^{(1,2,3)}\zeromeanbosonictestproj^{(1,2)}\symproj^{(1,2,3)}, \quad \Theta = \zeromeanbosonictestproj^{(1,2)}\symproj^{(1,2,3)}\zeromeanbosonictestproj^{(1,2)},
\]
observe that they have identical eigenspectra, and decompose
\[
\Theta = \frac{\zeromeanbosonictestproj^{(1,2)} + 2Q}{3},\quad Q = \zeromeanbosonictestproj^{(1,2)}\SWAP^{(2,3)}\zeromeanbosonictestproj^{(1,2)}.
\]
By the analysis in \cref{sec:spectral-gap-abstract-proof}, demonstrating a spectral gap on $Q$ implies a spectral gap on $\Xi$, which gives a constant sample complexity upper bound for our tester.  Concretely, our goal will be to show the following theorem.

\begin{theorem}\label{lem:spec_gap_zero_mean}
    Let $\lambda_1$ be the largest eigenvalue of $\Xi$ and let $\lambda$ be any other (distinct) eigenvalue of $\Xi$. Then $\lambda_1 = 1$ and $\lambda \in [0, 7/12]$.
\end{theorem}

We begin by defining invariant subspaces under $\Gamma$ and use this to decompose the Hilbert space into finite-dimensional subspaces, where $\Xi$ can be block-diagonalized by representation-theoretic methods. Choose $k \in \N$ and define the total-occupation number sectors
\[
V_k = \linspan\{\ket{x}\ket{y}\ket{z}: x,y,z \in \N^n,\quad \sum_{i=1}^n x_i + y_i + z_i = k\}
\]

\begin{lemma}\label{lem:particle-number-invariant-subspace}
    The $V_k$ are preserved under $\SWAP^{(2,3)}, \symproj^{(1,2,3)},$ and $\zeromeanbosonictestproj^{(1,2)}$.
\end{lemma}
\begin{proof}
    The statement is trivial for $\SWAP^{(2,3)}$ and $\symproj^{(1,2,3)}$. We now show that $\zeromeanbosonictestproj^{(1,2)}$ preserves the $V_k$. Each term of $\Gamma$ can be written as $\ri(a_{1,i}a_{2,i}^\dagger - a_{1,i}^\dagger a_{2,i})$. For any vector $\ket{x}\!\ket{y}\!\ket{z} \in V_k$ that isn't annihilated by $a_{1,i}a_{2,i}^\dagger$, it's clear that the action of $a_{1,i}a_{2,i}^\dagger$ is to move a photon from the $i$th mode of register $1$ to the $i$th mode of register $2$. The opposite is accomplished $a_{1,i}^\dagger a_{2,i}$. That is, all terms in $\Gamma$ are hopping terms between modes, and so preserve the total particle number and the subspace $V_k$. Since $\Gamma^{(1,2)}$ is Hermitian and $\Pi_{\ker \Gamma^{(1,2)}}$ is an eigenprojector of $\Gamma^{(1,2)}$, standard results from linear algebra tell us that $V_k$ is invariant under $\Pi_{\ker \Gamma^{(1,2)}}$. The invariance of $V_k$ under $\zeromeanbosonictestproj^{(1,2)}$ follows from its invariance under $\Pi_{\ker \Gamma^{(1,2)}}$ and $\symproj^{(1,2)}$.
\end{proof}

Now, for $R \in \orthgrp(3)$, define $U_R$ as follows:
\begin{equation}\label{eq:o3-representation-fock}
        U_R a_{s,i} U_R = \sum_{t=1}^3 R_{t,s} a_{t,i}
\end{equation}
We show that when $R$ consists of a rotation about the $z$-axis, $U_{R(\alpha)}$ takes a particularly nice form.

\begin{lemma}\label{lem:so-representation-conjugate-action}
    Let
    \[
    R(\alpha) = \begin{pmatrix}
        \cos \alpha & - \sin \alpha & 0 \\
        \sin \alpha & \cos \alpha & 0 \\
        0 & 0 & 1 
    \end{pmatrix}
    \]
    for $\alpha \in [0, 2\pi)$. Then
    \[
    U_{R(\alpha)} = e^{-\ri \alpha \Gamma}.
    \]
\end{lemma}
\begin{proof}
    For simplicity, we will denote $\Gamma = \Gamma^{(1,2)}$ throughout this proof.
    Since Gaussian unitaries are uniquely defined by their action on the raising and lowering operators by conjugation, it suffices to show the following $3$ statements:
    \begin{equation}\label{eq:first-conjugate-G-statement}
    e^{-\ri \alpha \Gamma}a_{1,i} e^{\ri \alpha \Gamma} = \cos \alpha \cdot a_{1,i} + \sin \alpha \cdot a_{2,i}
    \end{equation}
    \begin{equation}\label{eq:second-conjugate-G-statement}
    e^{-\ri \alpha \Gamma}a_{2,i} e^{\ri \alpha \Gamma} = -\sin \alpha \cdot a_{1,i} + \cos \alpha \cdot a_{2,i}
    \end{equation}
    \begin{equation}\label{eq:third-conjugate-G-statement}
    e^{-\ri \alpha \Gamma}a_{3,i} e^{\ri \alpha \Gamma} = a_{3,i}
\end{equation}

To show equations \eqref{eq:first-conjugate-G-statement} and \eqref{eq:second-conjugate-G-statement}, we will use the well-known Campbell identity:
\[
e^{A}Be^{-A} = \sum_{k=0}^\infty \frac{[A,B]_k}{k!},
\]
where $[A,B]_k$ is the $k$-th order nested commutator of $A$ and $B$.

Applying the bosonic canonical commutation relations, we have
\begin{align}
    [\Gamma, a_{1,i}] &= \ri [a_{1,i}a_{2,i}^\dagger - a_{1,i}^\dagger a_{2,i}, a_{1,i}] \\
    &= \ri\left(a_{1,i}a_{2,i}^\dagger a_{1,i} - a_{1,i}^\dagger a_{2,i} a_{1,i} -  a_{1,i}^2a_{2,i}^\dagger +  a_{1,i} a_{2,i}^\dagger a_{2,i}^\dagger \right)\\
    &= \ri[a_{1,i}, a_{1,i}^\dagger] \cdot a_{2,i} = \ri a_{2,i}.
\end{align}
Antisymmetrically, we have
\begin{align}
    [\Gamma, a_{2,i}] &= \ri[a_{1,i}a_{2,i}^\dagger - a_{1,i}^\dagger a_{2,i}, a_{2,i}] \\
    &= \ri\left(a_{1,i}a_{2,i}^\dagger a_{2,i} - a_{1,i}^\dagger a_{2,i}^2 -  a_{2,i}a_{1,i}a_{2,i}^\dagger +  a_{2,i} a_{1,i}^\dagger a_{2,i}^\dagger\right) \\
    &= \ri  a_{1,i} \cdot [a_{2,i}^\dagger, a_{2,i}]  = - \ri a_{1,i}.
\end{align}
We now observe a pattern in the nested commutators of $-\ri \alpha \Gamma$ and the $a_{1,i},a_{2,i}$:
\begin{equation}
    [-\ri \alpha \Gamma, a_{1,i}]_0 = a_{1,i}, \  [-\ri \alpha \Gamma, a_{1,i}]_1 = \alpha a_{2,i}, \  [-\ri \alpha \Gamma, a_{1,i}]_2 = -\alpha^2 a_{1,i}, \  [-\ri \alpha \Gamma, a_{1,i}]_3 = -\alpha^3 a_{2,i}
\end{equation}

\begin{equation}
    [-\ri \alpha \Gamma, a_{2,i}]_0 = a_{2,i}, \  [-\ri \alpha \Gamma, a_{2,i}]_1 = -\alpha \alpha a_{1,i}, \  [-\ri \alpha \Gamma, a_{2,i}]_2 = -\alpha^2 a_{2,i}, \  [-\ri \alpha \Gamma, a_{2,i}]_3 = \alpha^3 a_{1,i}
\end{equation}

Thus, by the Campbell identity,
\begin{align}
    e^{-\ri \alpha \Gamma}a_{1,i}e^{\ri \alpha \Gamma} &= \sum_{k=0}^{\infty} \frac{[-\ri \alpha \Gamma, a_{1,i}]_k}{k!} 
    \\ &= a_{1,i} \sum_{k=0}^\infty \frac{ (-1)^k\alpha^{2k}}{(2k)!} + a_{2,i} \sum_{k=0}^\infty \frac{(-1)^k\alpha^{2k+1}}{(2k+1)!} \\
    &= \cos\alpha \cdot a_{1,i} + \sin \alpha \cdot a_{2,i}
\end{align}
and
\begin{align}
    e^{-\ri \alpha \Gamma}a_{2,i}e^{\ri \alpha \Gamma} &= \sum_{k=0}^{\infty} \frac{[-\ri \alpha \Gamma, a_{2,i}]}{k!} 
    \\ &= \alpha_{2,i} \sum_{k=0}^\infty \frac{(-1)^k\alpha^{2k}}{(2k)!} - a_{1,i} \sum_{k=0}^\infty \frac{(-1)^k \alpha^{2k+1}}{(2k+1)!} \\
    &= -\sin\alpha \cdot a_{1,i} + \cos \alpha \cdot a_{2,i}.
\end{align}

Finally, equation \eqref{eq:third-conjugate-G-statement} simply follows from the fact that $\Gamma^{(1,2)}$ acts trivially on the third register.
\end{proof}

We show that this acceptance projector is the Haar average of a subrepresentation of $\orthgrp(3)$.

\begin{lemma} \label{lem:zero-mean-projector-invariant}
    Let $\mu_{\orthgrp(2)}$ be the Haar measure on $\orthgrp(2)$ and define
    \begin{equation}
    H = \left\{\begin{pmatrix}
        R^\prime & 0 \\ 0 & 1
    \end{pmatrix} : R^\prime \in \orthgrp(2)\right\}.
    \end{equation}
    Then
    \begin{equation}
    \zeromeanbosonictestproj^{(1,2)} = \int_{H} U_{R} \ \rd \mu_{\orthgrp(2)}(R).
    \end{equation}
\end{lemma}
\begin{proof}
    Write
    \[
    R(\alpha) = \begin{pmatrix}
        \cos \alpha & - \sin \alpha & 0 \\
        \sin \alpha & \cos \alpha & 0 \\
        0 & 0 & 1 
    \end{pmatrix}.
    \]
    This forms a subgroup $K \subseteq H$ isomorphic to $\splorthgrp(2)$.
    Now, by definition, we have
    \[
    U_R a_{1,i} U_R^\dagger = \cos \alpha \cdot a_{1,i} + \sin \alpha \cdot a_{2,i}, \quad U_R a_{2,i} U_R^\dagger = -\sin \alpha \cdot a_{1,i} + \cos \alpha \cdot a_{2,i}, \quad U_R a_{3,i} U_R^\dagger = a_{3,i}.
    \]
    By \cref{lem:so-representation-conjugate-action}, we have that $U_{R(\alpha)} = e^{-\ri \alpha \Gamma}$.
    Thus we compute
    \begin{align}
    \int_{K} U_{R} \ \rd \mu_{\splorthgrp(2)}(R) &= \frac{1}{2\pi}\int_0^{2\pi} U_{R(\alpha)} \ \rd \alpha \\
    &= \frac{1}{2\pi}\int_{0}^{2 \pi} e^{-\ri\alpha \Gamma} \ \rd \alpha \\
    &= \frac{1}{2\pi}\sum_{k=0}^\infty \int_{0}^{2 \pi} \frac{(-\ri)^k \alpha^k \Gamma^k}{k!} \ \rd \alpha \\
    &= \frac{1}{2\pi}\sum_{k=0}^\infty \frac{(-\ri)^k (2\pi)^{k+1} \Gamma^k}{(k+1)!} \\
    &= I + \sum_{k=1}^\infty \frac{(-\ri)^k (2\pi)^{k} \Gamma^k}{(k+1)!}
    \end{align}
    Now, $\Gamma$ has integer eigenvalues. On the $0$ eigenspace, the above operator acts as the identity. For an nonzero integer eigenvalue $\lambda$, the above operator acts as 
    \begin{equation}
        \sum_{k=0}^\infty \frac{(-2\pi \ri \lambda)^{k}}{(k+1)!} = -\frac{1}{2\pi \ri \lambda} \sum_{k=1}^\infty \frac{(-2\pi \ri \lambda)^{k}}{k!} = \frac{1 - e^{2\pi \ri \lambda}}{2\pi \ri \lambda} = 0.
    \end{equation}
    Thus 
    \begin{equation}
        \int_{K} U_{R} \ \rd \mu_{\splorthgrp(2)}(R) = \Pi_{\ker \Gamma^{(1,2)}}.
    \end{equation}

    We now compute the Haar integral over $\orthgrp(2)$. Consider the coset $F$ of $H / K$ generated by $SK$, where $S$ is the $3\times 3$ permutation matrix that swaps the first and second rows. We have that $H = K \sqcup F$, and that $\mu_{\orthgrp(2)}(F) =  \mu_{\orthgrp(2)}(K)$. Thus, by multiplicative invariance of the Haar measure,
    \begin{align}
        \int_{H} U_R \ \rd \mu_{\orthgrp(2)}(R) &= \frac{1}{2} \int_{K} U_R \ \rd \mu_{\splorthgrp(2)}(R) + \frac{1}{2} \int_{K} U_{SR} \ \rd \mu_{\splorthgrp(2)}(R). \\
        &= \frac{1}{2} \int_{K} U_R \ \rd \mu_{\splorthgrp(2)}(R) + U_{S}\frac{1}{2} \int_{K} U_{R} \ \rd \mu_{\splorthgrp(2)}(R). \\
        &= \frac{I + \SWAP^{(1,2)}}{2} \Pi_{\ker \Gamma^{(1,2)}} \\
        &= \symproj^{(1,2)} \Pi_{\ker \Gamma^{(1,2)}}. \qedhere
    \end{align}
\end{proof}

Now, the entire result follows from the fact that the spectral gap of the three-copy projector operator is identical to that of the fermionic Gaussian tester.
\begin{lemma}\label{lem:spec_gap_zero_mean_helper}
    Let $P = \zeromeanbosonictestproj^{(1,2)}\SWAP^{(2,3)}\zeromeanbosonictestproj^{(1,2)}$ and let $\lambda$ be any eigenvalue of $P$ that is not equal to $1$. Then $\lambda \in [-1/2,3/8]$.
\end{lemma}
\begin{proof}
    Since $P$ preserves $V_k$ by \cref{lem:particle-number-invariant-subspace}, it suffices to find the eigenvalues of $P$ on each block corresponding to the subspace $V_k$.
    
    Consider the total number operator $\hat{N}$:
    \begin{equation}
    \hat{N} = \sum_{i=1}^n a_{1,i}^\dagger a_{1,i} + a_{2,i}^\dagger a_{2,i} + a_{3,i}^\dagger a_{3,i},
    \end{equation}
    where $a_{s,i}$ indicates the lowering operator on mode $i$ of register $s$. 

    Choose an $R \in \orthgrp(3)$ and recall the unitary representation of $\orthgrp(3)$ defined in \eqref{eq:o3-representation-fock} as
    \begin{equation}
        U_R a_{s,i}^\dagger U_R^\dagger = \sum_{t=1}^3 R_{t,s} a_{t,i}^\dagger
    \end{equation}
    for each $s \in \{1,2,3\}$. We observe that
    \begin{align}
        U_R\hat{N}U_R^\dagger &= \sum_{i=1}^n \sum_{s=1}^3 U_R a_{s,i}^\dagger a_{s,i} U_R^\dagger \\
        &= \sum_{i=1}^n \sum_{s=1}^3 U_R a_{s,i}^\dagger U_R^\dagger U_R a_{s,i} U_R^\dagger \\ %
        &= \hat{N},
    \end{align}
    so the action of $\orthgrp(3)$ defined by $U_R$ preserves total particle number.
    
    Observe that because $V_k$ is invariant under the action $R \to U_R$, it is a representation of $\orthgrp(3)$. $V_k$ is also clearly finite dimensional. Indeed, we have
    \[
    \dim(V_k) = \binom{k + 3n - 1}{3n - 1}.
    \]
    As such, fixing $k$, we can decompose $V_k$ into the irreducible representations of $\orthgrp(3)$:
    \begin{equation}
        V_k \cong \left(\bigoplus_{\ell \geq 0} \calY_{\ell}^{\oplus m_{k,\ell}^+} \right) \oplus \left(\bigoplus_{\ell \geq 0} \left(\calY_{\ell} \otimes {\det}\right)^{\oplus m_{k,\ell}^-} \right) = W_k,
    \end{equation}
    where $m_{k,\ell}^+$ and $m_{k,\ell}^-$ indicate multiplicities.
    We will now compute the eigenvalues of the operator $P$ on this subspace. The action of $P$ is also block-diagonal on these irreps, and it suffices to calculate the eigenvalues of $P_{\calY_\ell}$ and $P_{\calY_\ell \otimes {\det}}$ --- the multiplicities $m_{k,\ell}^+$ and $m_{k,\ell}^-$ clearly do not affect the spectrum. Let $\wt{U}_R$ be the corresponding mapping of the representation $U_R$ on the $W_k$
    
    Consider the subgroup $H \cong \orthgrp(2)$ of $\orthgrp(3)$ that acts nontrivially only on the first two registers:
    \begin{equation}
    H = \left\{\begin{pmatrix}
        R^\prime & 0 \\ 0 & 1
    \end{pmatrix} : R^\prime \in \orthgrp(2)\right\}
    \end{equation}

    By \cref{lem:zero-mean-projector-invariant}, $\zeromeanbosonictestproj$ is the projector onto the subspace invariant under $U_{R_z(\alpha)}$ for $R_z(\alpha) \in H$. Fix $\ell$ and consider the subspace of $\calY_\ell$ that is invariant under some $R_z(\alpha) \in H$. By \cref{eq:spherical-z-rotation-action}, we have that
    \[
    \wt{U}_{R_z(\alpha)}Y_{\ell, m}(\theta,\phi) = e^{-\ri m \alpha}Y_{\ell, m}(\theta,\phi),
    \]
    and so the only subspace of $\calY_\ell$ invariant under all $\wt{U}_R$ for $R \in H$ is spanned by $Y_{\ell, 0}$. We denote this by $\calY_\ell^0$ for convenience. Now, we turn our attention to the representations $\calY_\ell \otimes {\det}$. If $R \in H$ consists of a reflection as well as a rotation, then $\calY_\ell \otimes {\det}$ negates it, so $\left.P\right|_{\calY_\ell \otimes {\det}} = 0$. Otherwise, $\left.\wt{U}_{R}\right|_{\calY_\ell \otimes {\det}} $ is essentially identical to $\left.\wt{U}_{R}\right|_{\calY_\ell}$. Thus it suffices to consider the eigenspectrum of $\left.P\right|_{\calY_\ell^0}$. The $\SWAP^{(2,3)}$ term of $P$ can be written as the action of a permutation matrix in $\orthgrp(3)$:
    \[
     R_{\SWAP_{y,z}} = \begin{pmatrix}
        1 & 0 & 0 \\ 0 & 0 & 1 \\ 0 & 1 & 0
    \end{pmatrix} = \begin{pmatrix}
        1 & 0 & 0 \\ 0 & 0 & -1 \\ 0 & 1 & 0
    \end{pmatrix}\begin{pmatrix}
        1 & 0 & 0 \\ 0 & -1 & 0 \\ 0 & 0 & 1
    \end{pmatrix} = R_x(-\pi/2) \cdot \diag(1,-1,1).
    \]
    The second matrix is within $H$, and so does not change $\calY_\ell^0$. It suffices now to compute
    \[
    \int \overline{Y_{\ell,0}}(\theta, \phi)  \wt{U}_{R_{\SWAP_{y,z}}} Y_{\ell,0}(\theta, \phi)\ \rd \theta \rd \phi = \int \overline{Y_{\ell,0}}(\theta, \phi)  \wt{U}_{R_{x}(-\pi/2)} Y_{\ell,0}(\theta, \phi)\ \rd \theta \rd \phi
    \]
    We will compute $\left.\wt{U}_{R_{x}(-\pi/2)}\right|_{Y_{\ell,0}}$.  It will be convenient to write $Y_{\ell, \theta}$ as a function of rectangular coordinates $x = \sin \theta \cos \phi, y = \sin \theta \sin \phi,$ and $z = \cos \theta$.
    Recall that, by \eqref{eq:spherical-harmonics}
    \[
    Y_{\ell,0}(\theta, \phi) = \sqrt{\frac{2 \ell + 1}{4\pi}} P_\ell(\cos \theta) = \sqrt{\frac{2 \ell + 1}{4\pi}} P_\ell(z)
    \]
    which is independent of $\phi$. The action of $R_{x}(-\pi/2)$ swaps the $y$ and $z$ coordinates:
    \[
    \wt{U}_{R_{x}(-\pi/2)} Y_{\ell,0}(x,y,z) = Y_{\ell,0}(x,z,y).
    \]
    Thus we have
    \begin{align}
    \int_{S^2} \overline{Y_{\ell,0}}(x,y,z)  Y_{\ell,0}(x,z,y)\ \rd x \rd y \rd z &= \frac{2\ell + 1}{4 \pi} \int_{S^2} P_{\ell}(z) P_{\ell}(y)\ \rd x \rd y \rd z \\
    &= \frac{2 \ell + 1}{4 \pi} \int_{S^2} P_{\ell}(e_3 \cdot (x,y,z)) P_{\ell}(e_2 \cdot (x,y,z)) \\
    &= P_\ell(e_3 \cdot e_2) = P_\ell(0) & \tag{By \cref{prop:legendre-inner-product-identity}} 
    \end{align}
    Thus, when $\ell = 2r$ the associated nonzero eigenvalue of $P$ is
    \[
    P_{2r}(0) = (-1)^r \cdot \frac{\binom{2r}{r}}{4^r}
    \]
    and when $\ell = 2r+1$ the associated eigenvalue of $P$ is $P_{2r+1}(0) = 0$. This exactly matches the spectrum of the corresponding $P$ in the fermionic case, and so every eigenvalue of $P$ is either $1$ or in $[-1/2,3/8]$. \qedhere
    
\end{proof}

\subsection{The algorithm}

We now state our algorithm that efficiently constructs the projector $\zeromeanbosonictestproj$. The idea is to consider a Gaussian unitary $V_i$ that acts on $i$th modes of both the first and second register that satisfies:
\[
V_i a_{1,i} V_i^\dagger = \frac{a_{1,i} +  a_{2,i}}{\sqrt{2}}, \quad V_i a_{2,i} V_i^\dagger = \frac{\ri(a_{1,i} - a_{2,i})}{\sqrt{2}}.
\]
Applying $V_i$ to each term $\Gamma_i$ of $\Gamma^{(1,2)}$ gives us
\begin{align}
V_i \Gamma_iV_i^\dagger &=\ri V_i\left(a_{1,i}a_{2,i}^\dagger - a_{1,i}^\dagger a_{2,i}\right) V_i^\dagger \\
&= a_{1,i}^\dagger a_{1,i} - a_{2,i}^\dagger a_{2,i} \\
&= \hat{N}_{1,i} - \hat{N}_{2,i},
\end{align}
where $\hat{N}_{t,i}$ is the number operator for mode $i$ of register $t$. Taking $V = \bigotimes_{i=1}^n V_i$, we have that
\begin{equation}
    V\Gamma V^\dagger = \hat{N}_{1} - \hat{N}_{2}
\end{equation}
Thus measuring the equal-photon subspace projects onto $\ker \Gamma^{(1,2)}$. Since we supply a state $\ket{\psi}^{\otimes 2}$, we are already within the symmetric subspace. Thus \cref{alg:zero-mean-bosonic-tester} efficiently implements the projector $\zeromeanbosonictestproj$ given an input state $\ket{\psi}^{\otimes 2}$.

\begin{algorithm}[H]
    \caption{Base tester for zero-mean bosonic Gaussian states.}\label{alg:zero-mean-bosonic-tester}
    \KwInput{Two copies of a bosonic pure state $\ket{\psi}$.}

    \BlankLine

    Apply the Gaussian unitary $V = \bigotimes_{i} V_i$ defined above.
    
    Measure the occupation numbers $N_1$ and $N_2$ 
    of the first and second register, respectively.

    \uIf{$N_1 = N_2$}{
        \KwRet{\accept}
    }
    \Else{
        \KwRet{\reject}
    }
\end{algorithm}

We remark that the Gaussian rotation $V$ can be seen as a bosonic analogue of Bell sampling. Since $V$ is also local and a passive Gaussian, it can also be implemented in an experimentally friendly manner.

Now we show that repeating \cref{alg:zero-mean-bosonic-tester} results as a subroutine of \cref{alg:constant-copy-tester} yields an effective tolerant tester for zero-mean bosonic Gaussian states.

\begin{proof}[Proof of \cref{thm:testing-zero-mean-bosonic-gaussians}]
Once again, we invoke \cref{alg:constant-copy-tester} with \cref{alg:zero-mean-bosonic-tester} serving as the subroutine $\calA$. Perfect completeness of the test follows from \cref{lem:perfect-completeness-zero-mean}. Invoking \cref{thm:rejection-infidelity-bounds} with $m = 3$, $k = 2$, and a spectral gap of $\Delta = 5/12$ (\cref{lem:spec_gap_zero_mean}), we have that the rejection probability for any state $\ket{\psi}$ obeys
    \begin{equation}
        \frac{1}{8} \, \dist(\psi, \zeromeanbosonicgaussianset)^2 \leq q(\psi) \leq 2 \, \dist(\psi, \zeromeanbosonicgaussianset)^2
    \end{equation}
    Hence \cref{thm:framework-tolerant-testing} yields the claim with $\beta = 4$.
\end{proof}

\section{Testing bosonic coherent states}\label{sec:coherent-state-testing}

The last of our property testing algorithms is for bosonic coherent states. %
These states can be written as tensor products of single-mode coherent states:
\[
\ket{\alpha} = e^{-|\alpha|^2/2}\sum_{k=0}^\infty \frac{\alpha^k}{\sqrt{k!}} \ket{k},
\]
where $\alpha \in \mathbb{C}$.

We present the following property testing result for $\coherentstateset_n$:
\begin{theorem}\label{thm:testing-coherent-states}
    Let $\beta = 2\sqrt{2}$. For all $0 \leq \eps_1 < \eps_2 < 1$ with $\eps_2 > \beta \eps_1$ and $0 < \delta < \frac{1}{2}$, there is an $O\paren*{\frac{\log(1/\delta)}{(\eps_2 - \beta\eps_1)^2}}$-copy $(\eps_1, \eps_2, \delta)$-tester for coherent states. The tester acts on two copies at a time.
\end{theorem}

The key idea behind our property tester is the action of the evenly-mixing beam splitter $\beamsplitter = \exp\left(\frac{\pi}{4}\sum_{i=1}^n a_{i}^\dagger \otimes a_{i} - a_{i} \otimes a_{i}^\dagger\right)$. For $\ket{\alpha}, \ket{\beta} \in \coherentstateset_n$, we have the well-known identity
\begin{equation}\label{eq:beamsplitter-on-coherent-states}
\beamsplitter \ket{\alpha}\!\ket{\beta} = \ket{\frac{\alpha + \beta}{\sqrt{2}}}\!\ket{\frac{\beta - \alpha}{\sqrt{2}}}.
\end{equation}
For a proof of this fact, we refer the reader to \cite{serafini2017continuous}. Our property testing algorithm thus simply applies $\beamsplitter$ to $\ket{\alpha}\!\ket{\alpha}$ and checks that the second register is in the vacuum state. Perfect completeness of this protocol is essentially trivial.

Let $\coherenttestproj$ be the projector implemented by this algorithm. As it turns out, this projector is proportional to a $2$-copy Haar average over coherent states.

\begin{proposition}\label{prop:acceptance-projector}
    Let $\coherenttestproj$ be the projector onto the accepting subspace of the test. We can write
    \[
    \left(\frac{2}{\pi}\right)^n\int_{\mathbb{C}^n}\ketbra{\alpha}{\alpha}^{\otimes2} \ \rd^{2n}\alpha.
    \]
\end{proposition}

\begin{proof}
We will use \cref{fact:coherent-identity-resolution}.
\begin{align*}
\coherenttestproj &= \beamsplitter^\dagger(I \otimes \ketbra{0^n}{0^n})\beamsplitter \\
&= \pi^{-n}\int_{\mathbb{C}^n}\beamsplitter^\dagger(\ketbra{\beta}{\beta} \otimes \ketbra{0^n}{0^n})\beamsplitter \ \rd^{2n}\beta \\
&= \pi^{-n}\int_{\mathbb{C}^n}\ketbra{\beta/\sqrt{2}}{\beta/\sqrt{2}} \otimes \ketbra{\beta/\sqrt{2}}{\beta/\sqrt{2}} \ \rd^{2n}\beta \\
&= \left(\frac{2}{\pi}\right)^n\int_{\mathbb{C}^n}\ketbra{\alpha}{\alpha}^{\otimes2} \ \rd^{2n}\alpha. \qedhere
\end{align*}
\end{proof}

While this exact fact will not be critical to our analysis, we remark that the form of this projector matches the general approach outlined in \cref{sec:framework}.
Thus the spectral gap technique once again seems natural, and we will use our framework to show soundness.

\subsection{Spectral gap analysis}
\begin{lemma}\label{lem:coherent-states-spectral-gap}

Let $\lambda_1$ be the largest eigenvalue of $\Xi$ and let $\lambda$ be any other (distinct) eigenvalue of $\Xi$. Then $\lambda_1 = 1$ and $\lambda \in [0, 1/2]$.
\end{lemma}

\begin{proof}
    As in the previous sections, we analyze the equivalent eigengap of
    \[
    \Theta = \coherenttestproj^{(1,2)} \symproj^{(1,2,3)}\coherenttestproj^{(1,2)}  = \frac{\coherenttestproj^{(1,2)} + 2 Q }{3}.
    \]
    where $Q = \coherenttestproj^{(1,2)} \SWAP^{(2,3)}\coherenttestproj^{(1,2)}$. In particular, it will suffice to show that the second largest eigenvalue of $\Theta$ is at most $\frac{1}{2}$.
    
    Borrowing notation from previous sections, we write $a_{t,i}$ as the lowering operator acting on mode $i$ of the $t$-th register.
    For each mode $i \in [n]$, define
    \begin{equation}
        b_i = \frac{a_{1,i}^\dagger + a_{2,i}^\dagger}{\sqrt{2}},\quad d_i = \frac{a_{1,i}^\dagger - a_{2,i}^\dagger}{\sqrt{2}}.    
    \end{equation}

    Observe that %
    \[
    \beamsplitter^\dagger a_{1,i}^\dagger \beamsplitter = b_i
    \]
    and
    \[
    \beamsplitter^\dagger a_{2,i}^\dagger \beamsplitter = -d_i.
    \]

    The only nonzero eigenvalues of $\Theta$ correspond to eigenvectors that exist in the subspace $\calV$ which $\coherenttestproj^{(1,2)}$ projects to: that is, states whose second register is the vacuum state after applying $\beamsplitter$. As such, this subspace is the set of all states obtained by applying some product of $b_i$ and $a_{3,i}^\dagger$ to the vacuum state $\ket{0^n,0^n,0^n}$. We consider the action of $\SWAP^{(2,3)}$ on these generators:
    \begin{equation}
        \SWAP^{(2,3)} b_i \SWAP^{(2,3)} = \frac{a_{1,i}^\dagger + a_{3,i}^\dagger}{\sqrt{2}} = \frac{1}{2}b_{i} + \frac{1}{2} d_i + \frac{1}{\sqrt{2}}a_{3,i}^\dagger
    \end{equation}
    \begin{equation}
        \SWAP^{(2,3)} a_{3,i}^\dagger \SWAP^{(2,3)} = a_{2,i}^\dagger = \frac{1}{\sqrt{2}}b_{i} - \frac{1}{\sqrt{2}}d_i.
    \end{equation}
    Since $d_i$ is projected out by $\coherenttestproj^{(1,2)}$, we have 
    \begin{equation}\label{eq:swap-bi-action-projected}
        \SWAP^{(2,3)} b_i \SWAP^{(2,3)} = \frac{1}{2}b_{i} + \frac{1}{\sqrt{2}}a_{3,i}^\dagger
    \end{equation}
    \begin{equation}\label{eq:swap-a3i-action-projected}
        \SWAP^{(2,3)} a_{3,i}^\dagger \SWAP^{(2,3)} = \frac{1}{\sqrt{2}}b_i
    \end{equation}
    in $\calV$.

    Now, we have
    \begin{equation}
        \calV = \linspan\left\{ \prod_{i=1}^n b_{i}^{x_i}{a_{3,i}^\dagger}^{y_i}\ket{0^n,0^n,0^n}\right\}.
    \end{equation}
    Writing
    \begin{equation}
        u_i = \sqrt{\frac{2}{3}}b_i + \frac{1}{\sqrt{3}} a_{3,i}^\dagger, \quad v_i = \frac{1}{\sqrt{3}}b_i - \sqrt{\frac{2}{3}} a_{3,i}^\dagger.
    \end{equation}
    This is a unitary change of basis from the generators $b_{i},a_{3,i}^\dagger$, and so, defining
    \begin{equation}
        \ket{s,t} = \prod_{i=1}^n u_i^{s_i} v_i^{t_i} \ket{0^n,0^n,0^n},
    \end{equation}
    we have $\calV = \linspan\{\ket{s,t}: s,t \in \N^n\}$.
    A simple calculation reveals that 
    \begin{equation}
        Q u_i  = \frac{1}{\sqrt{6}}b_i + \frac{1}{\sqrt{3}}a_{3,i}^\dagger + \frac{1}{\sqrt{6}}b_i = \sqrt{\frac{2}{3}}b_i + \frac{1}{\sqrt{3}} a_{3,i}^\dagger = u_i
    \end{equation}
    and
    \begin{equation}
        Q v_i =\frac{1}{2\sqrt{3}}b_i + \frac{1}{\sqrt{6}}a_{3,i}^\dagger - \frac{1}{\sqrt{3}} b_i = - \frac{1}{2\sqrt{3}} b_i + \frac{1}{\sqrt{6}} a_{3,i}^\dagger = -\frac{1}{2} v_i.
    \end{equation}
    Thus the $\ket{s,t}$ fully diagonalize $Q$:
    \begin{equation}
        Q\ket{s,t} = \left(-\frac{1}{2}\right)^{\sum_{i=1}^n t_i} \ket{s,t}.
    \end{equation}
    As such, the eigenvalues of $\Theta$ on $\calV$ (and therefore the eigenvalues of $\Xi$ on $\calV$) can be written
    \begin{equation}
        \nu_k = \frac{1 + 2\left(-\frac{1}{2}\right)^k}{3}
    \end{equation}
    for $k \in \N$. In this way, we obtain that the second largest distinct eigenvalue of $\Theta$ is for $k=2$, yielding $\nu_2 = \frac{1}{2}$ and proving the statement.
\end{proof}

\subsection{The algorithm}

Now we are ready to state and analyze our efficient algorithm to test coherent states.

\begin{algorithm}[H]
    \caption{Base tester for bosonic coherent states.}\label{alg:coherent-state-tester}
    \KwInput{Two copies of a bosonic pure state $\ket{\psi}$.}

    \BlankLine

    Apply the Gaussian unitary $\beamsplitter = \exp\left(\frac{\pi}{4}\sum_{i=1}^n a_{i}^\dagger \otimes a_{i} - a_{i} \otimes a_{i}^\dagger\right)$ to $\ket{\psi^{\otimes 2}}$.
    
    Measure the second register in the Fock basis and let $N_2$ be the number of bosons detected.

    \uIf{$N_2 = 0$}{
        \KwRet{\accept}
    }
    \Else{
        \KwRet{\reject}
    }
\end{algorithm}

Using the spectral gap derived in \cref{lem:coherent-states-spectral-gap}, we are able to complete the soundness portion of the argument and show that \cref{alg:coherent-state-tester} yields a tolerant property tester for $\coherentstateset_n$.

\begin{proof}[Proof of \cref{thm:testing-coherent-states}]
We run \cref{alg:constant-copy-tester} using  \cref{alg:coherent-state-tester} serving as the subroutine $\calA$. 
Perfect completeness is trivial from equation \eqref{eq:beamsplitter-on-coherent-states}.
By \cref{thm:rejection-infidelity-bounds} with the spectral gap $\Delta = 1/2$ (\cref{lem:coherent-states-spectral-gap}), $k=2$, and $m=3$, we have that the rejection probability for any state $\ket{\psi}$ obeys
    \begin{equation}
        \frac{1}{4} \, \dist(\psi, \coherentstateset_n)^2 \leq q(\psi) \leq 2 \, \dist(\psi, \coherentstateset_n)^2
    \end{equation}
    Hence \cref{thm:framework-tolerant-testing} yields the claim with $\beta = 2\sqrt{2}$.
\end{proof}

\section{Lower bounds}\label{sec:lower_bounds}

In this section we prove \cref{thm:general-testing-lower-bound}, which establishes lower bounds for testing each of our state classes.

\subsection{Reduction to binary state discrimination}

The argument is based on Helstrom's bound applied to two states \cite{helstrom1969quantum}. A lower bound for this problem implies a lower bound for testing, because any tester must be able to at least distinguish between one state close to the class $\calC$ and another state far from $\calC$.

\begin{proposition}[Binary state discrimination lower bound]\label{prop:helstrom}
    Let $\ket{\psi_A}, \ket{\psi_B}$ be pure states such that $F \coloneqq |{\braket{\psi_A | \psi_B}}|^2 \neq 0$. Every quantum algorithm which can successfully distinguish between the two states with probability at least $1 - \delta$ must consume at least
    \begin{equation}
        N \geq \frac{\log\mathopen{}\left( \frac{1}{4\delta(1-\delta)} \right)}{\log(1/F)}
    \end{equation}
    copies of the unknown state.
\end{proposition}

\begin{proof}
    Denote $\psi_j \coloneqq \ketbra{\psi_j}{\psi_j}$. Let $\{M, I-M\}$ be any POVM on $N$ copies such that the algorithm outputs the label for $\psi_A$ upon seeing $M$. Correctness of the algorithm implies
    \begin{equation}
        \tr(M \psi_A^{\otimes N}) \geq 1 - \delta, \quad  \tr(M \psi_B^{\otimes N}) \leq \delta.
    \end{equation}
    Hence $\tr(M (\psi_A^{\otimes N} - \psi_B^{\otimes N})) \geq 1 - 2\delta$. Meanwhile the variational characterization of the trace distance with $0 \preceq M \preceq I$ gives
    \begin{equation}
        \tr(M (\psi_A^{\otimes N} - \psi_B^{\otimes N})) \leq \dist(\psi_A^{\otimes N}, \psi_B^{\otimes N}) = \sqrt{1 - F^N}.
    \end{equation}
    Solving for $N$ yields the claim.
\end{proof}

Next we establish a sufficient condition for constructing hard-to-distinguish states. This is one version of such a hard instance; another version useful for the bosonic classes is presented in \cref{lem:interp_dist_AB}.

\begin{lemma}[Hard instance interpolation I]\label{lem:interpolation_state}
    Let $\calC$ be a class of pure states. Suppose there exist orthonormal states $\ket{A}, \ket{B}$ where $\ket{A} \in \calC$ and it holds that
    \begin{equation}
        |{\braket{A | \phi}}| + |{\braket{B | \phi}}| \leq 1 \quad \text{for every } \ket{\phi} \in \calC.
    \end{equation}
    For any $0 \leq \theta \leq \frac{\pi}{4}$, define
    \begin{equation}\label{eq:interpolation_state}
        \ket{\Psi_\theta} \coloneqq \cos\theta \ket{A} + \sin\theta \ket{B}.
    \end{equation}
    Then
    \begin{equation}
        \dist(\Psi_\theta, \calC) = \sin\theta.
    \end{equation}
\end{lemma}

\begin{proof}
    By restricting $0 \leq \theta \leq \frac{\pi}{4}$, we have $0 \leq \sin\theta \leq \cos\theta$. Thus by hypothesis, for any $\ket{\phi} \in \calC$ we have
    \begin{equation}
        |{\braket{\phi | \Psi_\theta}}| \leq \cos\theta |{\braket{\phi | A}}| + \sin\theta |{\braket{\phi | B}}| \leq \cos\theta.
    \end{equation}
    This is saturated at $\ket{\phi} = \ket{A}$, implying $\max_{\ket{\phi} \in \calC} |{\braket{\phi | \Psi_\theta}}| = \cos(\theta)$. Hence $\dist(\Psi_\theta, \calC) = \sqrt{1 - \cos^2\theta}$.
\end{proof}

\begin{theorem}[Lower bound master theorem]\label{thm:generic_testing_lower_bound}
    Let $\calC$ be a class of pure states, and let $0 \leq \eps_A < \eps_B \leq \frac{1}{\sqrt{2}}$. Suppose that the states $\ket{A}, \ket{B}$ of \cref{lem:interpolation_state} exist for the class $\calC$. Then there exist states $\ket{\psi_A}, \ket{\psi_B}$ which obey
    \begin{equation}\label{eq:state_dist_boundary}
        \dist(\psi_A, \calC) \leq \eps_A, \quad \dist(\psi_B, \calC) \geq \eps_B,
    \end{equation}
    such that any quantum algorithm which can distinguish these two states with probability at least $2/3$ must consume at least
    \begin{equation}
        N \geq \frac{\log(9/8)}{4(\eps_B - \eps_A)^2}
    \end{equation}
    copies of the unknown state.
\end{theorem}

\begin{proof}
    Set $\theta_j = \arcsin(\eps_j)$ and put $\ket{\psi_j} = \ket{\Psi_{\theta_j}}$ according to \cref{eq:interpolation_state}. \cref{eq:state_dist_boundary} follows from \cref{lem:interpolation_state}. Observe that
    \begin{equation}
    \begin{split}
        F \equiv |{\braket{\psi_A | \psi_B}}|^2 &= |{\cos(\theta_A)\cos(\theta_B) + \sin(\theta_A)\sin(\theta_B)}|^2\\
        &= \cos^2(\theta_B - \theta_A).
    \end{split}
    \end{equation}
    Denote $x = \sin^2(\theta_B - \theta_A)$ such that $F = 1-x$. Since $0 \leq \theta_B - \theta_A \leq \frac{\pi}{4}$ and $\sin^2$ is increasing on this interval, we have $x \leq \frac{1}{2}$. The elementary inequality $\log\frac{1}{1-x} \leq 2x$ for $x \in [0, \frac{1}{2}]$ implies $\log\frac{1}{F} \leq 2\sin^2(\theta_B - \theta_A)$. We further bound this using
    \begin{equation}
        \sin(\theta_B - \theta_A) \leq \theta_B - \theta_A = \arcsin(\eps_B) - \arcsin(\eps_A) \leq \frac{\eps_B - \eps_A}{\sqrt{1 - \eps_B^2}} \leq \sqrt{2} (\eps_B - \eps_A),
    \end{equation}
    where the penultimate inequality follows from the mean-value theorem applied to $\arcsin'(t) = \frac{1}{\sqrt{1-t^2}}$. Apply this bound, $\log\frac{1}{F} \leq 4(\eps_B - \eps_A)^2$, to \cref{prop:helstrom} with $\delta = 1/3$ to obtain the copy complexity lower bound.
\end{proof}

To apply this lower bound to a particular class of states, we merely have to show that there exist states $\ket{A}, \ket{B}$ satisfying the premise of \cref{lem:interpolation_state}. In the sequel, we exhibit this property for fermionic Gaussian states that are (i) pure and (ii) also number eigenstates (Slater determinants).

\subsection{Hard instances for fermionic Gaussian states}

\paragraph{Pure Gaussian states.} Take $\fermionicgaussianset$, the set of all pure fermionic Gaussian states on $n \geq 4$ modes (as otherwise, all pure even states on at most $3$ modes are Gaussian \cite{bravyi2005classical} so the tester trivially passes using zero copies).

\begin{lemma}\label{lem:hard_instance_gaussian}
    Let $\ket{A} = \ket{0^n}$ and $\ket{B} = \ket{1^4 0^{n-4}}$. For every $n \geq 4$ and every $\ket{\phi} \in \fermionicgaussianset$,
    \begin{equation}
        |{\braket{A | \phi}}| + |{\braket{B | \phi}}| \leq 1.
    \end{equation}
\end{lemma}

\begin{proof}
    We first prove the $n = 4$ statement, then extend to larger $n$ afterwards. If $\ket{\phi}$ has odd parity then both amplitudes clearly vanish. Henceforth we assume it has even parity. By the Bloch--Messiah normal form, any such state can be written as
    \begin{equation}
        \ket{\phi} = U (u_1 + v_1 a_1^\dagger a_2^\dagger) (u_2 + v_2 a_3^\dagger a_4^\dagger) \ket{0^4}
    \end{equation}
    where $U$ is a particle-preserving Gaussian unitary and $u_j, v_j \in \C$ such that $|u_j|^2 + |v_j|^2 = 1$ \cite{bloch1962canonical,kraus2009pairing}. Because the zero- and four-particle sectors are one-dimensional, $U$ acts as a trivial phase on them so
    \begin{equation}
        |{\braket{0000 | \phi}}| = |u_1 u_2| \quad \text{and} \quad |{\braket{1111 | \phi}}| = |v_1 v_2|.
    \end{equation}
    Apply Cauchy-Schwarz to the unit vectors $(|u_1|, |v_1|)$ and $(|u_2|, |v_2|)$:
    \begin{equation}\label{eq:4-mode_bound}
        |u_1 u_2| + |v_1 v_2| \leq 1,
    \end{equation}
    which proves the claim for $n = 4$.

    Now we prove the statement for $n > 4$. Denote
    \begin{equation}
        x \coloneqq \braket{A | \phi} \quad \text{and} \quad y \coloneqq \braket{B | \phi}.
    \end{equation}
    We are always free to choose a phase $\alpha \in \R$ such that $|x + e^{-i\alpha} y| = |x| + |y|$. If we define the four-mode state $\ket{\chi} \coloneqq \frac{\ket{0000} + e^{i\alpha}\ket{1111}}{\sqrt{2}}$ and the extended $n$-mode state $\ket{\widetilde{\chi}} \coloneqq \ket{\chi} \ket{0^{n-4}}$, then
    \begin{equation}
        |{\braket{\widetilde{\chi} | \phi}}|^2 = \frac{|{\braket{0^n | \phi} + e^{-i\alpha} \braket{1^4 0^{n-4} | \phi}}|^2}{2} = \frac{(|x| + |y|)^2}{2}.
    \end{equation}
    Define $\sigma_4 \coloneqq \tr_{5,\ldots,n}(\ketbra{\phi}{\phi})$ where the partial trace is over the last $n-4$ modes. Observe that
    \begin{equation}
        \braket{\chi | \sigma_4 | \chi} = \sum_{b \in \{0, 1\}^{n-4}} |z_b|^2
    \end{equation}
    where $z_b = \braket{\chi, b | \phi}$. In particular, $z_{0^{n-4}} = \braket{\widetilde{\chi} | \phi}$ implying that
    \begin{equation}
        |{\braket{\widetilde{\chi} | \phi}}|^2 = \frac{(|x| + |y|)^2}{2} \leq \braket{\chi | \sigma_4 | \chi}.
    \end{equation}
    We now aim to bound this four-mode overlap.

    Recall two standard facts about Gaussian states (for instance, see \cite{bittel2025optimal} for proofs). First, the mode-reduced state of a Gaussian state is mixed Gaussian;~and second, every mixed Gaussian state can be written as a convex combination of pure Gaussian states (although parity is no longer conserved). These imply that $\sigma_4$ is a mixed Gaussian state which can be expressed as $\sigma_4 = \sum_{j} p_j \ketbra{\eta_j}{\eta_j}$ for some four-mode Gaussian states $\ket{\eta_j}$. But we have already proven in \cref{eq:4-mode_bound} that on four modes, $|{\braket{0000 | \eta_j}}| + |{\braket{1111 | \eta_j}}| \leq 1$, which means that
    \begin{equation}
        |{\braket{\chi | \eta_j}}| = \frac{|{\braket{0000 | \eta_j} + e^{-i\alpha} \braket{1111 | \eta_j}}|}{\sqrt{2}} \leq \frac{1}{\sqrt{2}}.
    \end{equation}
    Thus ${\braket{\chi | \sigma_4 | \chi}} = \sum_j p_j |{\braket{\chi | \eta_j}}|^2 \leq \frac{1}{2}$, implying that $|x| + |y| \leq 1$, which is what we set out to show.
\end{proof}

As a consequence, we have the claimed lower bound for pure Gaussian state testing.

\begin{corollary}
    Put any $n \geq 4$ and let $\fermionicgaussianset$ be the class of all pure $n$-mode fermionic Gaussian states. The lower bound of \cref{thm:generic_testing_lower_bound} applies to $\fermionicgaussianset$.
\end{corollary}

\begin{proof}
    Combine \cref{thm:generic_testing_lower_bound} with \cref{lem:hard_instance_gaussian}.
\end{proof}

\paragraph{Slater determinants.} Note that the vacuum state and all one-fermion states are Slater determinants. Moreover by particle--hole duality, the unique $n$-fermion state can be identified with the vacuum, and all $(n-1)$-fermion states can be identified with one-fermion states;~hence these are all Slater determinants as well. Thus Slater testing is only nontrivial when the particle number is between $2 \leq \eta \leq n - 2$. We assume $\slaterdeterminantset$ is the set of all Slater determinants within this particle-number interval.

\begin{lemma}\label{lem:hard_instance_slater}
    Let $\eta \geq 2$ be an integer and $A, B \subset [n]$ be disjoint subsets of size $\eta$. Every $\ket{\phi} \in \slaterdeterminantset$ satisfies
    \[
    |{\braket{A | \phi}}| + |{\braket{B | \phi}}| \leq 1.
    \]
\end{lemma}

\begin{proof}
    If $\ket{\phi}$ has particle number different from $\eta$, the statement trivially holds. Then, let $u \in \C^{n \times \eta}$ be the column-orthonormal matrix describing the $\eta$-particle $\ket{\phi}$. The amplitudes are $\braket{A | \phi} = \det u[A]$ where $u[A] \in \C^{\eta \times \eta}$ selects the $A$-rows of $u$, and similarly for $\braket{B | \phi}$. Define $P = u[A]^\dagger u[A]$ and $Q = u[B]^\dagger u[B]$. Since $A$ and $B$ are disjoint, $P$ and $Q$ are PSD, and $u^\dagger u = I$, we have
    \[
    P + Q \preceq I \implies \tr(P + Q) \leq \eta.
    \]
    Observe that
    \[
    |{\det u[A]}| = \sqrt{\det P} \leq \left( \frac{\tr \, P}{\eta} \right)^{\eta/2} \leq \frac{\tr \, P}{\eta},
    \]
    where the first bound is the AM-GM inequality on the eigenvalues of $P$, and the second holds because $\frac{\tr \, P}{\eta} \in [0, 1]$ and $\eta/2 \geq 1$. The same estimate holds for $\det u[B]$ via $Q$. Hence
    \[
    |{\det u[A]}| + |{\det u[B]}| \leq \frac{\tr(P+Q)}{\eta} \leq 1. \qedhere
    \]
\end{proof}

\begin{corollary}
    Put any $n \geq 4$ and let $\slaterdeterminantset$ be the class of all $n$-mode Slater determinants with particle numbers within $\{2, 3, \ldots, n-2\}$. The lower bound of \cref{thm:generic_testing_lower_bound} applies to $\slaterdeterminantset$.
\end{corollary}

\begin{proof}
    Combining \cref{thm:generic_testing_lower_bound} with \cref{lem:hard_instance_slater} implies the statement for $2 \leq \eta \leq \lfloor n/2 \rfloor$, as there is enough room to always choose disjoint $\eta$-subsets of $[n]$. For $\lfloor n/2 \rfloor < \eta \leq n - 2$, consider the particle--hole transformation $U^\dagger a_j U = a_j^\dagger$. By unitary invariance of the trace distance, this preserves the promise of \cref{eq:state_dist_boundary}. Thus we can map every such instance to the $2 \leq \eta \leq \lfloor n/2 \rfloor$ case with states $U\ket{\psi_A}$ and $U\ket{\psi_B}$, for which the preceding argument applies.
\end{proof}

\subsection{Relaxed reduction criterion}

Now we turn to proving lower bounds for bosonic families of states. In this case, we need a helper lemma which relaxes the condition of \cref{lem:interpolation_state} for finding hard-to-distinguish instances. This is because many classes of bosonic states lack exact orthogonality between their members;~for example, any two coherent states have overlap $|{\braket{\alpha | \beta}}|^2 = e^{-|\alpha - \beta|^2} > 0$.

\begin{lemma}[Hard instance interpolation II]\label{lem:interp_dist_AB}
    Let $\calC$ be any class of pure states. Suppose that a pure state $\ket{X}$ has a closest state $\ket{A} \in \calC$ in fidelity,
    \[
    q \coloneqq |{\braket{A | X}}|^2 = \sup_{\ket{\phi} \in \calC} |{\braket{\phi | X}}|^2,
    \]
    such that $q < 1$. Choose the global phase of $\ket{A}$ so that $\braket{A | X} = \sqrt{q}$ and define
    \[
    \ket{B} \coloneqq \frac{\ket{X} - \sqrt{q}\ket{A}}{\sqrt{1-q}}.
    \]
    Then $\ket{A}, \ket{B}$ are orthonormal, and the state $\ket{\Psi_\theta} = \cos\theta \ket{A} + \sin\theta \ket{B}$ obeys
    \begin{equation}\label{eq:interp_dist_sin}
    \dist(\Psi_\theta, \calC) = \sin\theta \quad \text{for } 0 \leq \theta \leq T,
    \end{equation}
    where $T = \arccos\sqrt{q}$. When $q \leq 1/2$, we can instead fix $T = \pi/4$.
\end{lemma}

\begin{proof}
    Recall the Fubini--Study distance between pure states:
    \[
    \fubiniStudyDistance(\psi_1, \psi_2) = \arccos{\abs{\braket{\psi_1|\psi_2}}},
    \]
    which lies within $[0, \pi/2]$. By hypothesis, $\inf_{\ket{\phi} \in \calC} \fubiniStudyDistance(\phi, X) = T$, and by construction,
    \[
    \ket{X} = \cos T \ket{A} + \sin T \ket{B}.
    \]
    Hence for $\theta \leq T$,
    \[
    \begin{split}
        \fubiniStudyDistance(\Psi_\theta, X) &= \arccos{\abs{\cos(T)\cos(\theta) + \sin(T)\sin(\theta)}}\\
        &= \arccos(\cos(T-\theta))\\
        &= T - \theta.
    \end{split}
    \]
    The infimum implies that for any $\ket{\phi} \in \calC$, we can apply triangle inequality to the Fubini--Study metric to get
    \[
    T \leq \fubiniStudyDistance(\phi, X) \leq \fubiniStudyDistance(\phi, \Psi_\theta) + \fubiniStudyDistance(\Psi_\theta, X) = T - \theta + \fubiniStudyDistance(\phi, \Psi_\theta).
    \]
    Rearranging implies $\fubiniStudyDistance(\phi, \Psi_\theta) \geq \theta$. In particular, the infimum attains this at $\ket{\phi} = \ket{A} \in \calC$. \cref{eq:interp_dist_sin} follows by the standard conversion $\dist(\cdot, \cdot) = \sqrt{1 - \cos^2(\fubiniStudyDistance(\cdot, \cdot))}$.
\end{proof}

When $q \leq 1/2$, we recover the same premise as in \cref{lem:interpolation_state} and so the lower bound of \cref{thm:generic_testing_lower_bound} applies here as well. Note that this does not require the more restrictive condition that $\max_{\ket{\phi} \in \calC} \left( \abs{\braket{\phi | A}} + \abs{\braket{\phi | B}} \right) \leq 1$.

\subsection{Hard instances for bosonic Gaussian states}

\paragraph{Pure Gaussian states.} We represent the elements of $\bosonicgaussianset$ in their Bargmann form \cite{bargmann1961hilbert,cariolaro2015fock}:
\begin{equation}\label{eq:bargmann}
    \ket{\phi} = C \exp\mathopen{}\left( \frac{1}{2} \sum_{j,k=1}^n Z_{jk} a_j^\dagger a_k^\dagger + \sum_{j=1}^n \mu_j a_j^\dagger \right) \ket{0^n}
\end{equation}
where $Z = Z^T \in \C^{n \times n}$, $\mu \in \C^n$, and $C$ is a normalization constant. In fact, we will only need the single-mode case, from which we can tensorize to get the $n$-mode hard instance.

\begin{lemma}\label{lem:one-photon_overlap_gauss}
    Let $\ket{X} = \ket{1, 0^{n-1}}$ be the Fock state with a single boson in the first mode. It holds that
    \[
    \sup_{\ket{\phi} \in \bosonicgaussianset} \abs{\braket{X | \phi}}^2 = \frac{3\sqrt{3}}{4e} < \frac{1}{2},
    \]
    which is maximized by (for instance)
    \begin{equation}\label{eq:A_state_bgs}
    \ket{A} = \sqrt{\frac{\sqrt{3}}{2e}} \exp\mathopen{}\left(- \frac{1}{4} a_1^\dagger a_1^\dagger + \sqrt{\frac{3}{2}} a_1^\dagger  \right) \ket{0^n}.
    \end{equation}
\end{lemma}

\begin{proof}
    First, consider a single mode. By \cref{eq:bargmann} we write
    \[
    \ket{\phi_1} = C \exp\mathopen{}\left( \frac{\lambda}{2} a^\dagger a^\dagger + \mu a^\dagger \right) \ket{0}
    \]
    where we assume the convention $|\lambda| < 1$. The normalization factor $C$ can be determined by a complex Gaussian integral. This is given by \cite[Eq.~(40)]{cariolaro2015fock}, which is parametrized as $\lambda = e^{i\theta} \tanh r$, $\mu = \alpha \sech r$:
    \begin{equation}\label{eq:bargmann_norm}
    \begin{split}
        |C|^2 &= (\sech r) \exp\mathopen{}\left( -|\alpha|^2 - \Re[\alpha^2 e^{-i\theta}] \tanh r \right)\\
        &= \sqrt{1 - |\lambda|^2} \exp\mathopen{}\left( -\frac{|\mu|^2 + \Re[\lambda^* \mu^2]}{1 - |\lambda|^2} \right).
    \end{split}
    \end{equation}
    Expanding the exponential $e^{\frac{\lambda}{2}(a^\dagger)^2 + \mu a^\dagger} = 1 + \mu a^\dagger + \cdots$ where all other terms are at least degree two in $a^\dagger$, we have that $\braket{1 | \phi_1} = C\mu$. Denote $s = |\lambda|$ and $x = |\mu|^2$ so that
    \[
    \begin{split}
    \abs{\braket{1 | \phi_1}}^2 &= \sqrt{1 - s^2} \exp\mathopen{}\left( -\frac{x + \Re[\lambda^* \mu^2]}{1 - s^2} \right) \mu\\
    &\leq \sqrt{1 - s^2} \exp\mathopen{}\left( -\frac{x - sx}{1 - s^2} \right) x\\
    &= \sqrt{1 - s^2} \exp\mathopen{}\left( -\frac{x}{1 + s} \right) x,
    \end{split}
    \]
    where the inequality follows from choosing optimal phases such that $\Re[\lambda^* \mu^2] \geq -sx$. For fixed $s$, this bound is maximized at $x = 1 + s$, implying
    \[
    \abs{\braket{1 | \phi_1}}^2 \leq \frac{(1+s) \sqrt{1-s^2}}{e}.
    \]
    By inspection, this is maximized on $s \in [0, 1)$ at $s = \frac{1}{2}$. Equality is attained by setting $\lambda = -\frac{1}{2}$ and $\mu = \sqrt{\frac{3}{2}}$, so the bound is tight.

    Now we extend to $n$ bosonic modes. Let $\ket{\phi}$ be any pure Gaussian state and let
    \[
    \ket{g} = (\identity \otimes \bra{0^{n-1}}) \ket{\phi}
    \]
    be its projection onto the vacuum on modes $2, \ldots, n$. From \cref{eq:bargmann} we have that either $\ket{g} = \sqrt{p} \ket{\widetilde{g}}$ for some $p \in (0, 1]$ and single-mode $\ket{\widetilde{g}}$, corresponding to $Z = \lambda \oplus 0_{(n-1) \times (n-1)}$ and $\mu = (\mu_1, 0, \ldots, 0)$; otherwise, $\ket{g} = 0$ (at least one of the $n-1$ modes have no support with the vacuum). Consequently, we can uniformly write the overlap as
    \[
    \abs{\braket{1, 0^{n-1} | \phi}}^2 = p \abs{\braket{1 | \widetilde{g}}}^2 \leq \sup_{\ket{\phi_1}} \abs{\braket{1 | \phi_1}}^2 = \frac{3\sqrt{3}}{4e},
    \]
    and a maximizing state is found by tensoring the maximizing single-mode state with $n-1$ vacua.
\end{proof}

\begin{corollary}\label{cor:bgs_lower_bound}
    Let $\bosonicgaussianset$ be the class of all $n$-mode pure bosonic Gaussian states. The lower bound of \cref{thm:generic_testing_lower_bound} applies to $\bosonicgaussianset$.
\end{corollary}

\begin{proof}
    Combining \cref{lem:one-photon_overlap_gauss} and \cref{lem:interp_dist_AB} implies that there exist states $\ket{A} \in \bosonicgaussianset$, given by \cref{eq:A_state_bgs}, and
    \[
    \ket{B} = \frac{\ket{1, 0^{n-1}} - \sqrt{q} \ket{A}}{\sqrt{1 - q}}, \quad \text{where } q = \frac{3\sqrt{3}}{4e},
    \]
    that satisfy the premise of \cref{lem:interp_dist_AB} with $q \leq 1/2$. Hence \cref{thm:generic_testing_lower_bound} applies.
\end{proof}

\paragraph{Zero-mean Gaussian states.} The set of zero-mean bosonic Gaussian states $\zeromeanbosonicgaussianset \subset \bosonicgaussianset$ consists of all $\mu = 0$ such states. The previous argument can be minimally modified to handle this class.

\begin{lemma}\label{lem:two-photon_overlap}
    The two-boson state $\ket{X} = \ket{2, 0^{n-1}}$ satisfies
    \[
    \sup_{\ket{\phi} \in \zeromeanbosonicgaussianset} \abs{\braket{X | \phi}}^2 = \frac{1}{3\sqrt{3}} < \frac{1}{2},
    \]
    which is maximized by (for instance)
    \begin{equation}\label{eq:hard_A_zero-mean}
    \ket{A} = \frac{1}{3^{1/4}} \exp\mathopen{}\left( \frac{1}{2} \sqrt{\frac{2}{3}} a_1^\dagger a_1^\dagger \right) \ket{0^n}.
    \end{equation}
\end{lemma}

\begin{proof}
    On one mode, the Bargmann representation of zero-mean Gaussian states is (see \cref{eq:bargmann_norm} for the normalization factor)
    \[
    \ket{\phi_1} = (1 - |\lambda|^2)^{1/4} \exp\mathopen{}\left( \frac{\lambda}{2} (a^\dagger)^2 \right) \ket{0}.
    \]
    Again by expanding in degrees of $a^\dagger$, the only two-boson support has squared amplitude
    \[
    \abs{\braket{2 | \phi_1}}^2 = \sqrt{1 - |\lambda|^2} \frac{|\lambda|^2}{2}.
    \]
    This is maximized on $|\lambda|^2 \in (0, 1]$ at $|\lambda|^2 = \frac{2}{3}$. By the same argument as in the proof of \cref{lem:one-photon_overlap_gauss}, tensoring with $\ket{0^{n-1}}$ extends the supremum to $n$-mode states.
\end{proof}

\begin{corollary}
    Let $\zeromeanbosonicgaussianset$ be the class of all $n$-mode zero-mean pure bosonic Gaussian states. The lower bound of \cref{thm:generic_testing_lower_bound} applies to $\zeromeanbosonicgaussianset$.
\end{corollary}

\begin{proof}
    The argument follows \cref{cor:bgs_lower_bound}, except using \cref{lem:two-photon_overlap} and states $\ket{A} \in \zeromeanbosonicgaussianset$ as in \cref{eq:hard_A_zero-mean} and
    \[
    \ket{B} = \frac{\ket{2, 0^{n-1}} - \sqrt{q} \ket{A}}{\sqrt{1 - q}}, \quad \text{where } q = \frac{1}{3\sqrt{3}}. \qedhere
    \]
\end{proof}

\paragraph{Coherent states.} Recall that the set of coherent states $\coherentstateset_n$ on $n$ modes consists of all
\[
\ket{\alpha} = e^{-\|\alpha\|^2/2} \exp\mathopen{}\left( \sum_{j=1}^n \alpha_j a_j^\dagger \right) \ket{0^n}, \quad \alpha \in \C^n.
\]
We exhibit a very simple state which has squared overlap less than $1/2$ with any coherent state, thereby satisfying the conditions of \cref{lem:interp_dist_AB}.

\begin{lemma}\label{lem:one-photon_overlap}
    The one-boson state $\ket{X} = \ket{1, 0^{n-1}}$ satisfies
    \[
    \sup_{\alpha \in \C^n} \abs{\braket{\alpha | X}}^2 = e^{-1} < \frac{1}{2},
    \]
    which is attained by $\alpha = (1, 0, \ldots, 0)$.
\end{lemma}

\begin{proof}
    To compute the overlap we consider $\exp\mathopen{}\left( \sum_{j=1}^n \alpha_j^* a_j \right) \ket{X}$. Take the Taylor expansion of the exponential and recognize that the only term that survives is the term linear in $a_1$:
    \[
    \braket{\alpha | X} = \braket{0^n | e^{-\|\alpha\|^2/2} \alpha_1^* a_1 | 1, 0^{n-1}} = e^{-\|\alpha\|^2/2} \, \alpha_1^*.
    \]
    Let $x = \|\alpha\|^2$. The squared magnitude of the above is at most $e^{-\|\alpha\|^2} |\alpha_1|^2 \leq e^{-x} x$, which is maximized when $x = 1$. Indeed, this is achievable when $\alpha$ is a standard unit vector.
\end{proof}

\begin{corollary}
    Let $\coherentstateset_n$ be the class of all $n$-mode coherent states. The lower bound of \cref{thm:generic_testing_lower_bound} applies to $\coherentstateset_n$.
\end{corollary}

\begin{proof}
    The argument follows \cref{cor:bgs_lower_bound}, except using \cref{lem:one-photon_overlap} and states
    \[
    \ket{A} = \ket{(1, 0, \ldots, 0)} \in \coherentstateset_n, \quad \ket{B} = \frac{\ket{1, 0^{n-1}} - e^{-1/2} \ket{A}}{\sqrt{1 - e^{-1}}}. \qedhere
    \]
\end{proof}

\bibliographystyle{alphaurl}
\bibliography{refs}
\addcontentsline{toc}{section}{References}

\end{document}